\documentclass[reqno]{amsart}
\usepackage{amssymb,latexsym,upref,enumerate,fouridx}
\usepackage{mathrsfs,color}
\usepackage[scr=boondoxupr]{mathalpha} 
\usepackage{centernot}
\usepackage[colorlinks,linkcolor=blue,citecolor=blue]{hyperref}
\usepackage{cite} 

\allowdisplaybreaks

\numberwithin{equation}{section}
\numberwithin{footnote}{section}

\theoremstyle{plain}
\newtheorem{thm}{Theorem}[section]

\newtheorem{prop}[thm]{Proposition}

\theoremstyle{definition}
\newtheorem{Def}[thm]{Definition}

\theoremstyle{remark}

\newtheorem{ex}[thm]{Example}

\newcommand{\ep}{\epsilon}

\newcommand{\D}{D}

\input Tdef.def

\renewcommand{\emph}[1]{\textit{#1}}

\title{The extremely-tilted fluid regime near asymptotically Kasner big bang singularities}
\author[F. Beyer]{Florian Beyer}
\address{Dept of Mathematics and Statistics\\
730 Cumberland St\\
University of Otago, Dunedin 9016\\ New Zealand}
\email{florian.beyer@otago.ac.nz }

\begin{document}

\begin{abstract}
    In this paper, we solve the relativistic Euler equations with a linear barotropic equation of state on a large class of background spacetimes with Kasner big bang asymptotics. Building on previous work in the asymptotically non-tilted regime \cite{BeyerLeLoch:2017,BeyerOliynyk:2020, beyerStabilityFLRWSolutions2023}, which applies when the speed of sound of the fluid is large in comparison to the Kasner exponents, we now consider the asymptotically extremely-tilted regime when the speed of sound is small. We solve the Cauchy problem for the Euler equations towards the big bang singularity and prove, without any symmetry assumptions or smallness conditions on the Cauchy data, that the solutions exist globally in time provided the mean curvature of the initial hypersurface is sufficiently large. Finally, we prove that the solutions exhibit the asymptotics expected from standard heuristic arguments in this regime; in particular, fluid particles are driven towards the speed of light in the direction of the largest Kasner exponent as the big bang singularity is approached.
\end{abstract}

\maketitle


\section{Introduction}
\label{sec:introduction}

\subsection{Background}
Cosmological Einstein--matter solutions have been studied for decades, and many outstanding questions remain. In this paper, the \emph{cosmological setting} is that of globally hyperbolic spacetimes with compact Cauchy surfaces \cite{bartnik1988}, and our focus is on the time asymptotics of fluid matter\footnote{Related questions have been investigated for other matter types, like electromagnetic fields and Vlasov matter; an incomplete list of relevant works is \cite{an2025a,fajman2025a,liu2022}.}. In particular, we are interested in the backward-time asymptotics near ``the beginning'', i.e.\ the \emph{big bang singularity} in the finite past\footnote{Without loss of generality we always choose our time orientation so that the contracting time direction corresponds to the past.}. As a basic point of comparison, one of the simplest model classes is the spatially homogeneous and isotropic \emph{Friedmann--Lema\^itre--Robertson--Walker} (FLRW) family: the big bang corresponds to the cosmic scale factor approaching zero, and in the spatially flat case elementary arguments show that such a singularity is unavoidable provided the matter satisfies the strong energy condition.
More generally, big bang singularity formation is expected to be stable for broad classes of solutions of the Einstein--matter equations beyond FLRW, in agreement with Hawking's singularity theorem \cite{hawkingLargeScaleStructure1973} (see also \cite[Chapter~14, Theorem~55A]{oneill1983a}). Influential works \cite{belinskii1970,lifshitz1963} provide heuristic insight into the dynamics near a ``general singularity''; their conclusions are often referred to as the \emph{BKL conjecture}, which suggests that cosmological big bang singularities are typically spacelike, local, and oscillatory. In particular, most kinds of matter are expected to become asymptotically ``negligible'' in the sense that the ratio of the matter energy density to the square of the mean curvature of the cosmological time function approaches zero. This leads to one of the central mottos of the field: \emph{matter does not matter}.

There is evidence for the BKL conjecture, for example \cite{BeguinDutilleul:2023,beguin:2010,liebscher_et_al:2011,ringstrom2001a} for rigorous work in the spatially homogeneous setting and \cite{anderssonAsymptoticSilenceGeneric2005, curtis2005, garfinkle2002,garfinkle2002a,garfinkle2004,garfinkle2007,berger2001,marshall2025a} for numerical work beyond spatial homogeneity. At the same time, more recent work on spikes \cite{berger1993,coley2015, lim2008, lim2009, rendall2001} and weak null singularities \cite{dafermos2017,luk2017} indicates that the BKL picture is incomplete. Nevertheless, it remains a useful heuristic guide for many investigations in mathematical cosmology, and we adopt this perspective here.

One exception to the \emph{matter does not matter} hypothesis is  a minimally coupled scalar field (or, more generally, a \emph{stiff} fluid). In fact, it is common to assume the presence of such matter (in addition to other matter fields) given significant evidence that the oscillatory dynamics is generically suppressed and that the system enters a \emph{convergent} so-called \emph{asymptotically velocity term dominated} (AVTD) regime \cite{eardley1972,isenberg1990} near the big bang. In recent years, remarkable progress has been made in rigorously establishing the past stability of FLRW solutions to the Einstein--scalar field equations and their AVTD big bang singularities under generic perturbations without symmetries. The first such FLRW big bang stability result was obtained in the seminal articles \cite{rodnianski2014, rodnianski2018}. These results were subsequently extended in \cite{fournodavlos2020b}, which establishes AVTD big bang stability for nonlinear perturbations of the Kasner family (see Example~\ref{ex:KSF}). A very detailed analysis of AVTD big bang asymptotics for a large class of Einstein--scalar field models was later established in the series of works \cite{ringstrom2021,ringstrom2021a,ringström2022,ringström2022a,groeniger2023,franco-grisales2026,franco-grisales2024,franco-grisales2026} with a particular emphasis on  \emph{reference-independence} -- a notion that is also important for this paper here and is discussed in more detail below.  
A further important development is the introduction of spatial localisation techniques and their applications to the big bang stability problem in \cite{BeyerOliynyk:2021,BeyerOliynykZheng:2025,zheng2026,franco-grisales2026,athanasiou2024}.
For completeness, we also mention that even in the absence of a scalar field, AVTD behaviour may still occur, especially for classes of vacuum spacetimes with symmetries \cite{ames2021, ames2021a,chruscielStrongCosmicCensorship1990,isenberg1990,fournodavlos2020b,ringstrom2009a, choquet-bruhat2004,ames2019,ames2013a,andersson2001,BeyerLeLoch:2017,choquet-bruhat2006,klinger2015,clausen2007,damour2002,fournodavlos2020,heinzle2012,isenberg1999,isenberg2002,kichenassamy1998,stahl2002,athanasiou2024,franco-grisales2024, dong2026}.

\subsection{Informal discussion of results in this paper}
Generally motivated by the BKL heuristics, the AVTD property and the \emph{matter does not matter} motto, we solve the Euler equations near the big bang singularity in a large class of \emph{prescribed} spacetimes that satisfy the AVTD property in this paper. To this end, we define the class of \emph{spacetimes with Kasner big bang asymptotics}  in Definition~\ref{def:asympKasner}. Roughly speaking, a spacetime with Kasner big bang asymptotics according to Definition~\ref{def:asympKasner} approaches a prescribed\footnote{The spatially flat FLRW spacetimes are particular members of this family.} \emph{Kasner spacetime} (cf.\ Example~\ref{ex:KSF}) at each spatial point with prescribed convergence rates in agreement with the BKL and AVTD heuristics. Such a spacetime is itself not required to satisfy the field equations, so this allows us to perform our analysis in a wide range of different scenarios.
Regarding the fluids, we choose a linear barotropic equation of state in which the fluid pressure is directly proportional to the fluid density. The proportionality factor is a free parameter of our models and equals the square of the speed of sound of the fluid $c_s^2$. Given this setup, a very interesting criticality regarding the parameter $c_s^2$ has been observed  \cite{sandin2010,BeyerLeLoch:2017,beyerStabilityFLRWSolutions2023,BeyerOliynyk:2020} which we discuss heuristically in Section~\ref{sec:convvariables.N}.
While the fluid density always blows up at the big bang singularity, the fluid tilt vector, representing the peculiar motion of galaxies relative to the cosmic flow in standard cosmology, either decays to zero if the speed of sound is large (in a particular sense), or approaches a unit vector if the speed of sound is small. The first regime is called \emph{asymptotically non-tilted} and represents a fluid whose tilt vector is uniformly small near the big bang. The latter regime is often called \emph{asymptotically extremely-tilted} and is particularly intriguing because it means that the fluid particles are driven towards the speed of light by gravity near the big bang singularity.  It is also noteworthy that it is the radiation fluid case given by $c_s^2=1/N$ in $N$ spatial dimensions that is \emph{critical} in the special case of FLRW background spacetimes. 
While the asymptotically non-tilted regime at the big bang is quite well understood rigorously \cite{BeyerLeLoch:2017,beyerStabilityFLRWSolutions2023,BeyerOliynyk:2020}, the asymptotically extremely-tilted case is open and is technically more challenging. Very little is known about the exact critical case beyond \cite{BeyerLeLoch:2017}.

There are by now more rigorous results regarding the behaviour of fluids  in the \emph{expanding time direction} of cosmological models \cite{oliynyk2024, beyer2023a, fournodavlos2024a, marshall2023, fournodavlos2025, oliynyk2021,oliynyk2016, hadžić2015, lubbe2013a, rodnianski2013, speck2012, speck2013} where the critical dependence on the speed of sound is reversed. 
Consistently with the heuristics, the asymptotically non-tilted regime generally occurs when the speed of sound is sufficiently large, while  the asymptotically extremely-tilted regime occurs generally for small sound speeds. 
Intriguingly, the analyses of \cite{fournodavlos2024a, fournodavlos2025}  suggest the existence of further effects that are not predicted by the heuristics and which manifest themselves in additional restrictions for the speed of sound parameter. The precise origin of these effects and
whether these are merely technical rather than physical is an open question.

In this paper we rigorously study the outstanding problem of asymptotically extremely-tilted fluids near the big bang. 
A first main question of interest is under what conditions the fluid solutions are free of shocks and exist, within a prescribed regularity class, for all times towards the big bang singularity in this regime. If this is the case, the next question is whether we can rigorously confirm the heuristic predictions in detail. A particular concern is the effect of anisotropies in the background geometries near the big bang -- an issue that is not relevant for the dynamics in the expanding direction where anisotropies generally die out. 
Our final main interest is whether our rigorous analysis uncovers additional non-heuristic effects near the big bang in analogy to the effects identified in \cite{fournodavlos2024a, fournodavlos2025} in the expanding case mentioned above. In this paper, we study the fluids on a large class of \emph{prescribed} spacetimes (in contrast to the works in \cite{fournodavlos2024a, fournodavlos2025}), and this provides the opportunity to attribute such effects to particular properties of the geometry near the big bang.  

Both of our two main results of this paper, Theorems~\ref{thm:Euler1} and \ref{thm:Euler2}, can be summarised informally as follows; the particular details which distinguish these two results are discussed later. In preparation for this informal summary we list only those  quantities from Definition~\ref{def:asympKasner} that we require here: the largest Kasner exponent $P$ (which may be a smooth function of space), as well as the projection map $\hat h$ onto the corresponding eigenspace of the rescaled Weingarten map and its orthogonal complement $\check h$ (both assumed to be smooth as well).
\begin{thm}[Informal version of Theorems~\ref{thm:Euler1} and \ref{thm:Euler2}]
    Pick a background spacetime with Kasner big bang asymptotics in the sense of Definition~\ref{def:asympKasner} (which satisfies additional conditions below). There exists a constant $c_*^2\in [0,P)$ so that, given any speed of sound parameter $c_s^2$ in the range 
    \begin{equation}
        \label{eq:informalnonheuristic}
        c_*^2<c_s^2<P<1,
    \end{equation} 
    all solutions of the Euler equations launched from spacelike hypersurfaces whose mean curvature $H$ is sufficiently large
    exist globally towards the big bang singularity and are asymptotically extremely-tilted. Given the time function $t=1/H$, the densities associated with the fluid solutions are $\rho=O\bigl(t^{-(1-P)(1+c_s^2)/(1-c_s^2)}\bigr)$, the Lorentz factors $\Gamma=O\bigl(t^{-(P-c_s^2)/(1-c_s^2)}\bigr)$ and the tilt vectors $\nu$ approach unit vectors with $|\hat h \nu|\rightarrow 1$ and $\check h\nu\rightarrow 0$ in the limit $t\searrow 0$ at the big bang. 
\end{thm}

Indeed, the tilt vector $\nu$ therefore approaches a unit vector in agreement with the heuristics of the asymptotically extremely-tilted regime. The Lorentz factor $\Gamma$, which describes the relative speed between fluid particles and standard cosmic observers (see Section~\ref{sec:Frauendiener.N} for detailed discussions), approaches infinity. Moreover, our theorems describe the effect of anisotropies: the tilt vector converges in the eigenspace of the largest Kasner exponent $P$ and its orthogonal complement decays. 
The theorems therefore reveal a remarkable coupling between the fluid tilt and the eigenspaces of the Kasner exponents. Consequently, if the eigenspace of $P$ is one-dimensional at each point in $\Sigma$, then the asymptotic direction of the tilt vector is uniquely determined at each point near the big bang. If the eigenspace of $P$ has dimension larger than one at each point -- for instance in the asymptotically isotropic case where all Kasner exponents coincide -- then the theorems yield less precise information about the asymptotic direction of the tilt. The way fluids couple to anisotropies in this regime near the singularity was observed previously, see for example \cite{sandin2010}; similar phenomena are also exhibited by Vlasov matter near the big bang \cite{fajman2025a}.

A particular consequence is that $\rho/H^2$ approaches zero in the limit $t=0$ for our fluid solutions. Hence, fluids in the \emph{asymptotically extremely-tilted} regime \eqref{eq:informalnonheuristic} are consistent with the \emph{matter does not matter} motto.

Another distinctive feature of our results is that they impose no size restrictions on the initial data, in particular no restriction on the size of their spatial derivatives. The fluid data is especially not restricted to a near-spatially homogeneous setting; in fact, there is no requirement that the fluid data is close to the data of any reference fluid. This motivates our terminology of \emph{reference-independent} fluid solutions. To control error terms that arise in our approach we must adjust the initial time $T_0$, i.e., the time when the fluid solutions are launched towards the past, which leads to the requirement that $T_0$ be small -- or, equivalently, that the mean curvature of the initial hypersurface is large.

One of the particular concerns of this paper is the role of the constant $c_*^2$ in \eqref{eq:informalnonheuristic}. Purely heuristically, this constant should always simply be zero. So any non-zero requirement can be interpreted as a non-heuristic effect. In this paper we think of $c_*^2$ as the big bang version of the analogous upper bounds for the speed of sound in the expanding case identified in \cite{fournodavlos2024a,fournodavlos2025} mentioned above.
The estimates we obtain for $c_*^2$ in Theorems~\ref{thm:Euler1} and \ref{thm:Euler2} depend on the precise assumptions we impose on the background spacetimes and  differ significantly between Theorem~\ref{thm:Euler1} and \ref{thm:Euler2}. Roughly speaking, Theorem~\ref{thm:Euler1} assumes less about the pointwise Kasner limit of the background spacetime than Theorem~\ref{thm:Euler2}, but more about the decay rates of the background spacetime towards the Kasner limit.
Interestingly, Theorem~\ref{thm:Euler1} yields worse bounds for $c_*^2$ than Theorem~\ref{thm:Euler2} where we can attribute the non-heuristic constant $c_*^2$ to a single term in the Euler equations. 
Given the assumption in Theorem~\ref{thm:Euler2} that the eigenspace of the largest eigenvalue $P$ is one-dimensional, the tilt vector $\nu$ converges to the unique unit vector in this one-dimensional space (unique up to orientation). 
We explain in detail in Section~\ref{sec:thm2} that the strong result in Theorem~\ref{thm:Euler2} is then in particular a consequence of an additional \emph{large eigenvalue separation condition} (hence the name, ``strong anisotropic case'') which then guarantees rapid decay of the components of $\nu$ in the orthogonal complement.

So what is the technical reason why we can in general \emph{not} set\footnote{Even in cases where our theorems would allow the choice $c_*^2=0$, for example, the exact Kasner family defined in Example~\ref{ex:KSF}, our approach does not cover the dust case $c_s^2=0$: dust requires a fundamentally different choice of variables and is generally incompatible with our framework. Details are given in Section~\ref{sec:Frauendiener.N}.} $c_*^2=0$ in \eqref{eq:informalnonheuristic}?
This is mainly due to the rather involved construction of positive definite PDE energies which do not degenerate near the big bang in the asymptotically extremely-tilted regime (a problem which is less severe in the asymptotically non-tilted regime). Indeed, it had been observed before \cite{BeyerLeLoch:2017} that the most direct candidates for such energies degenerate near $t=0$ when the tilt becomes extreme. Our two Theorems~\ref{thm:Euler1} and \ref{thm:Euler2} construct non-degenerate PDE energies in new distinct ways, each of which is tailored to the particular scenario of interest. 

\subsection{Organisation of the paper}
Our paper is organised as follows. In Section~\ref{presec:background}, we introduce our notation and conventions, especially index conventions, and start off  in Section~\ref{sec:background.N} with a discussion of our geometric setup and our class of background geometries: \emph{spacetimes with Kasner big bang asymptotics} in Definition~\ref{def:asympKasner}. We list immediate consequences  and, in Example~\ref{ex:KSF}, introduce the particularly important family of exact solutions of the Einstein--scalar field equations: \emph{Kasner--scalar field spacetimes}. We provide some further auxiliary material  in  Section~\ref{sec:analysisbg.N}. In Section~\ref{sec:Frauendiener.N}  we introduce the relativistic Euler equations and discuss the fundamental fluid heuristics near the big bang singularity in Section~\ref{sec:convvariables.N}.
Because the Euler equations given in Section~\ref{sec:mainformulations.N} are   amenable neither to the proof of Theorem~\ref{thm:Euler1} (Section~\ref{sec:proof1}) nor the proof of Theorem~\ref{thm:Euler2} (Section~\ref{sec:proof2}), we provide a detailed derivation of alternative Euler PDE systems in Sections~\ref{sec:variablescont.N}~and~\ref{sec:FuchsianPDE.N}. We discuss Theorem~\ref{thm:Euler1} and its consequences in Section~\ref{sec:thm1}.
Finally, Section~\ref{sec:thm2} is devoted to Theorem~\ref{thm:Euler2}. 

\section{Notation, conventions and background}
\label{presec:background}

\subsection{Background spacetimes with Kasner big bang asymptotics}
\label{sec:background.N}

The goal of this paper is to solve the relativistic Euler equations on prescribed Lorentzian manifolds $(M,g_{\mu\nu})$ with ``reasonable'' big bang asymptotics. Before we fix our class of spacetimes in Definition~\ref{def:asympKasner}, we assemble some necessary geometric background material.

Let $\Sigma$ be a smooth $N$-dimensional compact parallelisable manifold (without boundary) with $N\ge 2$ and $I=(0,T_0]$ for $T_0>0$. Consider the differentiable manifold $M=I\times\Sigma$ equipped with a smooth Lorentzian metric $g_{\mu\nu}$.
We use Greek lower case letters $\mu,\nu,\ldots$ to denote abstract tensor indices for tensor fields defined with respect to $TM$ and $T^*M$. We lower and raise indices by $g_{\mu\nu}$ and its inverse $g^{\mu\nu}$, respectively. The covariant derivative associated with $g_{\mu\nu}$ is denoted by $\nabla_\mu$. 
We also suppose that the map 
\begin{equation}
    \label{eq:tmap}
    t:M=(0,T_0]\times\Sigma\rightarrow\Rbb,\quad t(\tau,x)=\tau,
\end{equation} 
has a smooth future pointing timelike gradient $\nabla_\mu t$ and that $\Sigma_\tau=\{\tau\}\times\Sigma\cong\Sigma$ is a spacelike Cauchy surface for each\footnote{In what follows, we abuse notation by using the same symbol $t$ for both the map in \eqref{eq:tmap} itself and its values in $(0,T_0]$.} $t\in (0,T_0]$.  
We define
\begin{equation}
    \label{eq:defnlapse}
    n_\mu=-\alpha\nabla_\mu t,\quad n_\mu n^\mu=-1,
\end{equation} 
as the future pointing hypersurface orthogonal timelike vector unit normal to the collection of hypersurfaces $\Sigma_t$ given by all $t\in (0,T_0]$. The smooth positive function $\alpha: M\rightarrow\Rbb$ is called the \emph{lapse}. The \emph{spatial metric} induced on each $\Sigma_t$ by $g_{\mu\nu}$ is
\begin{equation}
    \label{eq:defh}
    h_{\mu\nu}=g_{\mu\nu}+n_\nu n_\nu.
\end{equation}
Any tensor field on $M$ that is annihilated by contractions with $n_\mu$ or $n^\mu$ on each of its tensor slots is called \emph{purely spatial}. Scalar fields are always considered as purely spatial.

The orthogonal decomposition
\begin{equation}
    \label{eq:defdotnk}
    \nabla_\mu n_\nu
    =-\dot n_\nu n_\mu+k_{\mu\nu}\quad\text{with}\quad \dot n_\nu n^\nu=0,\quad k_{\mu\nu}n^\nu=0,
\end{equation}
uniquely determines the purely spatial fields $\dot n_\mu$ (\emph{acceleration}) and $k_{\mu\nu}$ (\emph{second fundamental form}). Because $n_\mu$ is hypersurface orthogonal, we have $k_{[\mu\nu]}=0$. The \emph{expansion scalar} (or \emph{mean curvature}) of the hypersurfaces $\Sigma_t$ is 
\begin{equation}
    \label{eq:defdotnk.2.N}
     H=\frac 1N {k_\mu}^\mu,
\end{equation}
which we assume to be positive on $M$.
One can also show that \eqref{eq:defnlapse} and \eqref{eq:defdotnk} imply that
\begin{equation}
    \label{eq:dotnalpharel}
    \dot n_\mu=\alpha^{-1}h_\mu{}^\nu\nabla_\nu \alpha.
\end{equation}

Consider next the $1$-parameter family of canonical embeddings
\begin{equation}
    \label{eq:timedepdiffeo}
    \Psi_t:\Sigma\rightarrow \Sigma_t\subset M,\quad t\in (0,T_0].
\end{equation}
Given a smooth purely spatial tensor field $T_{\mu\ldots\nu}$ of rank $(0,s)$ on $M$, its family of pull-backs to $\Sigma$ along $\Psi_t$ for $t\in (0,T_0]$ defines a smooth \emph{time-dependent tensor field} on $\Sigma$, that is, a map $(0,T_0]\times\Sigma\rightarrow T^{*}_{(0,s)}\Sigma$ consistent with the $\Sigma$-cotangent bundle projection, denoted by $T_{\Omega\ldots\Sigma}$. To keep it brief, we sometimes refer to a smooth time-dependent tensor field $T_{\Omega\ldots\Sigma}$ as a \emph{tensor field on $(0,T_0]\times\Sigma$}. In contrast to this we say that a \emph{tensor field on $\Sigma$ is time-independent}, or, simply a \emph{tensor field on $\Sigma$}, if it is a smooth map $\Sigma\rightarrow T^{*}_{(0,s)}\Sigma$ consistent with the $\Sigma$-cotangent bundle projection. In general, we use indices $\Lambda,\Omega,\ldots$ for abstract indices of time-dependent or time-independent tensor fields on $\Sigma$. All of this applies natural to smooth tensor fields on $M$ of arbitrary rank $(r,s)$ as well.

A general, not necessarily purely spatial smooth tensor field on $M$ can also be fully represented via pull-back to $\Sigma$. To this end we  first perform the orthogonal decomposition with respect to $n^\mu$ on $M$. Each uniquely determined resulting field is purely spatial, and can then be pulled back to $\Sigma$ as above. Our convention is that the pull-back of any piece obtained by the orthogonal decomposition via a contraction with $-n_\mu$ (or $n^\mu$, respectively) is labelled by a lower (or upper, respectively) $0$-index: for example, a vector field $V^\mu$ on $M$ yields the time-dependent scalar $V^0=V^\mu (-n_\mu)$ and the time-dependent vector field $V^\Lambda=h_\mu{}^\Lambda V^\mu$ on $\Sigma$ in this way, while a covector field $X_\mu$ on $M$ yields the time-dependent scalar $X_0=X_\mu n^\mu$ and the time-dependent covector field $X_\Lambda=h_\Lambda{}^\mu X_\mu$ on $\Sigma$. In this way, pull-backs of arbitrary tensor fields on $M$ are represented by fields on $\Sigma$ that have either a $0$ index or a capital Greek index. 

As a particular consequence, the tensor field $h_{\Lambda\Sigma}$, i.e., the pull-back of the spatial metric $h_{\mu\nu}$ defined in \eqref{eq:defh} obtained in this way, is a time-dependent Riemannian metric on $\Sigma$, and we denote its time-dependent covariant derivative by $\D_\Lambda$.

Given a smooth time-dependent function $f$ and purely spatial covector field $X_\mu$ on $M$, respectively, we also define the \emph{time derivative operator} $\D$ as
\begin{equation}
    \label{eq:defD}
    \D f=n^\sigma  \nabla_\sigma f,\quad \D X_\mu = n^\sigma h_\mu{}^\nu \nabla_\sigma X_\nu.
\end{equation} 
The pull-backs of these purely spatial fields are referred to as $Df$ and $DX_\Lambda$.
We remark that operator $D$ is related to the Lie derivative $\mathcal L_n$ as follows
\begin{equation}
    \label{eq:D2Lie}
    \D f=\mathcal L_n f,\quad \D X_\mu=\mathcal L_n X_\mu-k_\mu{}^\nu X_\nu.
\end{equation}
Given any $x\in\Sigma$, the vector field $t^\mu$ tangent to the curve $t\mapsto \Psi_t(x)$ in $M$ satisfies
\begin{equation}
    t^\mu\nabla_\mu t =1.
\end{equation} 
Together with \eqref{eq:defnlapse} this implies the existence of a purely spatial vector field $\beta^\mu$ on $M$, called the \emph{shift vector}, such that
\begin{equation}
    \label{eq:tvf}
    t^\mu=\alpha n^\mu+\beta^\mu.
\end{equation} 
In this whole paper we assume that the shift vector in \eqref{eq:tvf} vanishes
\begin{equation}
    \label{eq:vanishshift}
    t^\mu=\alpha n^\mu,\quad \beta^\mu=0.
\end{equation}
It is a standard result, see for instance \cite{lee2012}, that given any smooth function $f$ or covector field $X_\mu$ on $M$, the pull-backs of $\mathcal L_t f$ and $\mathcal L_t X_\mu$ to $\Sigma$ coincide, for any $x\in\Sigma$ and $t\in (0,T_0]$, with ordinary time derivatives
\begin{equation}
    \label{eq:timeder}
    \partial_t f=\lim_{h\rightarrow 0}\frac{f(t+h,x)-f(t,x)}{h},\quad 
    \partial_t X_\Lambda=\lim_{h\rightarrow 0}\frac{X_\Lambda(t+h,x)-X_\Lambda(t,x)}{h},\quad t\in (0,T_0],\quad x\in M.
\end{equation}

We are now ready to state the first central definition of this paper.

\begin{Def}[Spacetimes with Kasner big bang asymptotics]
    \label{def:asympKasner}
    We call a Lorentzian $(M=(0,T_0]\times\Sigma, g_{\mu\nu})$ a \textbf{spacetime with Kasner big bang asymptotics} provided there are smooth positive functions $\bar p$, $q$, $\bar q$ on $\Sigma$ and an $h_{\Lambda\Omega}$-orthonormal spatial frame\footnote{Indices $i,j,\ldots=1,\ldots,N$ for tensor fields on $\Sigma$ will represent tensor field components of time-dependent or time-independent fields on $\Sigma$ with respect to this orthonormal frame $e_i{}^\Omega$ throughout the whole paper.} $e_i{}^\Omega$ ($i=1,\ldots, N$) on $(0,T_0]\times \Sigma$ such that:
    \begin{enumerate}
        \item There is a smooth (not necessarily orthonormal) time-independent spatial frame $b_{\bar i}{}^\Omega$ ($\bar i=1,\ldots, N$) on $\Sigma$ and a map $\tau\in C_b^\infty((0,T_0]\times\Sigma, \mathrm{GL}(N,\mathbb{R}))$ written as $\tau^{\bar i}_i$ such that
        \begin{equation}
            \label{eq:framebound}
            e_i{}^\Lambda= t^{-1+\bar p} \tau^{\bar i}_i b_{\bar i}{}^\Lambda.
        \end{equation}
        \item There is a function $\alpha_*\in C_b^\infty((0,T_0]\times\Sigma, \Rbb^{>0})$ such that
        \begin{equation}
            \label{eq:alphaasympt}
            \alpha=1+t^q\alpha_*.
        \end{equation}
        \item \label{cond:spectral} Let $K_\Lambda{}^\Omega$ be the \emph{rescaled Weingarten map}
        \begin{equation}
            \label{eq:defK}
            K_\Lambda{}^\Omega=\frac{ k_\Lambda{}^\Omega}{N H}
        \end{equation}
        on $(0,T_0]\times\Sigma$.
        There exist a smooth tensor field $\Ktt_\Lambda{}^\Omega$ on $(0,T_0]\times\Sigma$ with time-independent orthonormal frame components\footnote{Observe that 
        the matrix of components $\Ktt^{ki}=\Ktt_i{}^j h^{ik}$ is necessarily symmetric and therefore diagonalisable at each point on $\Sigma$. In general, this is not enough to guarantee that the eigenvalues or eigenvectors are smooth though.} $\Ktt_i{}^j$, and,  a tensor field $K_{*,\Lambda}{}^\Omega$ on $(0,T_0]\times\Sigma$ whose orthonormal frame components $K_{*,i}{}^j$ are members of $C_b^\infty((0,T_0]\times\Sigma)$ such that
        \begin{equation}
            \label{eq:Kttdef}
            K_i{}^j=\Ktt_i{}^j+t^q K_{*,i}{}^j
        \end{equation}
        on $(0,T_0]\times\Sigma$.
        Furthermore, there exist a smooth tensor field $\hat h_\Lambda{}^\Omega$ with time-independent orthonormal frame components $\hat h_i{}^j$ and a smooth function $P$ on $\Sigma$ such that 
        \begin{enumerate}
            \item $\hat h_i{}^j\hat h_j{}^k=\hat h_i{}^k$,
            \item the map defined by the orthonormal frame components
            \begin{equation}
                \label{eq:Kttspecdecomp}
                \check L_i{}^j=P h_i{}^j-\Ktt_i{}^j
            \end{equation}
            is smooth on $\Sigma$, positive semi-definite, and its kernel equals the range of $\hat h_i{}^j$.
        \end{enumerate}
        The eigenvalues of $\Ktt_i{}^j$ are also often called \emph{Kasner exponents}.

        It follows that $P$ is the Kasner exponent and is smooth by assumption. The map $\hat h_i{}^j$ is the  \emph{projector into the eigensubspace corresponding to the largest eigenvalue} $P$ and is also smooth by assumption. Moreover, when we define
        \begin{equation}
            \check h_i{}^j=h_i{}^j-\hat h_i{}^j,
        \end{equation}
        then we have 
        \begin{equation}
            \hat h_i{}^j \Ktt_j{}^k \hat h_k{}^l=P \hat h_i{}^l,\quad \hat h_i{}^j \Ktt_j{}^k \check h_k{}^l=\check h_i{}^j \Ktt_j{}^k \hat h_k{}^l=0,\quad \check h_i{}^j \Ktt_j{}^k \check h_k{}^l=P \check h_i{}^l-\check L_i{}^l,
        \end{equation}
        at each point of $\Sigma$.
        \item The mean curvature $H$ defined in \eqref{eq:defdotnk.2.N} satisfies
        \begin{equation}
            \label{eq:CMCgauge}
            H=\frac{1}{\alpha N t}.
        \end{equation}
        \item There is a map $\Omega\in C_b^\infty((0,T_0]\times\Sigma, M_{N\times N}(\mathbb{R}))$ written as $\Omega_i{}^j$ such that
        $\Omega^{ki}=\Omega_i{}^j h^{ik}$ is antisymmetric and
        \begin{equation}
            \label{frametransp.1.N}
            D e_i{}^\Lambda=\alpha^{-1}t^{-1+q}\Omega_i{}^j e_j{}^\Lambda
        \end{equation}
        on $(0,T_0]\times\Sigma$.
        Recall \eqref{eq:defD} for the definition of the time derivative $\D$.
        \item There is a map $c\in C_b^\infty((0,T_0]\times\Sigma, M_{N\times N\times N}(\mathbb{R}))$ written as $c_{i}{}^j{}_k$ such that
        \begin{equation}
            \label{eq:framecommutators}
            \mathcal L_{e_i} e_k{}^\Lambda=[e_i,e_k]^\Lambda=t^{-1+\bar q} c_{i}{}^j{}_k e_j{}^\Lambda
        \end{equation}
        on $(0,T_0]\times\Sigma$.
    \end{enumerate}
\end{Def}

A particularly important example of spacetimes with Kasner big bang asymptotics is given in Example~\ref{ex:KSF} below.
We first discuss some basic consequences and properties of Kasner big bang asymptotics. First, given \eqref{eq:CMCgauge}, it follows that $t=0$ is a crushing singularity \cite{eardley1979} and past timelike geodesic incompleteness follows from Hawking's singularity theorem \cite[Chapter~14, Theorem~55A]{oneill1983a}. However, curvature blow up at $t=0$ is not implied by the assumptions and, in fact, $t=0$ may therefore not always represent a \textit{big bang singularity} in the sense of $C^2$-inextendibility. Observe also that we do not require that a spacetime with Kasner big bang asymptotics satisfies the Einstein equations. As a consequence, while the eigenvalues of $\Ktt_i{}^j$, one of which is $P$, always sum to $1$ at each point of $\Sigma$ by virtue of \eqref{eq:defK} and \eqref{eq:Kttdef}\footnote{For the important Example~\ref{ex:KSF}, this corresponds to the first equation in \eqref{eq:Kasnerrel}.}, we do in general not require an analogue of the second equation in \eqref{eq:Kasnerrel} which in Example~\ref{ex:KSF} is implied by the Hamiltonian constraint of the Einstein--scalar field equations. 
This paper is concerned with solutions of the Euler equations on quite a large class of {prescribed} spacetimes which, most importantly, contains those spacetimes that we believe to be physically relevant near the big bang. 

Next, as a consequence of (\ref{cond:spectral}) in Definition~\ref{def:asympKasner} and the orthogonal decomposition of the (co)tangent space of $\Sigma$ at each point provided by $\hat h_i{}^j$ and $\check h_i{}^j$, a covector $u_i$ defined at any point of $\Sigma$ can be uniquely decomposed as
\begin{equation}
    \label{eq:spectral_covector_decomp}
    u_i=\hat u_i+\check u_i,\quad \hat h_i{}^j\hat u_j=\hat u_i,\quad \check h_i{}^j\hat u_j=0,\quad \hat h_i{}^j\check u_j=0,\quad \check h_i{}^j\check u_j=\check u_i.
\end{equation}
It follows that $\hat u_i$ is an eigenvector of $\Ktt_i{}^j$ with eigenvalue $P$ while $\check u_i$ is in the orthogonal complement (and therefore not necessarily an eigenvector). The eigenvalues of $\check L_i{}^j$ are given by the differences between the largest eigenvalue $P$ and the remaining, by definition smaller, eigenvalues. The case in which $\check h_i{}^j$ and therefore $\check L_i{}^j$ have an $N$-dimensional kernel at each point of $\Sigma$ is interpreted as the \emph{asymptotically isotropic} case near the big bang characterised by the tensor $\Ktt_i{}^j$ being pure trace. In fact, since the eigenvalues of $\Ktt_i{}^j$, one of which is $P$, always sum to $1$ at each point of $\Sigma$ by virtue of \eqref{eq:defK} and \eqref{eq:Kttdef}, it follows that $P=1/N$ in the asymptotically isotropic case. This case is the main case of interest for the analysis of the asymptotics of fluids in the expanding time direction \cite{fournodavlos2021, fournodavlos2024a, hadžić2015, lubbe2013a, marshall2023, marshall2025, oliynyk2016, rendall2004a, rodnianski2013, speck2012, speck2013,fajman2025}; it is less relevant in the contracting direction near the big bang, but see \cite{anguige1999, goode1985, tod1999, tod1999a,lidsey2006,garfinkle2008}. 

A particularly noteworthy restriction imposed by Definition~\ref{def:asympKasner} arises from the smoothness assumption on $\hat h_{i}{}^j$ and $\check L_{i}{}^j$. Consequently, \emph{the dimension of the $P$-eigenspace of $\Ktt_i{}^j$ must be constant in space}. This assumption would be violated if we considered, for example, the case where another eigenvalue $\bar P$ equals the largest one $P$ at some points in $\Sigma$.  We remark that under certain conditions a localised version of the results of this paper, where we would stay away from such exceptional points,  would still be valid. 

Next we observe that \eqref{eq:CMCgauge} and \eqref{eq:alphaasympt} can be interpreted as stating that $t$ is an \emph{asymptotic CMC time function}. While we stick to this choice for this paper here, it turns out that our results straightforwardly extend to a class of lapse functions $\alpha$ that diverge at $t=0$ (violating \eqref{eq:alphaasympt})\footnote{Time functions with mildly diverging lapse functions occur in the results of \cite{BeyerOliynyk:2021, BeyerOliynykZheng:2025, beyerStabilityFLRWSolutions2023}.}. 

We next note that 
\begin{equation*}
        \mathcal L_n e_i{}^\mu=n^\rho\nabla_\rho e_i{}^\mu-e_i{}^\rho\nabla_\rho n{}^\mu
        =D e_i{}^\mu
        +\dot n_i n^\mu
        -e_i{}^\rho k_\rho{}^\mu,
\end{equation*}
using  \eqref{eq:defh}, \eqref{eq:defdotnk} and \eqref{eq:defD}. Pulling this back to $\Sigma$ along \eqref{eq:timedepdiffeo} for each $t\in (0,T_0]$ and using \eqref{eq:vanishshift}, \eqref{frametransp.1.N}, \eqref{eq:defK} and \eqref{eq:CMCgauge} we find
\begin{equation}
    \partial_t e_i{}^\Lambda
    =t^{-1+q}\Omega_i{}^j e_j{}^\Lambda
    -t^{-1}  K_\Sigma{}^\Lambda e_i{}^\Sigma.
\end{equation}
Given the conditions in Definition~\ref{def:asympKasner}, especially the restriction $q>0$, and, that the assumption that the frame $b_{\bar i}{}^\Omega$ is time-independent, and, that $P$ is the largest eigenvalue of $\Ktt_i{}^j$, we conclude together with \eqref{eq:framebound} that the constants $\bar p$ and $P$ must satisfy $\bar p=1-P$. In particular, the condition $\bar p>0$ therefore implies that $P<1$. In fact, motivated by this we shall always impose $P<1$ throughout this whole paper.
 
Next, from \eqref{eq:dotnalpharel}, \eqref{eq:alphaasympt} and \eqref{eq:framebound} it follows that there exist maps in $C_b^\infty((0,T_0]\times\Sigma, \Rbb^N)$ written as $\dot n_{*,\bar i}$ and $\dot n_{*,i}$ such that 
\begin{equation}
    \label{eq:ndotbound}
    \dot n_{\bar i}=t^q\dot n_{*,\bar i},\quad \dot n_i=t^{-1+q+\bar p}\dot n_{*,i}.
\end{equation}

As a final remark, several conditions in Definition~\ref{def:asympKasner}, in particular those on the rescaled Weingarten map in \eqref{eq:Kttdef}, are formulated in terms of components with respect to the orthonormal frame $e_i{}^\Omega$ rather than the time-independent frame $b_{\bar i}{}^\Omega$. Although $e_i{}^\Omega$ is typically singular as $t\to 0$ by \eqref{eq:framebound}, contemporary big bang stability results\footnote{It will be interesting in future work to study the implications of the recent results in \cite{franco-grisales2026} which achieve significantly better control regarding the big bang asymptotics under certain conditions.} for the Einstein--matter equations usually provide precisely such orthonormal-frame information, which motivates our choice. If comparable assumptions could instead be phrased in a frame regular as $t\to 0$ (such as $b_{\bar i}{}^\Omega$), one could obtain substantially stronger conclusions about the fluid behaviour. 

\begin{ex}[Kasner-scalar field spacetimes]
    \label{ex:KSF}

The \emph{Kasner--scalar field family of spacetimes} is a particularly important model for big bang formation in general relativity and has been discussed in detail in many of the references cited above. It is a family of spatially homogeneous spacetimes with Kasner big bang asymptotics according to Definition~\ref{def:asympKasner} that solve the\footnote{A scalar field $\phi$ is called \emph{free-massless} if it is a solution of the scalar wave equation with \emph{vanishing} potential.} Einstein--(free-massless) scalar-field equations in the explicit form
    \begin{equation}
        M=(0,\infty)\times\Tbb^N, \quad g_{\mu\nu}=-\nabla_\mu t\otimes\nabla_\nu t+\sum_{{\bar i}=1}^N t^{2P_{\bar i}} \nabla_\mu x_{\bar i}\otimes\nabla_\nu x_{\bar i},\quad \phi=A \log t+B, 
    \end{equation}
    for real constant parameters $A$, $B$, and $P_1, \ldots, P_N$, satisfying
    \begin{equation}
        \label{eq:Kasnerrel}
        \sum_{{\bar i}=1}^N P_{\bar i}=1,\quad \sum_{{\bar i}=1}^N P_{\bar i}^2=1-A^2,\quad |A|\in \Bigl[0,\sqrt{\frac{N-1}{N}}\Bigr], \quad B\in\Rbb.
    \end{equation}
    Important special cases are the \emph{Kasner--vacuum solutions} given by $A=B=0$, and, the isotropic spatially flat Friedmann--Lema\^itre--Robertson--Walker solutions given by $|A|=\sqrt{(N-1)/N}$, $B\in\Rbb$, in which case \eqref{eq:Kasnerrel} imply that $P_{\bar i}=1/N$ for all ${\bar i}=1,\ldots,N$. Let us now express the general family of Kasner-scalar field spacetimes in the language of Definition~\ref{def:asympKasner}. First we choose $T_0=\infty$ and $\Sigma=\Tbb^N$.
    We identify the frame $b_{\bar i}{}^\Omega$ in the definition with the spatial coordinate frame $\partial_{x_{\ib}}{}^\Lambda$ and take the orthonormal frame $e_{i}{}^\Omega$ to be
    \begin{equation}
        e_i{}^\Lambda=t^{-P_{\ib}}\partial_{x_{\ib}}{}^\Lambda \delta_i^{\ib}
    \end{equation}
    for all $i=1,\ldots,N$. Without loss of generality, assume that $P_1\ge P_2\ge\ldots\ge P_N$. Then we set 
    \begin{equation}
        \bar p=1-P_1
    \end{equation} 
    and $\tau_i{}^{\bar i}=t^{P_1-P_i}\delta_i{}^{\bar i}$ in \eqref{eq:framebound}. Next, we set
    $\alpha=1$ (and hence $\alpha_*=0$ according to \eqref{eq:alphaasympt}), and 
    \begin{equation}
        K_i{}^j=\Ktt_i{}^j=\mathrm{diag}\left(P_1,\ldots,P_N\right),
    \end{equation} 
    and hence $K_{*,i}{}^j=0$ according to \eqref{eq:alphaasympt}. It follows that $P=P_1$ and that
    \begin{equation}
        \hat L_i{}^j=\mathrm{diag}\left(0,P_1-P_2\ldots,P_1-P_N\right),
    \end{equation}
    according to \eqref{eq:Kttspecdecomp}. If, for example, $P_1>P_2$, then $\hat h_i{}^j$ is the projector onto the one-dimensional subspace spanned by $e_1{}^\Lambda$, and $\check h_i{}^j$ the projector onto the $(N-1)$-dimensional orthonormal complement spanned by $\{e_2{}^\Lambda,\ldots,e_N{}^\Lambda\}$. Returning to the general case, we next verify that $D e_i{}^\Lambda=0$ (and hence $\Omega_i{}^j=0$ according to \eqref{frametransp.1.N}) and that $\mathcal L_{e_i} e_k{}^\Lambda=0$ (and hence $c_{i}{}^j{}_k=0$ according to \eqref{eq:framecommutators}). This implies that we can choose both quantities $q$ and $\bar q$ as large as we like. In total, we have verified that the Kasner-scalar field family is a family of spacetimes with Kasner big bang asymptotics according to Definition~\ref{def:asympKasner}.
\end{ex}

\subsection{Analysis background}
\label{sec:analysisbg.N}

Let $T_0$ and $\Sigma$ be as discussed in Section~\ref{sec:background.N}. Recall that by assumption the manifold $\Sigma$ is parallelisable. In the rest of this paper, except for Section~\ref{sec:Frauendiener.N}, we express all of our time-dependent and time-independent tensor fields on $\Sigma$ in terms of the orthonormal frame $e_i{}^\Omega$ in Definition~\ref{def:asympKasner}. This means that we exclusively deal with maps of the form $(0,T_0]\times\Sigma\rightarrow \Rbb^s$ or $\Sigma\rightarrow \Rbb^s$ for some non-negative integer $s$, respectively. Let now $k$, $s$ and $\bar s$ be non-negative integers. For the purpose of this discussion, we write smooth $\Rbb^s$-valued maps on $(0,T_0]\times\Sigma$ or $\Sigma$ either in abstract index notation $f_I$ with $I,J,\ldots=1,\ldots s$ or in index-free notation $f$. We sometimes write arguments explicitly, for example $f(x)$ with $x\in\Sigma$, or, if it is clear from the context, we omit arguments.

We call a map $f_{\bar I}: (0,T_0]\times\Sigma\times\Rbb^s\rightarrow \Rbb^{\bar s}$ a\footnote{For definiteness, we only write down the time-dependent case here. Similar notions can be introduced for time-independent maps.} \emph{monomial of order $k$} if there is a smooth function $f_{\bar I}{}^{I_1\ldots I_k}: (0,T_0]\times\Sigma\rightarrow \Rbb^{\bar s}\times(\Rbb^s)^k$, called the \emph{coefficient} of the monomial, such that
\begin{equation}
    f_{\bar I}(t,x,v)=f_{\bar I}{}^{I_1\ldots I_k}(t,x)\, v_{I_1}\cdots v_{I_k},
\end{equation}
for all $t\in (0,T_0]$, $x\in\Sigma$ and $v\in\Rbb^s$. We say that the monomial $f_{\bar I}$ is \emph{uniformly bounded} if $f_{\bar I}{}^{I_1\ldots I_k}$ is a member of $C^\infty_b((0,T_0]\times\Sigma, \Rbb^{\bar s}\times(\Rbb^s)^k)$. We call a map $f_{\bar I}: (0,T_0]\times\Sigma\times\Rbb^s\rightarrow \Rbb^{\bar s}$ a \emph{polynomial of order $k$} if it is a finite sum of monomials of order at most $k$ and at least one of the monomials is of order $k$. A polynomial is called \emph{uniformly bounded} if each monomial is uniformly bounded. Let $f:(0,T_0]\times\Sigma\rightarrow\Rbb$ be a smooth positive function. We say that a smooth map $F$ defined on $(0,T_0]\times\Sigma$ satisfies $F=O(f)$ provided $F/f$ is a uniformly bounded polynomial.

Given non-negative integers $s$ and $\bar s$, and an open subset $\Omega$ of $\Rbb^s$, we frequently encounter maps $F(t,x,v)$ in $C_b^\infty((0,T_0]\times\Sigma\times\Omega,\Rbb^{\bar s})$ in this paper for which there exist a non-negative integer $\bar{\bar s}$, a map $\chi\in C_b^\infty((0,T_0]\times\Sigma\times\Omega,\Rbb^{\bar{\bar s}})$ and a uniformly bounded polynomial $f_{\bar I}$ such that
\[F(t,x,v)=f_{\bar I}(t,x,\chi(t,x,v)),\]
for all $t\in (0,T_0]$, $x\in\Sigma$ and $v\in\Omega$. In this situation we state that \emph{$F$ can be identified with the uniformly bounded $f_{\bar I}$ polynomial for $\chi$ and $\Omega$.}

\section{The Euler equations for fluid solutions near the big bang}
\label{sec:mainformulations.N}

\subsection{Frauendiener-Walton formulation of the Euler equations and its $N+1$-decomposition}
\label{sec:Frauendiener.N}
\newcommand{\Vsy}{V}
\newcommand{\Vtm}{V_0}
\newcommand{\Vsp}{V}
\newcommand{\VspU}[1]{\Vsp^{#1}} 
\newcommand{\VspD}[1]{\Vsp_{#1}}
\newcommand{\Vspnorm}{|\Vsp|_h}
\newcommand{\CoeffSymmPre}[1]{{\tt {#1}}{}}
\newcommand{\CoeffSymm}{\CoeffSymmPre{a}}
\newcommand{\SourceSymmPre}[1]{{\tt {#1}}}
\newcommand{\SourceSymm}{\SourceSymmPre{f}}
\newcommand{\Coeff}{\CoeffSymmPre{A}}
\newcommand{\Source}{\SourceSymmPre{F}}
Given a background spacetime $(M,g_{\mu\nu})$ as discussed in Section~\ref{sec:background.N} (but not necessarily satisfying the properties of Definition~\ref{def:asympKasner}), we now consider a perfect fluid with a linear equation of state
\begin{equation}
    \label{eq:EOS}
    \mathcal P=c_s^2\rho,\quad c_s^2\in (0,1),
\end{equation}
where $\mathcal P$ denotes the fluid pressure and $\rho$ the fluid density. We consider the speed of sound $c_s^2$ of the fluid as a free parameter\footnote{We will explain the reason for excluding the values $c_s^2=0$ (dust) and $c_s^2=1$ (stiff fluid) below.} (given in units of the square of the speed of light) at this stage. In the framework of \cite{Walton2005,Frauendiener2003}, the complete information about the fluid is contained in a (generally unnormalised) smooth timelike vector field $V^\mu$ on $M$, with respect to which the fluid energy-momentum tensor takes the form
\begin{equation}
    \label{eq:energymomentumtensor}
    T_{\mu\nu}=\mathcal P_0\left(-\frac{1+c_s^2}{c_s^2}\frac{V_\mu V_\nu}{V_\rho V^\rho}+g_{\mu\nu}\right)\left(-V_\rho V^\rho\right)^{-(1+c_s^2)/(2c_s^2)},
\end{equation}
determined by another free constant $\mathcal P_0>0$. The Euler equations are equivalent to the divergence-freeness of the energy-momentum tensor and can be written as
\begin{equation}
    0=
    \frac{c_s^2}{\mathcal P_0(1+c_s^2)} (-V_\rho V^\rho)^{(3c_s^2+1)/c_s^2}\nabla_\mu T^{\mu\nu}=
    A^{\mu\nu \rho}\nabla_\mu V_\rho,\label{eq:Euler1}
\end{equation}
where
\begin{equation}
    A^{\mu\nu \rho}
    =-\frac{3 c_s^2+1}{c_s^2} \frac{V^\nu V^\rho V^\mu}{V^\sigma V_\sigma} +V^\mu g^{\nu\rho}
    +V^{\rho}g^{\mu\nu} +V^\nu g^{\rho\mu}. 
    \label{eq:Euler2}
\end{equation}
Recall from Section~\ref{sec:background.N} that we perform index operations with the metric $g_{\mu\nu}$ and its inverse $g^{\mu\nu}$. A key feature of the framework of \cite{Walton2005,Frauendiener2003} is that $A^{\mu\nu \rho}$ is a totally symmetric tensor field which, under suitable circumstances discussed below, has the crucial positive definiteness property discussed below.

To extract the physical information from an arbitrary given solution $V^\mu$ of \eqref{eq:Euler1}--\eqref{eq:Euler2}, the physical pressure $\mathcal P$, density $\rho$, and the (normalised) fluid velocity field $\mathcal u^\mu$ are given by
\begin{equation}
\label{eq:physicsquantitiesfluid}
    \mathcal P=\mathcal P_0 \bigl(-V^\mu V_\mu\bigr)^{-\frac {c_s^2+1}{2c_s^2}},\quad
    \rho=\frac{\mathcal P_0}{c_s^2} \bigl(-V^\mu V_\mu\bigr)^{-\frac {c_s^2+1}{2c_s^2}},\quad 
    \mathcal u^\mu=\frac{V^\mu}{\sqrt{-V^\sigma V_\sigma}}.
\end{equation} 
In order to further characterise the behaviour of the physical properties of the fluid based on the $N+1$-decomposition introduced in Section~\ref{sec:background.N}, we can decompose the fluid velocity field $\mathcal u^\mu$ as
\begin{equation}
    \label{eq:tilt.def}
    \mathcal u^\mu=\Gamma (n^\mu+\nu^\mu),\quad n_\mu \nu^\mu=0,
\end{equation}
where the purely spatial vector field $\nu^\mu$ is called the \emph{tilt} and $\Gamma$ the corresponding \emph{Lorentz factor}. It follows from the convention introduced in Section~\ref{sec:background.N} that the pull-backs to $\Sigma$ are 
\begin{equation}
    \label{eq:tilt.expr}
    \Gamma=-\mathcal u^\mu n_\mu=-\frac{V_0}{\sqrt{V_0^2-|V|_h^2}},\quad 
    \nu^0=0,\quad
    \nu^\Lambda=\frac{\mathcal u^\Omega h_{\Omega}{}^\Lambda}{\Gamma}
    =-\frac{\Vsp^\Lambda}{\Vtm}.
\end{equation}
Here we use the notation
\begin{equation}
    \Vspnorm^2:=h^{\Lambda\Omega}\Vsp_\Lambda \Vsp_\Omega,
\end{equation} 
which implies that
\begin{equation}
    V^\mu V_\mu=-\Vtm^2+\Vspnorm^2.
\end{equation}

The coefficients $A^\mu{}^{\nu \rho}$ defined in \eqref{eq:Euler2} are represented by the following tensor fields on $(0,T_0]\times\Sigma$:
\begin{align}
    \label{eq:AUDDcomp.First}
    A^{000}
    =& 
    -\frac{\Vtm^2+3c_s^2\Vspnorm^2}{(\Vtm^2-\Vspnorm^2)c_s^2} \,
    \Vtm,\\
     A^{00\Lambda}
    =&A^{0\Lambda0}=A^{\Lambda00}=\frac{(2 c_s^2+1)\Vtm^2+c_s^2\Vspnorm^2}{(\Vtm^2-\Vspnorm^2)c_s^2}\VspU \Lambda,\\
    A^{0\Omega\Lambda}
    =&A^{\Omega0\Lambda}
    =A^{\Omega\Lambda0}
    =-\Bigl(\frac{(3 c_s^2+1)}{(\Vtm^2-\Vspnorm^2)c_s^2}\VspU \Omega \VspU \Lambda +h^{\Omega\Lambda}\Bigr)\Vtm,\\
    A^{\Sigma\Omega\Lambda}=&\frac{3 c_s^2+1}{(\Vtm^2-\Vspnorm^2)c_s^2} \VspU \Omega \VspU \Lambda \VspU  \Sigma+\VspU  \Sigma h^{\Omega\Lambda}
    +h^{\Sigma\Lambda} \VspU {\Omega}
    +h^{\Sigma\Omega} \VspU {\Lambda}.\label{eq:AUDDcomp.Last}
\end{align}
In order to write \eqref{eq:Euler1} and \eqref{eq:Euler2} using these expressions next, we first derive
with the help of \eqref{eq:defh} and \eqref{eq:defdotnk} that
\begin{align}
    \nabla_0 V_0&=n^\mu n^\nu\nabla_\mu V_\nu
    =n^\mu\nabla_\mu (V_\nu n^\nu)-\dot n^\mu V_\mu
    =\D \Vtm-\dot n^\Lambda \Vsp_\Lambda,\\
    \nabla_\Lambda V_0&=h_\Lambda{}^\mu n^\nu\nabla_\mu V_\nu
    =\nabla_\Lambda (V_\nu n^\nu)-k_{\Lambda}{}^{\mu} V_\mu
    =\D_\Lambda \Vtm-k_{\Lambda}{}^{\Sigma} \Vsp_\Sigma,\\
    \nabla_0 V_\Lambda&=n^\mu h_\Lambda{}^{\nu} \nabla_\mu V_\nu
    =n^\mu h_\Lambda{}^{\nu} h_\nu{}^{\rho} \nabla_\mu V_\rho
    =n^\mu h_\Lambda{}^{\nu}  \nabla_\mu (h_\nu{}^{\rho} V_\rho)
    -\dot n_\Lambda V_0
    =\D \Vsp_\Lambda
    -\dot n_\Lambda \Vtm,\\
    \nabla_\Sigma V_\Lambda&=h_{\Sigma}{}^\mu h_\Lambda{}^{\nu} \nabla_\mu V_\nu
    =h_{\Sigma}{}^\mu h_\Lambda{}^{\nu} h_\nu{}^{\rho} \nabla_\mu V_\rho
    =h_{\Sigma}{}^\mu h_\Lambda{}^{\nu}  \nabla_\mu (h_\nu{}^{\rho} V_\rho)
    -k_{\Lambda\Sigma} \Vtm
    =\D_\Sigma\Vsp_\Lambda-k_{\Lambda\Sigma} \Vtm,
\end{align}
where we recall that we consider $\Vtm$ as a time-dependent scalar function on $\Sigma$, the operator $\D_\Lambda$ is the covariant derivative associated with $h_{\Lambda\Omega}$, and the operator $\D$ is  defined in \eqref{eq:defD}. 
We can now write \eqref{eq:Euler1} and \eqref{eq:Euler2} using the following vector notation
\begin{equation}
    \label{eq:EulerVectLHS}
    \CoeffSymm^0 \D \Vsy
    +\CoeffSymm^\Sigma \D_\Sigma{} \Vsy
    =\SourceSymm,
\end{equation}
where
\begin{equation}
    V=\begin{pmatrix}
        \Vtm\\
        \VspD \Lambda
    \end{pmatrix},
\end{equation}
and
\begin{align}
    \CoeffSymm^0&=
    \begin{pmatrix}
        \CoeffSymm^{000} & \CoeffSymm^{00\Lambda}\\
        \CoeffSymm^{0\Omega0} & \CoeffSymm^{0\Omega\Lambda}
    \end{pmatrix}
    =-\begin{pmatrix}
        A^{000} & A^{00\Lambda}\\
        A^{0\Omega0} & A^{0\Omega\Lambda}
    \end{pmatrix},\\
    \CoeffSymm^\Sigma&=
    \begin{pmatrix}
        \CoeffSymm^{\Sigma00} & \CoeffSymm^{\Sigma0\Lambda}\\
        \CoeffSymm^{\Sigma\Omega0} & \CoeffSymm^{\Sigma\Omega\Lambda}
    \end{pmatrix}
    =-\begin{pmatrix}
        A^{\Sigma00} & A^{\Sigma0\Lambda}\\
        A^{\Sigma\Omega0} & A^{\Sigma\Omega\Lambda}
    \end{pmatrix},
\end{align}
that is,
\begin{align}
    \CoeffSymm^0&
    =\begin{pmatrix}
        \frac{\Vtm^2+3c_s^2\Vspnorm^2}{(\Vtm^2-\Vspnorm^2)c_s^2} \,\Vtm & -\frac{(2 c_s^2+1)\Vtm^2+c_s^2\Vspnorm^2}{(\Vtm^2-\Vspnorm^2)c_s^2}\Vsp^\Lambda\\
        -\frac{(2 c_s^2+1)\Vtm^2+c_s^2\Vspnorm^2}{(\Vtm^2-\Vspnorm^2)c_s^2}\Vsp^\Omega & \Bigl(\frac{3 c_s^2+1}{(\Vtm^2-\Vspnorm^2)c_s^2}\Vsp^\Omega \Vsp^\Lambda +h^{\Omega\Lambda}\Bigr)\Vtm
    \end{pmatrix},\\
    \CoeffSymm^\Sigma&
    =\begin{pmatrix}
        -\frac{(2 c_s^2+1)\Vtm^2 +c_s^2\Vspnorm^2}{(\Vtm^2-\Vspnorm^2)c_s^2}\Vsp^\Sigma & \Bigl(\frac{3 c_s^2+1 }{(\Vtm^2-\Vspnorm^2)c_s^2} \Vsp^\Lambda \Vsp^\Sigma
        +{h^{\Sigma\Lambda}}\Bigr)\Vtm\\
        \Bigl(\frac{3 c_s^2+1 }{(\Vtm^2-\Vspnorm^2)c_s^2} \Vsp^\Omega \Vsp^\Sigma
    +{h^{\Sigma\Omega}}\Bigr) \Vtm 
    & -\frac{3 c_s^2+1}{(\Vtm^2-\Vspnorm^2)c_s^2} \Vsp^\Omega \Vsp^\Lambda \Vsp^\Sigma
    -\Vsp^\Sigma h^{\Omega\Lambda}
    -h^{\Sigma\Lambda} \Vsp^{\Omega}
    -h^{\Sigma\Omega} \Vsp^{\Lambda}
    \end{pmatrix},\\
    \SourceSymm
    &=\CoeffSymm^0 \begin{pmatrix}
    \dot n^\Sigma \Vsp_\Sigma\\
    k_\Lambda{}^\Sigma V_\Sigma+\dot n_\Lambda \Vtm  
  \end{pmatrix}
    +\CoeffSymm^\Sigma\begin{pmatrix}
    k_{\Sigma}{}^{\Sigma'}\Vsp_{\Sigma'} 
    \\
    {k_{\Sigma\Lambda}} \Vtm
  \end{pmatrix}\\
  &=\begin{pmatrix}
  k_\Sigma{}^\Sigma \Vtm&  k^{\Sigma\Lambda}\Vsp_\Sigma\\
  0 & 
  -k_\Sigma{}^\Sigma\Vtm h^{\Omega\Lambda} 
  - k^{\Omega\Lambda} \Vtm
  \end{pmatrix}V
  +\begin{pmatrix}
   -2 \Vtm \dot n^\Sigma \Vsp_\Sigma\\
    \Vtm^2\dot n^\Omega+\dot n^\Sigma\Vsp_\Sigma\Vsp^\Omega
  \end{pmatrix}
  \notag.
\end{align}
Observe that the matrices $\CoeffSymm^0$ and $\CoeffSymm^\Sigma$ are symmetric.
The inverse of the matrix $\CoeffSymm^0$ is
\begin{equation}
    (\CoeffSymm^0)^{-1}=
    \begin{pmatrix}
    \frac{(2 c_s^2+1) \Vspnorm^2 +c_s^2 \Vtm ^2}{\left(\Vtm ^2-\Vspnorm^2\right) \left(\Vtm ^2-c_s^2 \Vspnorm^2\right)} \Vtm
    & \frac{ c_s^2 \Vspnorm^2 +(2 c_s^2+1) \Vtm ^2}{\left(\Vtm ^2-\Vspnorm^2\right) \left(\Vtm ^2-c_s^2 \Vspnorm^2\right)} \Vsp_\Omega\\
    \frac{ c_s^2 \Vspnorm^2  +(2 c_s^2+1) \Vtm ^2}{\left(\Vtm ^2-\Vspnorm^2\right) \left(\Vtm ^2-c_s^2 \Vspnorm^2\right)} \Vsp_\Pi
    & \frac{\left(c_s^2 \Vspnorm^4-(c_s^2+1) \Vspnorm^2 \Vtm ^2+\Vtm ^4\right) h_{\Omega \Pi}+ \left((4 c_s^2+1) \Vtm ^2-c_s^2 \Vspnorm^2\right)\Vsp_\Omega\Vsp_\Pi}{\Vtm  \left(\Vtm ^2-\Vspnorm^2\right) \left(\Vtm ^2-c_s^2 \Vspnorm^2\right)} 
    \end{pmatrix}.
\end{equation}
Multiplying \eqref{eq:EulerVectLHS} from the left with this inverse matrix yields the non-symmetrised form of the Euler equations 
\begin{equation}
    \D\Vsy
    +\Coeff^\Sigma \D_\Sigma{} \Vsy
    =\Source,
\end{equation}
where
\begin{align}
    \Coeff^\Sigma=&(\CoeffSymm^0)^{-1}\CoeffSymm^\Sigma
    =\begin{pmatrix}
    0 & -\frac{  (c_s^2+1) \Vsp^{\Lambda} \Vsp^{\Sigma}}{\Vtm ^2-c_s^2 \Vspnorm^2}+\frac{  c_s^2 \left(\Vspnorm^2+\Vtm ^2\right) h^{\Lambda\Sigma}}{\Vtm ^2-c_s^2 \Vspnorm^2} \\
      h_{\Pi}{}^{\Sigma} & \frac{2   c_s^2 \Vtm  \Vsp_{\Pi} h^{\Lambda\Sigma}}{\Vtm ^2-c_s^2 \Vspnorm^2}-\frac{  \Vsp^{\Lambda} h_{\Pi}{}^{\Sigma}}{\Vtm }-\frac{  \Vsp^{\Sigma} h_{\Pi}{}^{\Lambda}}{\Vtm }-\frac{2   c_s^2 \Vsp^{\Lambda} \Vsp^{\Sigma} \Vsp_{\Pi}}{\Vtm(\Vtm ^2-c_s^2\Vspnorm^2 )} 
    \end{pmatrix},\\
    \Source&=(\CoeffSymm^0)^{-1}\SourceSymm\\
    =&\begin{pmatrix}
    \frac{c_s^2 k_\Omega{}^\Omega \left(\Vspnorm^2+\Vtm ^2\right)
    -(c_s^2+1)  k^{\Omega\Omega'} \Vsp_{\Omega} \Vsp_{\Omega'}}{\Vtm ^2-c_s^2 \Vspnorm^2}
     & 0
    \\
    0 & 
    \frac{2 c_s^2( \Vtm ^2 k_\Omega{}^\Omega
    -  k^{\Omega\Omega'} \Vsp_{\Omega} \Vsp_{\Omega'}
    )h_{\Pi}{}^\Lambda}{\Vtm ^2-c_s^2 \Vspnorm^2}
    -k_\Pi{}^\Lambda 
    \end{pmatrix}V
    +\begin{pmatrix}
        \dot n^\Omega\Vsp_\Omega\\
        +\dot n_\Pi \Vtm
    \end{pmatrix}.\notag
    \notag
\end{align}

Given the linear equation of state and restriction for the speed of sound parameter in \eqref{eq:EOS}, the Frauen\-diener-Walton formulation of the Euler equation is completely general. Especially we do not yet assume that the spacetime satisfies the properties of Definition~\ref{def:asympKasner}. Because the Euler equations in this form are symmetrisable\footnote{One can check that the matrix $\CoeffSymm^0$ is never positive definite for $c_s^2=0$. This is the main reason why we exclude this case for the whole paper.}, the local-in-time initial value problem is well-posed provided the initial data for the vector field $V^\mu$ are timelike. Regarding the global-in-time initial value problem on backgrounds with Kasner big bang asymptotics (Definition~\ref{def:asympKasner}), however, this formulation is useful for solutions $V^\mu$ that are \emph{uniformly} timelike only, especially the norm of the tilt $\nu^\mu$ must be bounded away from $1$. Only then, the matrix $\CoeffSymm^0$ is \emph{uniformly} positive definite. This turns out to be the case \cite{BeyerLeLoch:2017,BeyerOliynyk:2020, beyerStabilityFLRWSolutions2023} in the \emph{asymptotically non-tilted regime} when the speed-of-sound parameter $c_s^2$ is sufficiently large, but not in the \emph{asymptotically extremely-tilted case} of interest to the paper here when the speed of sound is small in the sense that $c_s^2<P$ in Definition~\ref{def:asympKasner}. 

\subsection{Convenient variables and heuristics}
\label{sec:convvariables.N}
\newcommand{\NewVar}{U}

The goal of this subsection now is to introduce a first new set of fluid variables which are particularly convenient to capture the heuristics of the asymptotically extremely-tilted dynamics.  
\renewcommand{\Vsy}{U}
\renewcommand{\Vtm}{u}
\renewcommand{\Vsp}{u}
\newcommand{\Pro}{\Pi}
\newcommand{\VspPronorm}{|\Vsp|_\Pro}
\newcommand{\VspProPerpnorm}{|\Vsp|_{\Pro^\perp}}
\newcommand{\VspPro}{{\hat\Vsy}{}}
\newcommand{\VspProPerp}{{\check\Vsy}{}}
\renewcommand{\CoeffSymm}{\CoeffSymmPre{b}}
\renewcommand{\SourceSymm}{\SourceSymmPre{g}}
\renewcommand{\Coeff}{\CoeffSymmPre{B}}
\renewcommand{\Source}{\SourceSymmPre{G}}

Assuming now that our background spacetime has Kasner big bang asymptotics according to Definition~\ref{def:asympKasner}, we start by expressing the Euler equations in terms of the corresponding spatial orthonormal frame $e_i{}^\Lambda$. We introduce a new variable vector $X$ as a $\Rbb^{N+1}$-valued function written as
\begin{equation}
    \label{eq:Xdef}
    X=\begin{pmatrix}
        u\\ \Vsp_k\\
    \end{pmatrix},
\end{equation}
where the $N$ scalar functions represented by $\Vsp_k$ are the orthonormal frame components of a covector $\Vsp_\Omega$ defined now. Using the general framework for variable transformations in Section~\ref{sec:nonlineartransf}, we choose 
\begin{equation}        
    \label{eq:VariableTrafo.1.Spec}
    \begin{pmatrix}
    V_0\\V_k
    \end{pmatrix}=\Phi(t,x,X)=\frac{t^{\Ptt-2\Ltt}}{\Vtm^2}\begin{pmatrix}
        -\sqrt{t^{2\Ltt}\Vtm^2+\Vspnorm^2}\\\Vsp_k
    \end{pmatrix},
\end{equation}
where $\Ltt$ and $\Ptt$ are, for now, arbitrary smooth functions on $\Sigma$. Observe that in the language of Section~\ref{sec:nonlineartransf} we have $s=\bar s=\sigma=N+1$ provided $u>0$, and \eqref{eq:VariableTrafo.1.Spec} therefore yields a local diffeomorphism. The inverse can be written as 
\begin{equation}
   V_\mu=\frac{t^{\mathtt P-2\Ltt}}{\Vtm^2}\left(\sqrt{t^{2\Ltt}\Vtm^2+\Vspnorm^2}\,n_\mu+u_\mu\right),\quad |V|^2_h=\frac{t^{2\Ptt-4\Ltt}}{u^4} \Vspnorm^2, \quad V_\mu V^\mu=-\frac{t^{2(\Ptt-\Ltt)}}{\Vtm^2},
\end{equation}
where
\begin{equation}
    |u|_h=\sqrt{h^{ij}u_i u_j},\quad |V|_h=\sqrt{h^{ij}V_i V_j}.
\end{equation}
Consequently, the normalised fluid velocity vector $\mathcal u^\mu$ in \eqref{eq:physicsquantitiesfluid} is timelike whenever $\Vtm$ is a finite positive quantity, and \eqref{eq:tilt.def} and \eqref{eq:tilt.expr} imply
\begin{equation}
    \label{eq:Gammanu.Var2}
    \rho=\frac{\mathcal P_0}{c_s^2}t^{-(\Ptt-\Ltt)(1+c_s^2)/c_s^2}\, u^{(1+c_s^2)/c_s^2},\quad
    \Gamma=t^{-\Ltt}\frac{\sqrt{t^{2\Ltt}\Vtm^2+\Vspnorm^2}}{\Vtm},\quad 
    \nu_i=\frac{\Vsp_i}{\sqrt{t^{2\Ltt}\Vtm^2+\Vspnorm^2}}.
\end{equation}

It will turn out to be convenient to express the Levi-Civita connection $\D_\Lambda$ (defined in terms of the physical time-dependent metric $h_{\Lambda\Omega}$) in terms of some other torsion-free \emph{time-independent} reference connection $\Dc_\Lambda$. We can specify this new connection in terms of a smooth time-dependent tensor field $C_i{}^k{}_j$ such that
\begin{equation}
    \label{eq:timeindepspatconn.1}
    \D_i e_j{}^\Lambda=\Dc_i e_j{}^\Lambda+C_i{}^k{}_j e_k{}^\Lambda.
\end{equation}
If we now impose that
\begin{equation}
    \label{eq:timeindepspatconn.2}
    \Dc_\Omega e_j{}^\Lambda=0,
\end{equation}
then it is a standard consequence of \eqref{eq:framecommutators} and the torsion-free condition that
\begin{equation}
    \label{eq:timeindepspatconn.3}
    t^{1-\bar q} C_i{}^k{}_j=\frac 12 h^{kl}\left(
        h_{ld}c_i{}^d{}_j
        +h_{jd}c_l{}^d{}_i
        -h_{id}c_j{}^d{}_l
    \right).
\end{equation}
It follows that all functions $t^{1-\bar q} C_i{}^k{}_j$ are in $C_b^\infty((0,T_0]\times\Sigma,\Rbb)$. We interpret $\Dc_\Lambda$ as the unique torsion-free connection with respect to which  orthonormal frame components are differentiated along $\Sigma$ like scalar functions. 

In order to find the system of PDEs for the variable $X$ equivalent to the Euler equations, we apply the formulas \eqref{eq:PDEgen.transf.nonsymm} -- \eqref{eq:PDEgen.transf.nonsymm.G} to \eqref{eq:VariableTrafo.1.Spec}. This, together with \eqref{eq:defD}, \eqref{eq:D2Lie}, \eqref{eq:tvf}, \eqref{eq:vanishshift}, \eqref{eq:timeder}, \eqref{frametransp.1.N} and \eqref{eq:timeindepspatconn.1} -- \eqref{eq:timeindepspatconn.3}, yields the following system for $X$ equivalent to \eqref{eq:EulerVectLHS} which we write in non-symmetrised form
\begin{equation}
    \label{eq:PDEgen.transf.nonsymm.N}
    \partial_t X
    +\alpha \mathtt b^{l} \Dc_l{} X
    = \mathtt g,
\end{equation}
with
\begin{align}
    \mathtt b^l=&\frac{1}{\sqrt{t^{2\Ltt}\Vtm^2+\Vspnorm^2}}\begin{pmatrix}
        \frac{(1-c_s^2)\Vspnorm^2+(1-2c_s^2)t^{2\Ltt}\Vtm^2}{(1-c_s^2)\Vspnorm^2+t^{2\Ltt}\Vtm^2}\Vsp^l 
        & \frac{c_s^2\Vtm}{(1-c_s^2)\Vspnorm^2+t^{2\Ltt}\Vtm^2} {P}^{kl} \\
        t^{2\Ltt}\Vtm {h}_{j}{}^{l} & {h}_{j}{}^{k} \Vsp^l
    \end{pmatrix},\label{eq:variables1.bL.NNN}\\
    \mathtt g=&\alpha\begin{pmatrix}
        \frac{\Ptt-\Ltt }{\alpha t}
        -c_s^2k_p{}^p\frac{t^{2\Ltt}\Vtm^2+\Vspnorm^2}{t^{2\Ltt}\Vtm^2+(1-c_s^2) \Vspnorm^2}
        +\frac{c_s^2  k^{mn}\Vsp_m\Vsp_n}{t^{2\Ltt}\Vtm^2+(1-c_s^2) \Vspnorm^2}
        & 0\\
        0 & \frac{\Ptt }{\alpha t} \delta_j{}^k- k_j{}^k
    \end{pmatrix} \begin{pmatrix}
    \Vtm\\
    \Vsp_k
    \end{pmatrix} \label{eq:variables1.g.NNN}\\
    &+
    \alpha\sqrt{t^{2\Ltt}\Vtm^2+\Vspnorm^2}\begin{pmatrix}
        \frac{(1-c_s^2)  \D_{p} \Ptt -\frac{ (1-2 c_s^2) t^{2 \Ltt}\Vtm^2 +(1-c_s^2) \Vspnorm^2}{t^{2\Ltt}\Vtm^2 +\Vspnorm^2}  \D_{p} \Ltt}{t^{2\Ltt}\Vtm^2+(1-c_s^2) \Vspnorm^2}\Vtm \Vsp^{p}\log (t)
        \\
        -\dot n_{j} +\frac{  \left(t^{2\Ltt}\Vtm^2 \delta_{j}{}^{p}+\Vsp_{j} \Vsp^{p}\right)\D_{p} \Ptt-t^{2\Ltt}\Vtm^2 \D_{j} \Ltt}{t^{2\Ltt}\Vtm^2+\Vspnorm^2}\log (t)
    \end{pmatrix}\notag\\
    &+t^{-1+q}\begin{pmatrix}
        0\\ \Omega_j{}^m u_m
    \end{pmatrix}
    +\alpha t^{-1+\bar q}\Coeff^{l} \begin{pmatrix}0\\  t^{1-\bar q} C_l{}^m{}_ju_m\end{pmatrix},
    \notag
\end{align}
where
\begin{equation}
    \label{eq:defP}
    P^{jk}=(t^{2\Ltt}\Vtm^2+\Vspnorm^2)h^{jk}-\Vsp^j \Vsp^k.
\end{equation}
Moreover, we find that 
\begin{equation}
    \CoeffSymm^0=\begin{pmatrix}
        \frac{t^{2\Ltt}\Vtm^2+(1 - c_s^2) \Vspnorm^2}{c_s^2} & 0\\
        0 & t^{-2\Ltt }P^{jk}
    \end{pmatrix}\label{eq:variables1.b0.N}
\end{equation}
is a symmetriser. This symmetriser is positive definite for all $t>0$ provided $u>0$ as we can verify from \eqref{eq:defP} for an arbitrary covector $x_k$ as follows
\[t^{2\Ltt}   P^{jk} x_j x_k
=(\Vtm^2+t^{-2\Ltt}\Vspnorm^2)|x|_h^2-t^{-2\Ltt}\left<u,x\right>_h^2
\ge (\Vtm^2+t^{-2\Ltt}\Vspnorm^2)|x|_h^2-t^{-2\Ltt}|x|_h^2\Vspnorm^2
=\Vtm^2|x|_h^2.
\]
The PDE \eqref{eq:PDEgen.nonsymm} therefore has a well-posed local-in-time initial value problem.
However, if $\Ltt>0$ and $|u|_h$ is bounded away from zero (which turns out to be the most important case of interest here, see below), this symmetriser is unbounded in the limit $t\rightarrow 0$: pick for example a covector $x_k$ with $|x|_h=1$ and $\left<u,x\right>_h=0$. It is for this reason that we cannot use the variables $X$ to analyse the problem of interest of this paper. Nevertheless, the variable $X$ is useful for the following heuristic discussion.

We now wish to develop some more quantitative intuition for the general behaviour of solutions of the Euler equations on arbitrary background spacetimes with Kasner big bang asymptotics (Definition~\ref{def:asympKasner}). It turns out that the variables $X$ are particularly suitable for this. For the purpose of this heuristic discussion, let us restrict to the case of spatially homogeneous background spacetimes with Kasner big bang asymptotics, and we assume that all surfaces $\Sigma_t$ are homogeneous for each $t\in (0,T_0]$. More precisely we assume that the orthonormal frame in Definition~\ref{def:asympKasner} is invariant under the action of the corresponding Killing vector fields, that $\alpha=1$, that the orthonormal frame is parallel transported, hence $\Omega_i{}^j=0$ according to \eqref{frametransp.1.N}, and that the spatial curvature of the background spacetime vanishes, that is, $C_l{}^m{}_j=0$ (see \eqref{eq:framecommutators}). According to the Bianchi classification of spatially homogeneous spacetimes \cite{wainwright1997}, our background spacetime is therefore a Bianchi I spacetime. We also restrict this heuristic discussion here to the analysis of the leading-order dynamics near $t=0$ only and consequently set $K_i{}^j=\Ktt_i{}^j$ in \eqref{eq:Kttdef}. Assuming finally that the fluid solution itself is invariant under the action of the corresponding Killing vector fields, and therefore $\Dc_l X=0$, we conclude that the equations \eqref{eq:PDEgen.transf.nonsymm.N}  reduce to the following ODE system
\begin{equation}
    \label{eq:PDEgen.transf.nonsymm.sh}
    \partial_t \begin{pmatrix}
    \Vtm\\
    \Vsp_j
    \end{pmatrix}
    =\frac{1}{t}\begin{pmatrix}
        \frac{(P-\Ltt-c_s^2)t^{2\Ltt}\Vtm^2+((P-\Ltt)(1-c_s^2)-c_s^2) \Vspnorm^2+c_s^2  \Ktt^{mn}\Vsp_m\Vsp_n}{t^{2\Ltt}\Vtm^2+(1-c_s^2) \Vspnorm^2} u
        \\
        \check L_j{}^k u_k
    \end{pmatrix},
\end{equation}
making use of  \eqref{eq:defdotnk.2.N}, \eqref{eq:defK}, \eqref{eq:Kttspecdecomp} and \eqref{eq:CMCgauge}, making the specific choice
 $\Ptt=P=const$ and assuming that $\Ltt=const$. 

We wish to classify all the different kinds of solutions $(u, u_j)$ of \eqref{eq:PDEgen.transf.nonsymm.sh}. 
We start by considering the special solution for which all variables $u_k$ vanish identically:
\begin{enumerate}
    \item \emph{Non-tilted fluid} $u_k=0$: The general solution is $u=u_* t^{P-c_s^2-\Ltt}$ for any constant $u_*>0$. According to \eqref{eq:Gammanu.Var2}, this implies that
    \begin{equation}
        \rho=\frac{\mathcal P_0}{c_s^2} u_*^{(1+c_s^2)/c_s^2}t^{-(1+c_s^2)},\quad
        \Gamma=1,\quad 
        \nu_i=0.
    \end{equation}
    The fluid is called \emph{non-tilted} (or \emph{comoving}) because the tilt $\nu_i$ vanishes identically.
\end{enumerate}

Let us next consider solutions of \eqref{eq:PDEgen.transf.nonsymm.sh}, for which the functions $u_k$ do not all vanish identically. We conclude for the solutions from  the before mentioned properties of $\check L_j{}^k$ that, (i), $\hat h_i{}^j u_j$ is a constant in time, and (ii), $\check h_i{}^j u_j$ decays like a positive power of $t$ at $t=0$. We can now use this information to simplify the equation for $u$ heuristically in leading order near $t=0$, writing $\hat u_{*,i}$ for the constant $\hat h_i{}^j u_j$ at $t=0$:
\begin{equation}
    \partial_t (t^\Ltt u)=\frac {P-c_s^2}t \frac{(t^{\Ltt}\Vtm)^2+|\hat u_*|^2_h}{(t^{\Ltt}\Vtm)^2+(1-c_s^2) |\hat u_*|^2_h} t^\Ltt u.
\end{equation}
Picking 
\begin{equation}
    \label{eq:Lttchoice.N}
    \Ltt=\frac{P-c_s^2}{1-c_s^2},
\end{equation}  
we can write an implicit formula for the general solution 
\begin{equation}
    \label{eq:implicitshu}
    \Bigl(t^{2\frac{P-c_s^2}{1-c_s^2}} u^2(t)+|\hat u_*|^2_h\Bigr)^{\frac{c_s^2}{2(1-c_s^2)}}u(t)=u_*,
\end{equation}
where $u_*>0$ is an integration constant.  This implicit solution has the following interesting special cases:
\begin{enumerate}
    \setcounter{enumi}{1}
    \item \emph{Asymptotically non-tilted fluids} $P<c_s^2$: In leading order at $t=0$, \eqref{eq:implicitshu} implies that 
    \[ u=u_*^{1-c_s^2}t^{{c_s^2}\frac{c_s^2-P}{1-c_s^2}},\]
    irrespective of the value of $|\hat u_*|_h$. According to \eqref{eq:Gammanu.Var2}, this means that
    \begin{equation}
        \rho=\frac{\mathcal P_0}{c_s^2}\, u_*^{(1-c_s^2)(1+c_s^2)/c_s^2}t^{-(1+c_s^2)},\quad
        \Gamma=1,\quad 
        \nu_i={\hat u_{*,i}}u_*^{-(1-c_s^2)}t^{{c_s^2-P}},
    \end{equation}
    in leading order at $t=0$. This is called the \emph{asymptotically non-tilted} (or \emph{asymptotically comoving}) case because the tilt $\nu_i$ decays to zero near $t=0$.
    \item \emph{Asymptotically extremely-tilted fluids} $P>c_s^2$: In leading order, 
    \[u=\begin{cases}
    u_*, & \text{if } \hat u_{*,i}\not=0,\\
    u_*^{1-c_s^2}t^{-{c_s^2}\frac{P-c_s^2}{1-c_s^2}},
    & \text{if } \hat u_{*,i}=0.
    \end{cases}\]
    It follows from \eqref{eq:Gammanu.Var2} that
    \begin{align}
        \rho&=\begin{cases}
            \frac{\mathcal P_0}{c_s^2} u_*^{(1+c_s^2)/c_s^2}t^{-\frac{1-P}{1-c_s^2}(1+c_s^2)}, & \text{if } \hat u_{*,i}\not=0,\\
            \frac{\mathcal P_0}{c_s^2}\, u_*^{(1-c_s^2)(1+c_s^2)/c_s^2}t^{-(1+c_s^2)}, & \text{if } \hat u_{*,i}=0,
        \end{cases}\\
        \Gamma&=\begin{cases}
            \frac{|\hat u_*|_h}{\Vtm_*}t^{-\frac{P-c_s^2}{1-c_s^2}}, & \text{if } \hat u_{*,i}\not=0,\\
            1, & \text{if } \hat u_{*,i}=0,
        \end{cases}\\
        \nu_i&=\begin{cases}
            \frac{\hat u_{*,i}}{|\hat u_*|_h}, &\text{if } \hat u_{*,i}\not=0,\\
            0,  & \text{if } \hat u_{*,i}=0.
        \end{cases}
    \end{align}
    This is called the \emph{asymptotically extremely-tilted} case because the tilt $\nu_i$ approaches a unit vector and the Lorentz factor diverges in leading order near $t=0$ at all points where $|\hat u_*|_h$ does not vanish.
\end{enumerate} 
To shorten our discussion, we omit the \emph{critical case} $P=c_s^2$.
Heuristically, we expect \emph{general} solutions of the Euler equations on spacetimes with Kasner big bang asymptotics, especially solutions without symmetries, to agree with these formulas pointwise asymptotics as $t\to 0$, possibly except at special points where $\hat u_{*,i}=0$. In fact, one can check that when we impose no symmetry restrictions, the general equations \eqref{eq:PDEgen.transf.nonsymm.N} are formally consistent with the just derived asymptotics in the sense that any error terms, that arise when we plug the above asymptotic formulas into these equations and take into account all cancellations, are small in the limit $t\searrow 0$ in comparison to the size of the time derivative and $1/t$-terms in the equations. 
Anticipating one of our later discussions, this is the case irrespective of the size of the positive quantities $\bar p$, $q$, and $\bar q$ in Definition~\ref{def:asympKasner}.  

Given \eqref{eq:CMCgauge}, \eqref{eq:alphaasympt} and the general condition $c_s^2<1$, the fluid density blows up more slowly than the square of the mean curvature $H^2\sim t^2$ in \emph{all} the cases above.  This suggests that fluids are consistent with the \emph{matter does not matter} motto -- irrespective of the regime. This has been well known in the asymptotically non-tilted case; see especially \cite{BeyerLeLoch:2017, BeyerLeLoch:2017,beyerStabilityFLRWSolutions2023}.

For this whole paper, we restrict to the  asymptotically extremely-tilted regime in the case that the fluid variable $u_i$, and especially the component $\hat h_i{}^j u_j$ vanishes nowhere on $\Sigma$. In the more general situation with special points at which $\hat u_{*,i}$ vanishes, we expect more delicate asymptotics, similar to results in \cite{oliynyk2024,beyer2023a} for the expanding case.  Numerical investigations of the corresponding \emph{tilt instability} near the big bang can be found in \cite{beyer2024a}.

We have seen that the system \eqref{eq:PDEgen.transf.nonsymm.N} is symmetrisable and \eqref{eq:variables1.b0.N} is a symmetriser. This implies, given certain further restrictions, that the local-in-time initial value problem of \eqref{eq:PDEgen.transf.nonsymm.N} is  well-posed for fluids on backgrounds with Kasner big bang asymptotics. In the asymptotically extremely-tilted case, however, the heuristic results above suggest that the fluid solutions are driven towards the boundary of where the fluid model is valid as $|\nu|\rightarrow 1$. As a consequence the symmetriser \eqref{eq:variables1.b0.N} is unbounded in the limit $t\rightarrow 0$.   The proofs of our two main theorems about the asymptotically extremely-tilted regime, Theorems~\ref{thm:Euler1} and \ref{thm:Euler2}, which we discuss next in Sections~\ref{sec:thm1} and \ref{sec:thm2},  rely on distinct ways to construct symmetrisers that are uniform positive definite even at the borderline of validity of the fluid model. 

\section{The asymptotically extremely-tilted fluid regime near the big bang}
\label{sec:thm1}
\subsection{The first main theorem}
\label{subsec:thm1}
In this section we present our first main theorem regarding the big bang dynamics of solutions of the Euler equations for a perfect fluid with linear equation of state \eqref{eq:EOS} on a spacetime $(M,g_{\mu\nu})$ with Kasner big bang asymptotics as in Definition~\ref{def:asympKasner} in the \emph{asymptotically extremely-tilted regime} $0<c_s^2<P<1$. 
We solve the Euler equations as an initial value problem launching from an initial time $t=T_0>0$ in the future of the big bang with initial data $\rho_{\mathbf{0}}>0$  (density) and $\nu_{\mathbf{0},i}$ (tilt vector) evolving to the past towards the big bang at $t=0$. The only restrictions (apart from regularity assumptions) on the initial data we are making here is that the initial tilt vector and density are nowhere vanishing. Beyond this, however, we require no smallness assumption for the initial data. 

\begin{thm}
    \label{thm:Euler1}
    Let $(M,g_{\mu\nu})$ be a spacetime with Kasner big bang asymptotics as in Definition~\ref{def:asympKasner}. Pick $\mathcal P_0>0$ and $c_s^2\in (0,1)$ such that the largest eigenvalue $P$ of the asymptotic rescaled Weingarten map $\Ktt_i{}^j$ in  Definition~\ref{def:asympKasner} satisfies
    \begin{equation}
        \label{eq:theorem1restr}
        0<\frac{P(x)-c_s^2}{1-c_s^2}<\min\{1,\bar p(x), q(x), \bar q(x)\}
    \end{equation}
    at each point $x$ in $\Sigma$. 
    Consider the class of fluid initial data $\nu_{\mathbf{0},i}\in H^{k}(\Sigma)$ and $\rho_{\mathbf{0}}\in H^{k}(\Sigma)$ with $k>N/2+3$ such that $|\nu_\mathbf{0}|^2_h<1$ and neither $\rho_{\mathbf{0}}$ nor $\hat h_i{}^j\nu_{\mathbf{0},j}$ vanish anywhere on $\Sigma$. 
    
    Given arbitrary fluid initial data in this class, the Cauchy problem of the Euler equations for a perfect fluid with equation of state \eqref{eq:EOS} and speed of sound parameter $c_s^2$ launched from $t=T_0>0$ with $T_0$ sufficiently small has a unique global classical solution  $\rho,\Gamma,\nu_i\in C^0((0,T_0],H^{k}(\Sigma))\cap C^1((0,T_0],H^{k-1}(\Sigma))\cap L^\infty((0,T_0],H^{k-2}(\Sigma))$ with $\rho>0$ and $0<\nu_i\nu^i<1$ on $(0,T_0]\times\Sigma$.
    There exist a non-negative map $u_*\in H^{k-2}(\Sigma)$ and a non-vanishing map $\hat u_{*,i}\in H^{k-2}(\Sigma)$ with $\check h_i{}^j \hat u_{*,j}=0$ and a constant $\kappa>0$ such that
    \begin{equation}
        \label{eq:thm1.rhoestimate}
        \Bnorm{t^{(1+c_s^2)\frac{1-P}{1-c_s^2}}\rho-T_0^{(1+c_s^2)\frac{1-P}{1-c_s^2}} \rho_\mathbf{0}(1+u_*)^{(1+c_s^2)/c_s^2}}_{H^{k-2}(\Sigma)}\lesssim t^\kappa,
    \end{equation}
    \begin{equation}
        \label{eq:thm1.Lorentzestimate}
        \Bnorm{t^{(P-c_s^2)/(1-c_s^2)}\Gamma-T_0^{(P-c_s^2)/(1-c_s^2)}\frac{\sqrt{\hat h^{kl}(\nu_{\mathbf{0},k}+\hat u_{*,k})(\nu_{\mathbf{0},l}+\hat u_{*,l})}}{\sqrt{1-|\nu_\mathbf{0}|^2}(1+u_*)}}_{H^{k-2}(\Sigma)}
        \lesssim t^\kappa,
    \end{equation}
    where $\nu_{\mathbf{0},l}+\hat u_{*,l}$ is non-vanishing on $\Sigma$,
    and
    \begin{equation}
        \label{eq:thm1.tiltestimate}
        \Bnorm{\nu_i(t)-\hat h_i{}^j\frac{ \nu_{\mathbf{0},j}+\hat u_{*,j}}{\sqrt{\hat h^{kl}(\nu_{\mathbf{0},k}+\hat u_{*,k})(\nu_{\mathbf{0},l}+\hat u_{*,l})}}}_{H^{k-2}(\Sigma)}
        \lesssim t^\kappa,
    \end{equation}
    for all $t\in (0,T_0]$. Thus the fluid represented by this solution is \emph{asymptotically extremely-tilted} near the big bang, in the sense that the Lorentz factor $\Gamma$ diverges and the tilt vector field $\nu_i$ converges to a unit vector in the eigenspace of the largest eigenvalue $P$ of $\Ktt_i{}^j$ at each point in $\Sigma$ in the limit $t\searrow 0$. The implicit constants here depend on the spacetime, the value of $c_s^2$ and the fluid initial data, but in particular not on $T_0$.
\end{thm}

The proof of Theorem~\ref{thm:Euler1} is discussed in Section~\ref{sec:proof1}, preceded by a detailed derivation in Sections~\ref{sec:variablescont.N} and \ref{sec:FuchsianPDE.N} of PDE systems equivalent to the Euler equations which are amenable to the central global-in-time Fuchsian analysis of our proof. Before we proceed with this discussion, however, we make a number of comments and remarks about Theorem~\ref{thm:Euler1}. First and foremost, Theorem~\ref{thm:Euler1} is consistent with the heuristics in Section~\ref{sec:convvariables.N} in that it rigorously confirms that a large class of fluid solutions in the regime $0<c_s^2<P<1$ is asymptotically extremely-tilted provided the restriction \eqref{eq:theorem1restr} holds. We observe that the restriction \eqref{eq:theorem1restr} trivially implies that $0<c_s^2<P<1$. As expected from the heuristics, the fluid particles approach the speed of light relative to observers at rest as indicated in particular by the divergence of the relativistic Lorentz factor $\Gamma$. The tilt vector converges to a unit vector in the eigenspace of the largest eigenvalue $P$ of $\Ktt_i{}^j$ while the orthogonal components decay.

Beyond the heuristic predictions, Theorem~\ref{thm:Euler1} reveals additional effects encapsulated in the restriction \eqref{eq:theorem1restr}.  As we will explain in detail in the following subsections, this restriction arises from considerations regarding a ``good'' PDE structure, and, via   the exponents $\bar p$, $q$ and $\bar q$ of Definition~\ref{def:asympKasner}, can be attributed to the rate of decay of terms in the Euler equations determined by the background spacetime.
We now show that \eqref{eq:theorem1restr} can be interpreted as a  \emph{lower limit} $c_*^2$ for the speed of sound parameter $c_s^2$  and therefore obtain a pointwise inequality of the form \eqref{eq:informalnonheuristic}.
To this end, fix an arbitrary point $x$ in $\Sigma$ and define the number $\Ltt_*=\min\{1,\bar p, q, \bar q\}$ there; this means that $\Ltt_*\in (0,1)$. Then, given the value $P$ at $x$, which must also be in $(0,1)$, the inequality \eqref{eq:theorem1restr} can be rearranged as an  inequality for $c_s^2$ at $x$:
\begin{equation}
    \label{eq:theorem1restr.2}
    \frac{P-\Ltt_*}{1-\Ltt_*}<c_s^2<P.
\end{equation}
Thus, in agreement with \eqref{eq:informalnonheuristic} we set
\begin{equation}
    \label{eq:theorem1restr.2.cstar}
    c_*^2=\max\bigl\{0,\frac{P-\Ltt_*}{1-\Ltt_*}\bigr\}.
\end{equation}
Because $P$ and $\Ltt_*$ are both numbers in $(0,1)$
it independently follows that $\frac{P-\Ltt_*}{1-\Ltt_*}<P$. It is easy to check that  $\frac{P-\Ltt_*}{1-\Ltt_*}$ is a monotonically decreasing function of $\Ltt_*$ and hence, the larger $\Ltt_*$ is the smaller is $c_*^2$. 
If we also take into account the condition $\bar p=1-P$ found at the end of Section~\ref{sec:background.N}, we have $\Ltt_*\le \bar p=1-P$. Given the before mentioned  monotonicity of the left side of \eqref{eq:theorem1restr.2}, we obtain a necessary (but in general not sufficient) restriction for $c_s^2$ by replacing  $\Ltt_*$ by $1-P$ in \eqref{eq:theorem1restr.2}: 
\begin{equation}
    \label{eq:theorem1restr.3}
    \frac{2P-1}P<c_s^2<P.
\end{equation}
Here we observe that the inequality $\frac{2P-1}P<P$ holds independently with $P$ in the range $(0,1)$.
In any case, it is conceivable (but not implied by the theorem) that the dynamics of the solutions differs significantly from the asymptotically extremely-tilted heuristics if \eqref{eq:theorem1restr.2} is violated.

We would further like to emphasise the following distinctive property of Theorem~\ref{thm:Euler1}, namely that it makes no ``small data'' restrictions on the fluid initial data and, especially, no restriction on the size of their spatial derivatives. In our proof, this so-called \emph{reference independence} is achieved by applying a standard \emph{small-data} global existence Fuchsian theorem to the Euler equations once the main singular leading-order fluid contributions have been subtracted off the fluid variables. These leading-order contributions are defined as solutions of so-called \emph{truncated Euler equations}, defined by only keeping those terms from the full Euler equations that we expect to dominate the asymptotics near $t=0$. In our case here, these truncated equations turn out to be so simple that they can be solved explicitly, and it is easy to see that resulting error terms for the full Euler equations can be controlled by adjusting the size of $T_0$. This is why Theorem~\ref{thm:Euler1} hypothesises $T_0$ to be small. All of this is in contrast to more standard nonlinear stability approaches where instead some given reference solution (often a spatially homogeneous solution) is subtracted off the variables. The advantage of such a \emph{reference-dependent} approach is that it is (at least in principle) not necessary to assume that $T_0$ is small to control the size of error terms. The obvious disadvantage is, however, that the class of solutions covered by a corresponding theorem is restricted to some near-reference setting. In any case, for completeness we note that a theorem analogous to Theorem~\ref{thm:Euler1} can be established in the near-spatially homogeneous setting where $T_0$ is allowed to be arbitrary large in principle. All of this makes sense: shock formation is a phenomenon expected for most fluid solutions on cosmological spacetimes, and the only conceivable ways to ``delay shock for long enough'' to reach the big bang singularity are to either assume that the solutions are ``sufficiently spatially homogeneous'' or that the initial time $T_0$ is sufficiently close to the final time $t=0$. 

As one last remark, we point out that the estimates \eqref{eq:thm1.rhoestimate} -- \eqref{eq:thm1.tiltestimate} could be improved in the sense that one could be a bit more precise about the constant $\kappa$. However, this would in general amount to further restrictions regarding the Kasner exponents as well as for the decay exponents  $q$, $\bar q$ and $\bar p$ in Definition~\ref{def:asympKasner}.

\subsection{Euler PDE systems for the asymptotically extremely-tilted case}
\label{sec:variablescont.N}

In this subsection we now continue the discussion we started in Section~\ref{sec:convvariables.N} with the view to deriving ``good'' PDEs to prove Theorem~\ref{thm:Euler1}. To this end 
we continue to assume that a spacetime with Kasner big bang asymptotics is given (see Definition~\ref{def:asympKasner}), and specialise to the asymptotically extremely-tilted regime $0<c_s^2<P<1$; recall that $P$ is in general a function on $\Sigma$, and hence we impose this inequality at each point of $\Sigma$.
As motived by the heuristics in the previous subsection we now make the choice
\begin{equation}
    \label{eq:parameterchoices}
    \Ptt=P \quad\text{and}\quad
    \Ltt=(P-c_s^2)/(1-c_s^2)
\end{equation} 
at each point of $\Sigma$.
Motivated by the heuristics discussed in Section~\ref{sec:convvariables.N} we set up our variables based on the intuition that $\Vtm$ and $\hat\Vsp_i=\hat h_i{}^j \Vsp_j$ converge to non-vanishing limits and that $\check\Vsp_i=\check h_i{}^j u_j$ decays to zero at $t=0$. To this end, in order to be keep track of the different asymptotics of the two parts $\hat\Vsp_i$ and $\check\Vsp_i$ of $\Vsp_i$, it is useful to make a variable transformation (in the sense of Section~\ref{sec:nonlineartransf}) that defines these as independent variables:
\begin{equation}
    \label{U2X}
    X=\begin{pmatrix}
        u\\
        u_j
    \end{pmatrix}
    =S U,\quad
    S= \begin{pmatrix}
    1 & 0 & 0\\
    0 & \delta_j{}^k & \delta_j{}^k
    \end{pmatrix},\quad 
    U=\begin{pmatrix}
        u\\
        \hat u_k\\
        \check u_k
    \end{pmatrix}.
\end{equation} 
The map
\begin{equation}
    \label{X2U}
    U=\begin{pmatrix}
        u\\
        \hat u_j\\
        \check u_j
    \end{pmatrix}
    =T X,\quad
    T=\begin{pmatrix}
        1 & 0\\
        0 & \hat h_j{}^k\\
        0 & \check h_j{}^k
    \end{pmatrix},
    \quad
    X=
    \begin{pmatrix}
        u\\
        u_k
    \end{pmatrix},
\end{equation} 
is the inverse of \eqref{U2X} if and only if 
\begin{equation}
    \label{eq:Uconsist}
    \check h_i{}^k\hat u_k=0,\quad \hat h_i{}^k\check u_k=0.
\end{equation}

In the following it is convenient to allow covectors $\hat u_k$ and $\check u_k$ that possibly violate \eqref{eq:Uconsist}. This avoids difficulties caused by the fact that the projection map $\hat h_i{}^j$ and its orthogonal complement $\check h_{i}{}^j$ are in general space-dependent, which clashes with some of the hypotheses of the Fuchsian theorem used to prove Theorem~\ref{thm:Euler1}. When we treat $\hat u_k$ and $\check u_k$ as independent, the dimension of the unknown variable vector is ``enlarged'': the dimension of $U$ exceeds that of the original unknown $X$.
As a consequence, \eqref{U2X} may fail to be invertible, and the evolution equations \eqref{eq:PDEgen.transf.nonsymm.N} for $X$ do in general not determine the full vector $U$ uniquely unless the \emph{consistency condition} \eqref{eq:Uconsist} holds. One of the main points of Section~\ref{sec:nonlineartransf} is how to deal with precisely this situation; cf.\ especially the discussion of \eqref{eq:VariableTrafo.abstr.ext.1} -- \eqref{eq:VariableTrafo.abstr.ext.2}. To be more precise
we now interpret $U$ as a $\Rbb^{1+2N}$-valued function of the form 
\[U=\begin{pmatrix}
        u\\
        \hat u_j\\
        \check u_j
    \end{pmatrix},\]
related to the previous $\Rbb^{1+N}$-valued function $X$ in \eqref{U2X} by the variable transformation \eqref{U2X}.  In the language of Section~\ref{sec:nonlineartransf}, this variable transformation implies a submersion with $s=1+2N$, $\bar s=\sigma=1+N$. The kernel of its Jacobian $S$ is an $N$ dimensional subspace of $\Rbb^s$, and the map
\begin{equation}
    \label{eq:piU}
    \pi=\mathrm{diag}\left(1,\hat h_j{}^k,\check h_j{}^k\right)
\end{equation}
projects $\Rbb^{s}$-vectors onto the complement of the kernel in $\Rbb^{s}$. Defining a $\Rbb^N$-valued map $\mathbf F$ as
\begin{equation}
    \mathbf F(t,x,U)=(\check h_i{}^k\hat u_k, \hat h_i{}^k\check u_k),
\end{equation}  
we easily verify the kernel condition \eqref{eq:kernelproperty} and identify the consistency condition \eqref{eq:Uconsist} with \eqref{eq:VariableTrafo.abstr.ext.2}. 
In the following we use the shorthand notation (irrespective of whether \eqref{eq:Uconsist} is violated or not)
\begin{equation}
    \label{eq:ushorthand.Pre}
    u_i=\hat u_i+\check u_i, \quad |u|_h^2=u_i u^i.
\end{equation}
Following the discussion in Section~\ref{sec:nonlineartransf} (especially \eqref{eq:transfcoeff.1} -- \eqref{eq:transfcoeff.3}) to define a consistent enlarged PDE system from \eqref{eq:PDEgen.transf.nonsymm.N} for the enlarged variable vector $U$ that preserves the consistency condition \eqref{eq:Uconsist}, we may choose an arbitrary constant $\eta>0$ and find
\renewcommand{\NewVar}{Z}
\newcommand{\Vnm}{z}
\begin{equation}
    \partial_t U
    +\alpha \Coeff^{l} 
    \Dc_l{} U
    =\frac 1t\Gsc \Pi U+\frac 1t\mathtt G_2 
    +t^{-1+Q_1}\mathtt G, 
    \label{eq:Uevol}
\end{equation}
where
\begin{align}
&\Btt^l 
=T \mathtt b^{l} S\label{eq:BlU}\\
=&\frac{1}{\sqrt{t^{2\Ltt}\Vtm^2+\Vspnorm^2}}\left(
    \begin{smallmatrix}
    \frac{(1-c_s^2)\Vspnorm^2+(1-2c_s^2)t^{2\Ltt}\Vtm^2}{(1-c_s^2)\Vspnorm^2+t^{2\Ltt}\Vtm^2}\Vsp^l 
    & \frac{c_s^2\Vtm}{(1-c_s^2)\Vspnorm^2+t^{2\Ltt}\Vtm^2} (\hat {P}^{k l} + \hat\Vsp^k \check\Vsp^l)
    & \frac{c_s^2\Vtm}{(1-c_s^2)\Vspnorm^2+t^{2\Ltt}\Vtm^2} (\check {P}^{k l} + \check\Vsp^k \hat\Vsp^l)\\
    t^{2\Ltt}\Vtm \hat {h}_{j}{}^{l} & \hat {h}_{j}{}^{k} \Vsp^l & 0 \\
    t^{2\Ltt}\Vtm \check {h}_{j}{}^{l} & 0& \check {h}_{j}{}^{k} \Vsp^l & 
\end{smallmatrix}
\right)\notag,
\end{align}
\begin{align}
\Gsc=&\begin{pmatrix}
    1
    & 0 & 0\\
    0 & h_i{}^j & 0\\
    0 & 0 & \check L_j{}^k+\eta\hat h_j{}^k
\end{pmatrix}\label{eq:Gscdef},\\
\Pi=&\begin{pmatrix}
    0 & 0 &0\\
    0 & 0 &0\\
    0 & 0 &h_j{}^k
\end{pmatrix}\label{eq:Pi},\\
\mathtt G_2 
    =&\begin{pmatrix}
    -\frac{c_s^2 \check L^{ll'}\check\Vsp_l\check\Vsp_{l'}}{t^{2\Ltt}\Vtm^2+(1-c_s^2) \Vspnorm^2}\Vtm\\
    0\\
    0
\end{pmatrix}.\label{eq:G2def}
\end{align}
A lengthy calculation would allow us to construct the function $t^{-1+Q_1}\mathtt G$ explicitly from \eqref{eq:PDEgen.transf.nonsymm.N} with \eqref{eq:variables1.bL.NNN} -- \eqref{eq:defP} together with \eqref{eq:transfcoeff.1} -- \eqref{eq:transfcoeff.3}. For the purposes of the paper only the following two observations are relevant: First, we have
\begin{equation}
    \label{eq:piperpG}
    \pi^\perp \mathtt G=0;
\end{equation}
cf.\ \eqref{eq:piU}.
Second, given any smooth function $Q_1:\Sigma\rightarrow\Rbb$ which satisfies the bound
\begin{equation}
    \label{eq:Q1inequal.1}
    0<Q_1<\min\{\bar p,q,\bar q,2\Ltt\}
\end{equation}
at each point of $\Sigma$, the function $\mathtt G$ in \eqref{eq:Uevol}
can be identified with a uniformly bounded polynomial in the sense of Section~\ref{sec:analysisbg.N} for
\begin{equation}
    \chi(t,x,U)=\bigl(u, \hat u_i, \check u_i, \sqrt{u^2 t^{2 \Ltt}+u_i u^i}, 1/\sqrt{u^2 t^{2 \Ltt}+u_i u^i}, 1/(u^2 t^{2 \Ltt}+(1-c_s^2)u_i u^i)\bigr)^T,
\end{equation}
and 
\begin{equation}
    \Omega=\left\{U\in\Rbb^{1+2N} \,\left|\, u>u_*, \,\hat u_{i}\hat u^i>\hat u_{*i}\hat u_{*}^i\right.\right\}
\end{equation}
given any constant $u_*>0$ and $\hat u_{*i}$ in $\Rbb^{N}$ with $\hat u_{*i}\hat u_{*}^i>0$. 
This observation follows directly from the assumptions in Definition~\ref{def:asympKasner}, especially:
\begin{enumerate}[(i)]
    \item the quantity $\bar p$ in \eqref{eq:Q1inequal.1} arises from \eqref{eq:framebound} and from the terms in \eqref{eq:variables1.g.NNN} proportional to $\Dc_p P$ and $\Dc_p\Ltt$,
    \item the quantity $q$ comes from writing $k_{ij}$ in \eqref{eq:variables1.g.NNN} as in \eqref{eq:defK}, \eqref{eq:Kttdef} and \eqref{eq:CMCgauge}, together with the $\Omega_j{}^m$-term in \eqref{eq:variables1.g.NNN},
    \item $\bar q$ arises from the $C_l{}^m{}_j$-term in \eqref{eq:variables1.g.NNN}, cf.\ \eqref{eq:timeindepspatconn.1} -- \eqref{eq:timeindepspatconn.3},
    \item $2\Ltt$ appears in \eqref{eq:Q1inequal.1} because we absorb all terms proportional to $t^{-1+2\Ltt}$ in \eqref{eq:variables1.g.NNN} into the function $\mathtt G$.
\end{enumerate}
These arguments also use \eqref{eq:alphaasympt}, \eqref{eq:ndotbound} and \eqref{eq:Kttspecdecomp}.

We emphasise that following Section~\ref{sec:nonlineartransf}, especially \eqref{eq:transfcoeff.1} -- \eqref{eq:transfcoeff.3}, we introduce the positive term $\eta\,\hat h_j{}^k$ in the definition of \eqref{eq:Gscdef}, and, as consequence especially of \eqref{eq:piperpG}, we have
\begin{equation}
    \label{eq:PDEgen.transf.nonsymm.NN.off.2}
    \partial_t (\check h_j{}^k \hat u_k)=0, \quad \partial_t (\hat h_j{}^k \check u_k)=\frac 1t \eta \hat h_j{}^k \check u_k.
\end{equation}
By this, the kernel condition \eqref{eq:Uconsist}
is preserved by \eqref{eq:Uevol}, and any solution $U$ of \eqref{eq:Uevol} which satisfies \eqref{eq:Uconsist} at the initial time therefore yields an equivalent solution $X$ of \eqref{eq:PDEgen.transf.nonsymm.N} via \eqref{U2X}. Vice versa any solution $X$ of \eqref{eq:PDEgen.transf.nonsymm.N} yields an equivalent solution $U$ of \eqref{eq:Uevol} via \eqref{X2U} which satisfies \eqref{eq:Uconsist}. Even more, however, by the positivity of $\eta$, the quantity $\hat h_j{}^k \check u_k$ decays towards $t=0$, and in fact, by the positivity of $\check L_j{}^k+\eta\hat h_j{}^k$, the full independent variable $\check u_i$ is a \emph{decaying variable} (in the language used by the Fuchsian method below).

We have written \eqref{eq:Uevol} in non-symmetrised form for simplicity, but note that
the symmetriser for this system derived from \eqref{eq:variables1.b0.N} following Section~\ref{sec:nonlineartransf} is
\begin{equation}
    \begin{pmatrix}
        \frac{t^{2\Ltt}\Vtm^2+(1 - c_s^2) \Vspnorm^2}{c_s^2} & 0 & 0\\
        0 & \hat h^{jk}+t^{-2\Ltt } \hat h_{j'}{}^{j}P^{j'k'} \hat h_{k'}{}^{k} & t^{-2\Ltt } \hat h_{j'}{}^{j}P^{j'k'} h_{k'}{}^{k}\\
        0 & t^{-2\Ltt } \check h_{j'}{}^{j}P^{j'k'} \hat h_{k'}{}^{k} & \check h^{jk}+t^{-2\Ltt } \check h_{j'}{}^{j}P^{j'k'} h_{k'}{}^{k}
    \end{pmatrix}.
\end{equation}
This symmetriser inherits the all the properties from \eqref{eq:variables1.b0.N} -- especially its problems. The variable vector $U$ is therefore not suitable for our global-in-time analysis.

We next consider a transformation to yet another variable vector $Z$. This new variable vector $Z$ is defined via the following variable transformation from the original vector $X$ defined in \eqref{eq:Xdef}
\begin{equation}
    \label{Z2U}
    X=\tilde\Phi (t,x,\NewVar)=\begin{pmatrix}
        \Vtm\\
        t^{\Ltt} z z_j 
    \end{pmatrix},\quad \NewVar=\begin{pmatrix}
        \Vtm\\\Vnm\\{\Vnm}_k
    \end{pmatrix},
\end{equation}
where $\Ltt$ is the function given by \eqref{eq:parameterchoices}. We restrict this transformation to the case $u,z>0$ as we explain below. 
The map
\begin{equation}
    \label{U2Z}
    Z=\tilde\Phi^{-1} (t,x,X)=\begin{pmatrix}
        \Vtm\\
        |u|_h\\
        t^{-\Ltt} \frac{u_i}{|u|_h}
    \end{pmatrix},
\end{equation}
is the inverse if
\begin{equation}
    \label{eq:znormalisationconstr}
    z_k z^k-t^{-2\Ltt}=0.
\end{equation} 

However, for similar reasons as discussed for the transformation between $X$ and $U$ above, it is convenient to generally allow $z_k$ to violate \eqref{eq:znormalisationconstr}, and, in accordance with the discussion in Section~\ref{sec:nonlineartransf} to consider the variable transformation \eqref{Z2U} in the submersion case with $s=2+N$ and $\sigma=\bar s=1+N$ with \eqref{eq:znormalisationconstr} represented by the $\Rbb$-valued map 
\begin{equation}
    \mathbf F(t,x,Z)=z_k z^k-t^{-2\Ltt},
\end{equation} 
which allows us to identify \eqref{eq:znormalisationconstr} with the consistency condition \eqref{eq:VariableTrafo.abstr.ext.2}. We can think of $Z$ again as ``enlarged'' variable vector. It is straightforward to verify the kernel property \eqref{eq:kernelproperty}. Following the discussion in Section~\ref{sec:nonlineartransf} (especially \eqref{eq:transfcoeff.1} -- \eqref{eq:transfcoeff.3}) to define a consistent enlarged PDE system from \eqref{eq:PDEgen.transf.nonsymm.N} for the enlarged variable vector $Z$ that preserves the consistency condition \eqref{eq:znormalisationconstr}
exploiting the fact that $\ker \mathbf D\tilde\Phi$ is one dimensional, one can find the following surprisingly simple PDE system, which is well-defined for $u,z>0$, written in non-symmetrised form as
\begin{align}
    \partial_t Z&
    +\alpha \Ctt^{l}(t,x,Z)\Dc_l{} Z
    =\frac 1t\Hsc(t,x,Z) Z+t^{-1+Q_2}\Htt(t,x,Z),\label{eq:Zevol}\\
    \Ctt^l(t,x,Z)
    &=\frac 1{\sqrt{u^2 t^{2 \Ltt}+z^2}}\left(
    \begin{smallmatrix}
    \frac{  (1 -2 c_s^2) u^2 t^{2 \Ltt}+(1-c_2^2) z^2}{u^2 t^{2 \Ltt}+(1-c_s^2)z^2} t^{\Ltt} z z^l
    & \frac{c_s^2 u^3 t^{3 \Ltt} z^l}{u^2 t^{2 \Ltt}+(1-c_s^2)z^2} 
    & \frac{c_s^2 t^{\Ltt} u z   (u^2 t^{2 \Ltt}+z^2) (h^{lk}-t^{2\Ltt}z^l z^k)}{u^2 t^{2 \Ltt}+(1-c_s^2)z^2}  \\
    u t^{3 \Ltt} z^l  &  t^{\Ltt} z z^l  & 0 \\
    \frac{u t^{\Ltt} \bigl(h_j{}^l-t^{2\Ltt} z_j z^l\bigr)}{z } & 0 & {z t^{\Ltt} z^l h_j{}^k} \\
    \end{smallmatrix}
    \right),\label{eq:Ctt}\\
    \Hsc(t,x,Z)&=\frac {(1-c_s^2)z^2}{u^2 t^{2\Ltt}+(1-c_s^2)z^2}
    \begin{pmatrix}
    0 & 0 & -t^{2 \Ltt} \frac{c_s^2}{1-c_s^2}    u {\check L}^{kq}z_{q} \\
    0 & 0 & t^{2 \Ltt}    z {\check L}^{kq} z_{q} \\
    0 & 0 & -\Ltt h_{j}{}^{k}+{\check L}_{j}{}^{k}
    -t^{2 \Ltt}\frac{ h_{j}{}^{k} \left( \Ltt u^2+(1-c_s^2) z^2 {\check L}^{rq} z_{q} z_{r}\right)-u^2 {\check L}_{j}{}^{k}}{(1-c_s^2) z^2}
    \end{pmatrix}\label{eq:Hsc}.
\end{align}
It is crucial for the subsequent analysis to impose conditions for the function $Q_2$ which allow us to
identify $\Htt$ with a uniformly bounded polynomial in the sense of Section~\ref{sec:analysisbg.N} 
for 
\begin{equation}
    \label{eq:chiZ}
    \chi(t,x,Z)=\bigl(u, z, z^{-1}, t^{\Ltt} z z_i, \sqrt{u^2 t^{2 \Ltt}+z^2}, (u^2 t^{2\Ltt}+(1-c_s^2)z^2)^{-1}\bigr)^T,
\end{equation}
and 
\begin{equation}
    \Omega=\left\{Z\in\Rbb^{2+N} \,\left|\, u>u_*, z>z_*\right.\right\}
\end{equation}
given any constants $z_*,u_*>0$. A straightforward (but lengthy) analysis reveals that this is the case if $Q_2:\Sigma\rightarrow\Rbb$ is any smooth function and satisfies the bound
\begin{equation}
    \label{eq:Q2inequal.1}
    0<Q_2<\min\{\Ltt,\bar p, q, \bar p+q-\Ltt, \bar p-\Ltt, q-\Ltt, \bar q-\Ltt\}
    =\min\{\Ltt, \bar p-\Ltt,q-\Ltt, \bar q-\Ltt\},
\end{equation}
everywhere on $\Sigma$. Here, in particular:
\begin{enumerate}[(i)]
    \item the quantity $\bar p$ arises n \eqref{eq:Q2inequal.1} from \eqref{eq:framebound} and from the terms proportional to $\Dc_p P$ and $\Dc_p\Ltt$,
    \item $q$ from expressing $k_{ij}$ via \eqref{eq:defK}, \eqref{eq:Kttdef} and \eqref{eq:CMCgauge},
    \item $q-\Ltt$ from the $\Omega_j{}^m$-term. Similarly, $\bar q-\Ltt$ comes from the $C_l{}^m{}_j$-term, cf.\ \eqref{eq:timeindepspatconn.1} -- \eqref{eq:timeindepspatconn.3},
    \item $\bar p+q-\Ltt$ from a term proportional to $t^{-\Ltt}\dot n_i$,
    \item $\bar p-\Ltt$ from a term proportional to $z^l\Dc_l P$,
    \item $\Ltt$ appears in \eqref{eq:Q2inequal.1} because we absorb all terms proportional to $t^{-1+\Ltt}$ (and $t$-powers larger than $-1+\Ltt$) into the function $\mathtt H$.
\end{enumerate}
As before, these conclusions also make use of \eqref{eq:alphaasympt}, \eqref{eq:ndotbound} and \eqref{eq:Kttspecdecomp}.

A particularly important property of this PDE system for $Z$ is that, as long as $u,z$ are bounded away from zero, these equations are uniformly symmetrisable with \emph{diagonal} symmetriser
\begin{equation}
    \label{eq:Zsymmetriser}
    \Stt=\begin{pmatrix}
        \frac{1-c_s^2}{c_s^2}z^2+\frac{t^{2 \Ltt} u^2 }{c_s^2} & 0 & 0 \\
        0 & u^2 & 0 \\
        0 & 0 &   \left(u^2 t^{2 \Ltt}+z^2\right)z^2 h^{pj}
    \end{pmatrix},
\end{equation}
for all $t\in (0,T_0]$ and $x\in\Sigma$.
This means that this system has a well-posed local-in-time initial value problem. The main problem with \eqref{eq:Zevol}, however, is the negative definite term
$-\Ltt h_j{}^k$ on the diagonal of \eqref{eq:Hsc}. It is useful to keep in mind that $\Ltt$, as given in \eqref{eq:parameterchoices}, is always larger than $0$ (as a  consequence of the asymptotically extremely-tilted condition $P>c_s^2$), but always smaller than $1$ (as consequence of our global restriction $P<1$).

In any case, it is a consequence of Section~\ref{sec:nonlineartransf} that the consistency condition \eqref{eq:znormalisationconstr} is preserved by this system, and that the $Z$- and $X$-systems are therefore equivalent for all solutions satisfying \eqref{eq:znormalisationconstr} initially.

\subsection{The main Fuchsian Euler PDE system and its local Cauchy problem}
\label{sec:FuchsianPDE.N}

In this subsection we now construct our final system of PDEs building on the results in Section~\ref{sec:variablescont.N} which will be of Fuchsian form and will be suitable for the Fuchsian analysis performed as part of the proof of Theorem~\ref{thm:Euler1}.  The result of this subsection is summarised in the following proposition. In all of what follows we continue to use the shorthand notation \eqref{eq:ushorthand.Pre} (even when the consistency condition \eqref{eq:Uconsist} is violated), and, given any map $M(t,x,Z)$, we write $M(t,x,U)$ as a shorthand for $M(t,x,\tilde\Phi^{-1}(t,x,S U))$ (even when the consistency condition \eqref{eq:znormalisationconstr} is violated); cf.\ \eqref{U2X} and \eqref{U2Z}.

\begin{prop}
    \label{prop:Fuchsiansystem}
    Pick arbitrary constants $T_0,\epsilon,\eta>0$ and $c_s^2\in (0,1)$ as well as $u_*>0$ and $\hat u_{*i}$ in $\Rbb^{N}$ with $\hat u_{*i}\hat u_{*}^i>0$. Let $((0,T_0]\times\Sigma,g_{\mu\nu})$ be a spacetime with Kasner big bang asymptotics as in Definition~\ref{def:asympKasner} such that 
    \begin{equation}
        \label{eq:lem.restr}
        0<\min\{\Ltt,\bar p-\Ltt-\epsilon,q-\Ltt, \bar q-\Ltt\}
    \end{equation}
    at each point of $\Sigma$. 
    
    Pick arbitrary functions $u_\mathbf{0}>u_*$ and $u_{\mathbf{0},i}$ in $H^{k}(\Sigma)$ with $k>N/2+2$ and 
    \[\hat h^{ij} u_{\mathbf{0},i} u_{\mathbf{0},j}>\hat u_{*i}\hat u_{*}^i\]
    for all $x\in\Sigma$.
    Set
    \begin{align}
        \label{eq:X0}
        X_{\mathbf{0}}&=\begin{pmatrix}
            u_\mathbf{0}\\
            u_{\mathbf{0},i}
        \end{pmatrix},\\
        \label{eq:W0fromX0}
        W_{\mathbf{0}}&=\begin{pmatrix}
            \sqrt{h^{ij} u_{\mathbf{0},i} u_{\mathbf{0},j}}\\
            u_\mathbf{0}\\
            \hat h_i{}^j u_{\mathbf{0},j}\\
            \check h_i{}^j u_{\mathbf{0},j}\\
            T_0^{\Ltt+\epsilon}b_{\bar i}{}^\Lambda\Dc_\Lambda u\\
            T_0^{\Ltt+\epsilon}b_{\bar i}{}^\Lambda\Dc_\Lambda |u|_h\\
            T_0^{\Ltt+\epsilon}b_{\bar i}{}^\Lambda\Dc_\Lambda (T_0^{-\Ltt} \frac{u_i}{|u|_h})
        \end{pmatrix}.
    \end{align}
    There is $T_1\in [0,T_0)$ and a unique classical solution $X\in C^0((T_1,T_0],H^{k}(\Sigma))\cap C^1((T_1,T_0],H^{k-1}(\Sigma))$ of the Cauchy problem of \eqref{eq:PDEgen.transf.nonsymm.N} launched from initial data $X_{\mathbf{0}}$ at $t=T_0$ such that 
    \begin{enumerate}[(i)]
        \item The function 
    \begin{equation} 
        \label{eq:WfromX}
        W=\begin{pmatrix}
            \sqrt{h^{ij} u_i u_j}\\
            T X\\
            t^{\Ltt+\epsilon}b_{\bar i}{}^\Lambda\Dc_\Lambda(\tilde\Phi^{-1} (t,x,X))
        \end{pmatrix}
    \end{equation}
    in $C^0((T_1,T_0],H^{k-1}(\Sigma))\cap C^1((T_1,T_0],H^{k-2}(\Sigma))$
    defined from the solution $X$ agrees with $W_{\mathbf{0}}$ in \eqref{eq:W0fromX0} at $t=T_0$ and is a solution of the following PDE system
    \begin{align}
        A^0(t,x,W)&\partial_t W
        + \alpha t^{-1+\bar p} \tau^{\bar l}_l b_{\bar l}{}^\Omega A^{l}(t,x,W) \Dc_\Omega W\label{eq:Wevol}\\
        = &t^{-1}\Asc(t,x,W)\Pbb W+t^{-1}\Fsc_2(t,x,W)+t^{-1+Q}\Fsc(t,x,W),\notag\\
        W=&\begin{pmatrix}
            z&
            U &
            {\mathbf Z}_{\bar l}
        \end{pmatrix}^T,\label{eq:Wdef}\\
        \Pbb=&\mathrm{diag}(0,\Pi, \mathrm{id}),\label{WPbb}\\
        A^0(t,x,W)=&\mathrm{diag}\Bigl(1,\mathrm{id}, \mathrm{diag}\Bigl(\frac{1-c_s^2}{c_s^2}z^2+\frac{t^{2 \Ltt} u^2 }{c_s^2}, u^2,\left(u^2 t^{2 \Ltt}+z^2\right)z^2 h^{pj}\Bigr)\Bigr),\label{WA0}\\
        A^{l}(t,x,W)=&\mathrm{diag}\left(0,0,\frac 1{\sqrt{u^2 t^{2 \Ltt}+z^2}}\left(
    \begin{smallmatrix}
    \frac{  (1 -2 c_s^2) u^2 t^{2 \Ltt}+(1-c_2^2) z^2}{c_s^2} u^l
    & z^{-1}{ u^3 t^{2 \Ltt} u^l}
    & { t^{\Ltt} u   z^{-1}(u^2 t^{2 \Ltt}+z^2) (z^2 h^{lk}-u^l u^k)}  \\
    u^3 z^{-1}t^{2 \Ltt} u^l  &  u^2 u^l  & 0 \\
    \left(u^2 t^{2 \Ltt}+z^2\right)z^{-1} {u t^{\Ltt} \bigl(z^2 h^{pl}-u^p u^l\bigr)} & 0 & \left(u^2 t^{2 \Ltt}+z^2\right)z^2 h^{pj}{u^l h_j{}^k} \\
    \end{smallmatrix}
    \right)\right),\label{WAl}\\
        \Asc(t,x,W)=&\mathrm{diag}\Bigl(0,\Gsc(t,x), \mathrm{diag}\,\Bigl(\frac{1-c_s^2}{c_s^2}z^2(\Ltt+\epsilon), u^2(\Ltt+\epsilon),\epsilon z^4 h^{pk}
        +z^4{\check L}^{pk}\Bigr)\Bigr),\label{WAsc}\\
        \Fsc_2(t,x,W)&\label{WFTwo}\\
        &\hspace{-8ex}=\begin{pmatrix}
            \frac {(1-c_s^2)z {\check L}^{kq} \check u_{q} \check u_k}{u^2 t^{2\Ltt}+(1-c_s^2)z^2} &
        \left(\begin{smallmatrix}
        -\frac{c_s^2 \check L^{ll'}\check\Vsp_l\check\Vsp_{l'}}{t^{2\Ltt}\Vtm^2+(1-c_s^2) \Vspnorm^2}\Vtm\\
        0\\
        0
        \end{smallmatrix}\right) & - \left(\begin{smallmatrix}
            {\check L}^{pq}{\check u}_{q} {\check u}_p & 0 & 0 \\
        0 &  -u^2 |u|_h^{-2}   {\check L}^{pq} {\check u}_{q} {\check u}_p&  0 \\
        0 & 0 & 
        |u|_h^2{   {\check L}^{rq} {\check u}_{q} {\check u}_{r}} h^{pk}
        +{  2   }{ |u|_h^2}u^p {\check L}^{qk} {\check u}_{q}
        \end{smallmatrix}\right){\mathbf Z}_{\bar i}
        \end{pmatrix}^T,\notag 
    \end{align}
    for all $(t,x)\in (T_1,T_0]\times\Sigma$, where $z$, $U$ and ${\mathbf Z}_{\bar l}$ refer to the respective components of $W$ according to \eqref{eq:Wevol}, where the components of $U$ are labelled as in \eqref{U2X} and where we make use of the shorthand notation in \eqref{eq:ushorthand.Pre}. The map
    $\Pi$ is defined in \eqref{eq:Pi} and the map $\Gsc(t,x)$ in \eqref{eq:Gscdef}. 
    Given any smooth function $Q:\Sigma\rightarrow\Rbb$ in \eqref{eq:Wevol} satisfying
    \begin{equation}
        \label{eq:Qinequal.1}
        0<Q<\min\{\Ltt,\bar p-\Ltt-\epsilon,q-\Ltt, \bar q-\Ltt\}
    \end{equation}
    for each point on $\Sigma$, the map $\Fsc$ in \eqref{eq:Wevol} 
    can be identified with a uniformly bounded polynomial in the sense of Section~\ref{sec:analysisbg.N} for 
    \begin{equation}
        \label{eq:Wchidef}
        \chi(t,x,W)=\bigl(z, z^{-1}, u, \hat u_i, \check u_i, {\mathbf Z}_{\bar i}, |u|^{-1}_h, \sqrt{u^2 t^{2 \Ltt}+u_i u^i}, 1/\sqrt{u^2 t^{2 \Ltt}+u_i u^i}, 1/(u^2 t^{2 \Ltt}+(1-c_s^2)u_i u^i)\bigr)^T,
    \end{equation}
    and 
    \begin{equation}
        \label{eq:WOmegadef}
        \Omega=\left\{W\in\Rbb^{2+4N+N^2} \,\left|\, u>u_*,\quad z>\sqrt{\hat u_{*i}\hat u_{*}^i\,\hat u_{i}},\quad \hat u^i>\hat u_{*i}\hat u_{*}^i\right.\right\}.
    \end{equation} 
    \item For all $(t,x)\in (T_1,T_0]\times \Sigma$, the value $W(t,x)$ is in $\Omega$.
    \item The $z$-, $U$- and ${\mathbf Z}_{\bar l}$-components of $W$ are members of $C^0((T_1,T_0],H^{k}(\Sigma))\cap C^1((T_1,T_0],H^{k-1}(\Sigma))$, $C^0((T_1,T_0],H^{k}(\Sigma))\cap C^1((T_1,T_0],H^{k-1}(\Sigma))$ and $C^0((T_1,T_0],H^{k-1}(\Sigma))\cap C^1((T_1,T_0],H^{k-2}(\Sigma))$, respectively, and satisfy
    \begin{equation}
        \label{eq:DUZbold}
        \Dc_{\bar l} U=:\Phi_0 (t,x,U) t^{-\Ltt-\epsilon} {\mathbf Z}_{\bar l}
        +\Phi_{\bar l}(t,x,U),\quad \Phi_{\bar l}(t,x,U)=b_{\bar l}{}^\Lambda\Phi_{\Lambda}(t,x,U)
    \end{equation}
    for all $(t,x)\in (T_1,T_0]\times\Sigma$, where 
    \begin{align}
        \label{eq:DUZbold.1}
        \Phi_0 (t,x,U)&=\begin{pmatrix}
            1 & 0 & 0\\
            0 & \frac{\hat h_{i}{}^j\hat u_j}{|u|_h} & t^{\Ltt}   |u|_h  \hat h_i{}^j\\
            0 & \frac{\check h_{i}{}^j\check u_j}{|u|_h} & t^{\Ltt}   |u|_h  \check h_i{}^j
        \end{pmatrix},\\
        \label{eq:DUZbold.2}
        \Phi_{\Lambda}(t,x,U)&=\begin{pmatrix}
            0\\
            \Dc_\Lambda \Ltt \log (t) \hat h_{i}{}^j\hat u_j + \Dc_\Lambda \hat h_i{}^j (\hat h_j{}^k\hat u_k+\check h_j{}^k\check u_k)\\
            \Dc_\Lambda \Ltt \log (t) \check h_{i}{}^j\check u_j + \Dc_\Lambda \check h_i{}^j (\hat h_j{}^k\hat u_k+\check h_j{}^k\check u_k),
        \end{pmatrix}.
    \end{align}
    \item The $U$-component of $W$  satisfies the consistency condition \eqref{eq:Uconsist} for all $(t,x)\in (T_1,T_0]\times\Sigma$.
    \item The $z$-component of $W$ satisfies the consistency condition 
    \begin{equation}
        \label{eq:zconsist}
        z=|u|_h,
    \end{equation}
    for all $(t,x)\in (T_1,T_0]\times\Sigma$.
    \end{enumerate}

    Vice versa, there exists a $T_1\in [0,T_0)$ and a unique  classical solution $W$ in $C^0((T_1,T_0],H^{k-1}(\Sigma))\cap C^1((T_1,T_0],H^{k-2}(\Sigma))$ of the Cauchy problem of \eqref{eq:Wevol} launched from the initial data $W_{\mathbf{0}}$ in \eqref{eq:W0fromX0} at $t=T_0$ with values in $\Omega$. The $z$-, $U$- and ${\mathbf Z}_{\bar l}$-components of $W$ are members of  
    $C^0((T_1,T_0],H^{k}(\Sigma))\cap C^1((T_1,T_0],H^{k-1}(\Sigma))$, $C^0((T_1,T_0],H^{k}(\Sigma))\cap C^1((T_1,T_0],H^{k-1}(\Sigma))$ and $C^0((T_1,T_0],H^{k-1}(\Sigma))\cap C^1((T_1,T_0],H^{k-2}(\Sigma))$, respectively, and
    satisfy \eqref{eq:DUZbold}, \eqref{eq:Uconsist} and \eqref{eq:zconsist} for all $(t,x)\in (T_1,T_0]\times\Sigma$. The function $X$ defined from $U$ via \eqref{U2X} is in $C^0((T_1,T_0],H^{k}(\Sigma))\cap C^1((T_1,T_0],H^{k-1}(\Sigma))$ and is
    a solution of  \eqref{eq:PDEgen.transf.nonsymm.N} for all $(t,x)\in (T_1,T_0]\times\Sigma$. This solution agrees with the function \eqref{eq:X0} at $t=T_0$.
\end{prop}

The rest of this subsection focusses on the proof of Proposition~\ref{prop:Fuchsiansystem}. Given a solution $W$ of \eqref{eq:Wevol}, its $U$-component -- the \emph{lower-order variable} -- carries the information about fluid and, because \eqref{eq:Uconsist} holds, allows us to construct the physical variables via \eqref{eq:Gammanu.Var2}. The ${\mathbf Z}_{\bar l}$-component of $W$ on the other hand plays the role of the \emph{top-order variable}. The $z$-component of $W$ is introduced here in order for us to deal with the special structure of the symmetriser \eqref{WA0} based on \eqref{eq:Zsymmetriser} in our reference-independent analysis in the proof of Theorem~\ref{thm:Euler1} in Section~\ref{sec:proof1}. For the proof of Proposition~\ref{prop:Fuchsiansystem} it is useful to adopt the following notation
\begin{equation} 
    \label{eq:Pscrdef}
    \Psc_z=\mathrm{diag}\,(1,0,0),\quad \Psc_U=\mathrm{diag}\,(0,\id,0),\quad \Psc_{{\mathbf Z}_{\bar l}}=\mathrm{diag}\,(0,0,\id).
\end{equation}
This means that we can refer to the $z$-, $U$- and ${\mathbf Z}_{\bar l}$-components of $W$ as $\Psc_z W$, $\Psc_U W$ and $\Psc_{{\mathbf Z}_{\bar l}}W$, respectively.

\begin{proof}[Proof of Proposition~\ref{prop:Fuchsiansystem}]
    The proof of this proposition is straightforward despite the algebraic complexity of some of the expressions we must deal with. There are two essential parts of the proof: First, establish the equivalence of the two initial value problems, for the $X$-variables on the one side and for the $W$-variables on the other side.  A crucial step for this is to actually construct the correct map $\Fsc$ in \eqref{eq:Wevol}. Existence and uniqueness of local classical solutions follows easily from the standard theory of symmetric hyperbolic PDEs. The second main part of this proof is to establish the polynomial properties stated in the proposition of this thus constructed map $\Fsc$.

    Last us start with the first part. Under the hypothesis of Proposition~\ref{prop:Fuchsiansystem}, it is a standard consequence of \cite[Proposition~1.4, Chapter~16]{taylor2011}, the previously discussed symmetrisability of \eqref{eq:PDEgen.transf.nonsymm.N} and that $X_{\mathbf{0}}\in H^k(\Sigma)$ with $k>N/2+1$ that the Cauchy problem of \eqref{eq:PDEgen.transf.nonsymm.N} launched from initial data $X_{\mathbf{0}}$ at $t=T_0$ has a unique local classical solution $X$ in $C^0((T_1,T_0],H^k(\Sigma))\cap C^1((T_1,T_0],H^{k-1}(\Sigma))$ for some $T_1\in [0,T_0)$. 
    The map $U$ defined from $X$ by \eqref{X2U} then satisfies \eqref{eq:Uconsist} and is a classical solution of \eqref{eq:Uevol} in $C^0((T_1,T_0],H^k(\Sigma))\cap C^1((T_1,T_0],H^{k-1}(\Sigma))$. In the same way the map $Z$ defined from $X$ by \eqref{U2Z} satisfies \eqref{eq:znormalisationconstr} and is a classical solution of \eqref{eq:Zevol} in $C^0((T_1,T_0],H^k(\Sigma))\cap C^1((T_1,T_0],H^{k-1}(\Sigma))$.
    Next we set 
    \begin{equation}
        \label{eq:Zbolddef}
        {\mathbf Z}_{\bar i}=t^{\Ltt+\epsilon}b_{\bar i}{}^\Lambda\Dc_\Lambda Z=t^{\Ltt+\epsilon}\Dc_{\bar i} Z,
    \end{equation}
    for the constant $\epsilon>0$ given in the theorem. This is a continuously differentiable member of $C^0((T_1,T_0],H^{k-1}(\Sigma))\cap C^1((T_1,T_0],H^{k-2}(\Sigma))$. 
    This definition involves the time-independent frame $b_{\bar i}{}^\Lambda$ introduced as part of Definition~\ref{def:asympKasner}. 
    Finally, we define the map $W$ either by \eqref{eq:WfromX} or \eqref{eq:Wdef}; both definitions for $W$ are equivalent if we identify $z$ with the second component of $Z$ which implies \eqref{eq:zconsist}. The resulting vector function $W$ is a continuously differentiable member of $C^0((T_1,T_0],H^{k-1}(\Sigma))\cap C^1((T_1,T_0],H^{k-2}(\Sigma))$. It is straightforward to show from \eqref{eq:X0} and \eqref{eq:WfromX} that $W$ agrees with $W_{\mathbf{0}}\in H^{k-1}(\Sigma)$ given in \eqref{eq:W0fromX0} at $t=T_0$. In particular this implies that $W(T_0,x)\in \Omega$ for all $x\in\Sigma$, and therefore, by continuity, we conclude that $W(t,x)\in\Omega$ for all $(t,x)\in (T_1,T_0]\times\Sigma$ if we shift $T_1$ closer to $T_0$ if necessary.
    
    Given these definitions it is not difficult to show from \eqref{eq:timeindepspatconn.1} -- \eqref{eq:timeindepspatconn.3}, \eqref{U2X}, \eqref{X2U}, \eqref{Z2U} and \eqref{U2Z} that
    \begin{align*}
        \Dc_\Lambda U=\Dc_\Lambda \begin{pmatrix}
            u\\\hat u_i\\\check u_i
        \end{pmatrix}&=\Dc_\Lambda (T \tilde\Phi(t,x,Z))= \begin{pmatrix}
            \Dc_\Lambda \Vtm\\
            \Dc_\Lambda \Ltt \log (t) t^{\Ltt} z \hat h_i{}^j z_j+t^{\Ltt} \Dc_\Lambda  z \hat h_i{}^j z_j
            +t^{\Ltt}   z \Dc_\Lambda \hat h_i{}^j z_j+t^{\Ltt}   z  \hat h_i{}^j \Dc_\Lambda z_j\\ 
            \Dc_\Lambda \Ltt \log (t) t^{\Ltt} z \check h_i{}^j z_j+t^{\Ltt} \Dc_\Lambda  z \check h_i{}^j z_j
            +t^{\Ltt}   z \Dc_\Lambda \check h_i{}^j z_j+t^{\Ltt}   z  \check h_i{}^j \Dc_\Lambda z_j
        \end{pmatrix}\\
        &= \begin{pmatrix}
            \Dc_\Lambda \Vtm\\
            \Dc_\Lambda \Ltt \log (t) \hat h_{i}{}^j\hat u_j+\Dc_\Lambda  z \frac{\hat h_{i}{}^j\hat u_j}{|u|_h} 
            +\Dc_\Lambda \hat h_i{}^j (\hat h_j{}^k\hat u_k+\check h_j{}^k\check u_k)+t^{\Ltt}   |u|_h  \hat h_i{}^j \Dc_\Lambda z_j\\ 
            \Dc_\Lambda \Ltt \log (t) \check h_{i}{}^j\check u_j+\Dc_\Lambda  z \frac{\check h_{i}{}^j\check u_j}{|u|_h}
            +\Dc_\Lambda \check h_i{}^j (\hat h_j{}^k\hat u_k+\check h_j{}^k\check u_k)+t^{\Ltt}   |u|_h  \check h_i{}^j \Dc_\Lambda z_j
        \end{pmatrix}\\
        &= \begin{pmatrix}
            1 & 0 & 0\\
            0 & \frac{\hat h_{i}{}^j\hat u_j}{|u|_h} & t^{\Ltt}   |u|_h  \hat h_i{}^j\\
            0 & \frac{\check h_{i}{}^j\check u_j}{|u|_h} & t^{\Ltt}   |u|_h  \check h_i{}^j
        \end{pmatrix} \Dc_\Lambda Z
        +\begin{pmatrix}
            0\\
            \Dc_\Lambda \Ltt \log (t) \hat h_{i}{}^j\hat u_j + \Dc_\Lambda \hat h_i{}^j (\hat h_j{}^k\hat u_k+\check h_j{}^k\check u_k)\\
            \Dc_\Lambda \Ltt \log (t) \check h_{i}{}^j\check u_j + \Dc_\Lambda \check h_i{}^j (\hat h_j{}^k\hat u_k+\check h_j{}^k\check u_k)
        \end{pmatrix}.
    \end{align*}
    Equation \eqref{eq:DUZbold} with \eqref{eq:DUZbold.1} and \eqref{eq:DUZbold.2} are then an immediate consequence of
    \eqref{eq:Zbolddef} and \eqref{eq:framebound}. We remark that we can use \eqref{eq:zconsist} here to replace $|u|_h$ by $z$.
    
    Let us now define the map $\Fsc$ so that the just defined $W$ is a solution of \eqref{eq:Wevol}. We do this by defining the three parts, $\mathscr P_z(t^{-1+Q}\Fsc(t,x,W))$, $\mathscr P_U(t^{-1+Q}\Fsc(t,x,W))$ and $\mathscr P_{{\mathbf Z}_{\bar l}}(t^{-1+Q}\Fsc(t,x,W))$, see \eqref{eq:Pscrdef}, separately.
     
    Since $Z$ is a solution of \eqref{eq:Zevol} and we have defined $z=\mathscr P_z W$  as the second component of $Z$ above, the equations \eqref{eq:Zevol} together with \eqref{eq:Zboldevol} imply that $\mathscr P_z(t^{-1+Q}\Fsc(t,x,W))$ must equal the projection on the second component of $t^{-1+Q_2}\Htt(t,x,Z)-\alpha t^{-1+\bar p-\Ltt-\epsilon} \tau^{\bar l}_l \Ctt^{l}(t,x,U){\mathbf Z}_{\bar l}$. With this definition, it is a straightforward consequence of \eqref{eq:Wevol} with \eqref{eq:Wdef} -- \eqref{WFTwo} and \eqref{eq:Zevol} with \eqref{eq:Ctt} and \eqref{eq:Hsc} 
    that $W$ solves the $\mathscr P_z$-projection of \eqref{eq:Wevol} for all $(t,x)\in (T_1,T_0]\times\Sigma$.

    In order to guarantee that $W$ also satisfies the $\mathscr P_U$-projection of \eqref{eq:Wevol} for all $(t,x)\in (T_1,T_0]\times\Sigma$ we argue correspondingly. Since $U$ is a solution of \eqref{eq:Uevol} 
    it is sufficient to define 
    \begin{equation}
        \label{eq:PFscW}
        \mathscr{P}_U(t^{-1+Q}\Fsc(t,x,W))
        :=t^{-1+Q_1}\mathtt G(t,x,U)-\alpha \Coeff^{l} t^{-1+\bar p} \tau^{\bar l}_l\Dc_{\bar l}{} U,
    \end{equation}
    from \eqref{eq:Uevol} and then express $\Dc_{\bar l}{} U$ in terms of \eqref{eq:DUZbold}. With this definition, it is a straightforward consequence of \eqref{eq:Wevol} with \eqref{eq:Wdef} -- \eqref{WFTwo} and \eqref{eq:Uevol} with \eqref{eq:BlU} -- \eqref{eq:G2def} 
    that $W$ solves the $\mathscr P_U$-projection of \eqref{eq:Wevol} for all $(t,x)\in (T_1,T_0]\times\Sigma$.

    Let us proceed with the $\mathcal P_{{\mathbf Z}_{\bar i}}$-projection of \eqref{eq:Wevol}. Given that the map $Z$ defined by \eqref{U2Z} from our local solution $X$ satisfies \eqref{eq:znormalisationconstr} and is a local classical solution of \eqref{eq:Zevol}, it is not difficult to show that \eqref{eq:Zbolddef} implies
    \begin{align*}
        \partial_t {\mathbf Z}_{\bar i}
        =&\frac{\Ltt+\epsilon} t  {\mathbf Z}_{\bar i} 
        + t^{\Ltt+\epsilon} \Dc_{\bar i} \partial_t  Z\\
        =&\frac{\Ltt+\epsilon} t  {\mathbf Z}_{\bar i} \\
        &- \alpha t^{-1+\bar p} \tau^{\bar l}_l  \dot n_{\bar i} \Ctt^{l}(t,x,Z) {\mathbf Z}_{\bar l}{}
        -\underline{\alpha t^{-1+\bar p} \tau^{\bar l}_l  \Dc_{\bar i} (\Ctt^{l}(t,x,Z)){\mathbf Z}_{\bar l}}\\
        &-\alpha  t^{-1+\bar p}\Ctt^{l}(t,x,Z) \Dc_{\bar i} \bar p \log(t) \tau^{\bar l}_l {\mathbf Z}_{\bar l}
        -\alpha  t^{-1+\bar p}\Ctt^{l}(t,x,Z) \Dc_{\bar i} \tau^{\bar l}_l {\mathbf Z}_{\bar l}
        -\alpha  t^{-1+\bar p}\Ctt^{l}(t,x,Z)  \tau^{\bar l}_l \Dc_{\bar i} b_{\bar l}{}^\Sigma \omega^{\bar m}{}_\Sigma{\mathbf Z}_{\bar m}\\
        &-\alpha \Ctt^{l}(t,x,Z)\Dc_l {\mathbf Z}_{\bar i}
        +\alpha  \Ctt^{l}(t,x,Z) \Dc_l \Ltt \log(t) {\mathbf Z}_{\bar i}
        +\alpha  \Ctt^{l}(t,x,Z)  \Dc_{l} b_{\bar i}{}^\Sigma \omega^{\bar m}{}_\Sigma{\mathbf Z}_{\bar m}\\
        &
        +\underline{t^{-1 +\Ltt+\epsilon} \Dc_{\bar i}(\Hsc(t,x,Z))  Z}
        + \frac 1t\Hsc(t,x,Z)  {\mathbf Z}_{\bar i}
        +\underline{t^{\Ltt+\epsilon} t^{-1+Q_2}\Dc_{\bar i} (\Htt(t,x,Z))},
    \end{align*}
    for all $(t,x)\in (T_1,T_0]\times\Sigma$,
    where we have used \eqref{eq:dotnalpharel} and the fact that the connection $\Dc_\Lambda$ is flat by \eqref{eq:timeindepspatconn.2}. In order to use the chain rule for the three underlined terms we use the following notation
    \begin{align}
        \Dc_{\bar i} (\Ctt^{l}(t,x,Z))&=\underbrace{(\Dc_{\bar i} \Ctt^{l})(t,x,Z)}_{=:\Ctt_{\bar i}^l(t,x,Z)}+\underbrace{\mathbf D_{ u} \Ctt^{l}(t,x,Z) \Dc_{\bar i}  u +\mathbf D_{ z} \Ctt^{l}(t,x,Z) \Dc_{\bar i}  z+\mathbf D_{ z_p} \Ctt^{l}(t,x,Z) \Dc_{\bar i}  z_p}_{=:t^{-\Ltt-\epsilon}{\mathbf Z}_{\bar i}\cdot \mathbf D\Ctt^l(t,x,Z)},\label{eq:defdCtt}\\
        \Dc_{\bar i} (\Hsc(t,x,Z))&=\underbrace{(\Dc_{\bar i} \Hsc)(t,x,Z)}_{=:\Hsc_{\bar i}(t,x,Z)}+\underbrace{\mathbf D_{ u} \Hsc(t,x,Z) \Dc_{\bar i}  u +\mathbf D_{ z} \Hsc(t,x,Z) \Dc_{\bar i}  z+\mathbf D_{ z_p} \Hsc(t,x,Z) \Dc_{\bar i}  z_p}_{=:t^{-\Ltt-\epsilon}{\mathbf Z}_{\bar i}\cdot \mathbf D\Hsc(t,x,Z)},\label{eq:defdHsc}\\
        \Dc_{\bar i} (\Htt(t,x,Z))&=\underbrace{(\Dc_{\bar i} \Htt)(t,x,Z)}_{=:\Htt_{\bar i}(t,x,Z)}+\underbrace{\mathbf D_{ u} \Htt(t,x,Z) \Dc_{\bar i}  u +\mathbf D_{ z} \Htt(t,x,Z) \Dc_{\bar i}  z+\mathbf D_{ z_p} \Htt(t,x,Z) \Dc_{\bar i}  z_p}_{=:t^{-\Ltt-\epsilon}{\mathbf Z}_{\bar i}\cdot \mathbf D\Htt(t,x,Z)}.\label{eq:defdHtt}
    \end{align}
    Together with \eqref{eq:framebound} this can be written as follows
    \begin{equation}
        \label{eq:Zboldevol}
        \partial_t  {\mathbf Z}_{\bar i}
        + \alpha t^{-1+\bar p} \tau^{\bar l}_l b_{\bar l}{}^{\Omega}\Ctt^{l}(t,x,Z) \Dc_\Omega {\mathbf Z}_{\bar i}
        =  
        t^{-1} {\mathbfscr H}(t,x,Z,{\mathbf Z}_{\bar j}) {\mathbf Z}_{\bar i}
        +t^{-1}\mathbf H_{\bar i}(t, x, Z,{\mathbf Z}_{\bar j}),
    \end{equation}
    where\footnote{We interpret ${\mathbfscr H}(t,x,Z,{\mathbf Z}_{\bar j})$ as a linear map acting on ${\mathbf Z}_{\bar i}$. For the terms ${\Ltt+\epsilon}$ and $\Hsc(t,x,Z)$ this action corresponds to matrix multiplication from the right. For the term $\mathbf D_{\cdot}\Hsc(t,x,Z)  Z$, input for this map is represented by the ``open slot'' labelled as $\cdot$.}
    \begin{align}
        {\mathbfscr H}(t,x,Z,{\mathbf Z}_{\bar j})
        =&{\Ltt+\epsilon} + \Hsc(t,x,Z)
        + \mathbf D_{\cdot}\Hsc(t,x,Z)  Z, \label{eq:mathbfscrH}\\
        \mathbf H_{\bar i}(t, x, Z,{\mathbf Z}_{\bar j})
        =&
        - \alpha t^{\bar p} \tau^{\bar l}_l  \dot n_{\bar i} \Ctt^{l}(t,x,Z) {\mathbf Z}_{\bar l}{}
        -\alpha t^{\bar p} \tau^{\bar l}_l  \Ctt_{\bar i}^{l}(t,x,Z){\mathbf Z}_{\bar l}\label{eq:Hbfi}\\
        &-\alpha t^{\bar p-\Ltt-\epsilon} \tau^{\bar l}_l  {\mathbf Z}_{\bar i} \cdot\mathbf D\Ctt^l(t,x,Z){\mathbf Z}_{\bar l}\notag\\
        &-\alpha  t^{\bar p}\Ctt^{l}(t,x,Z) \Dc_{\bar i} \bar p \log(t) \tau^{\bar l}_l {\mathbf Z}_{\bar l}
        -\alpha  t^{\bar p}\Ctt^{l}(t,x,Z) \Dc_{\bar i} \tau^{\bar l}_l {\mathbf Z}_{\bar l}
        \notag\\
        &
        +\alpha  t^{\bar p} \tau^{\bar l}_l \Ctt^{l}(t,x,Z) \Dc_{\bar l} \Ltt \log(t) {\mathbf Z}_{\bar i}
        +2\alpha  t^{\bar p} \tau^{\bar l}_l\Ctt^{l}(t,x,Z)  \Dc_{[\bar l} b_{\bar i]}{}^\Sigma \omega^{\bar m}{}_\Sigma{\mathbf Z}_{\bar m}\notag\\
        &
        +t^{\Ltt+\epsilon} \Hsc_{\bar i}(t,x,Z)  Z
        +t^{\Ltt+\epsilon+Q_2}\Htt_{\bar i}(t,x,Z)
        +t^{Q_2}{\mathbf Z}_{\bar i} \cdot\mathbf D\Htt(t,x,Z).\notag
    \end{align}
    Equations \eqref{eq:Zboldevol} hold for all $(t,x)\in (T_1,T_0]\times\Sigma$.
    Let us continue by rewriting \eqref{eq:mathbfscrH}:
        \begin{align}
        {\mathbfscr H}(t,x,U,{\mathbf Z}_{\bar l}):={\mathbfscr H}(t,x,\tilde\Phi^{-1}(t,x,S U),{\mathbf Z}_{\bar l})
            =&{\Ltt+\epsilon} + \Hsc(t,x,\tilde\Phi^{-1}(t,x,S U))\label{eq:Hscrbold.2}\\
            &+ \mathbf D_\cdot\Hsc(t,x,\tilde\Phi^{-1}(t,x,S U)) \tilde\Phi^{-1}(t,x,S U). \notag
        \end{align}
        Taking \eqref{eq:defdHsc} as the definition of $\mathbf D_\cdot\Hsc$ and \eqref{eq:Hsc} as the definition of $\Hsc$, we do this by using $Z=\tilde\Phi^{-1}(t,x,S U)$, the product rule and the $O$-notation introduced in Section~\ref{sec:analysisbg.N}. It is not difficult to show that
        \begin{align*}
            &\mathbf D_u\Hsc(t,x,\tilde\Phi^{-1}(t,x,S U)) \tilde\Phi^{-1}(t,x,S U) 
            =\begin{pmatrix}
            -|u|_h^{-2} \frac{c_s^2}{1-c_s^2}    {\check L}^{kq}\check u_{q} \check u_k\\
                0\\ 
            0
            \end{pmatrix}+O(t^\Ltt),\\
            &\mathbf D_z\Hsc(t,x,\tilde\Phi^{-1}(t,x,S U)) \tilde\Phi^{-1}(t,x,S U) 
            =\begin{pmatrix}
                0 \\
                |u|_h^{-2}   {\check L}^{kq} {\check u}_{q} {\check u}_k\\ 
            0
            \end{pmatrix}+O(t^\Ltt),\\
            &\mathbf D_{z_i}\Hsc(t,x,\tilde\Phi^{-1}(t,x,S U)) \tilde\Phi^{-1}(t,x,S U) 
            =\begin{pmatrix}
                0\\
                0\\
            -\frac{  2(1-c_s^2)  {\check L}^{qi} {\check u}_{q} }{(1-c_s^2) |u|_h^2}u_j
            \end{pmatrix}+O(t^\Ltt),
        \end{align*}
        and that therefore 
        \begin{align*}
            &\mathbf D_{\cdot}\Hsc(t,x,\tilde\Phi^{-1}(t,x,S U)) \tilde\Phi^{-1}(t,x,U) 
            =\begin{pmatrix}
            -|u|_h^{-2} \frac{c_s^2}{1-c_s^2}    {\check L}^{pq}{\check u}_{q} {\check u}_p
                &0
                &0\\
                0
                &|u|_h^{-2}   {\check L}^{pq} {\check u}_{q} {\check u}_p
                &0\\ 
            0 
            &0
            &-\frac{  2  {\check L}^{qk} {\check u}_{q} }{ |u|_h^2}u_j
            \end{pmatrix}+O(t^\Ltt).
        \end{align*}
        This together with \eqref{eq:Hsc} into \eqref{eq:Hscrbold.2} yields
        \begin{align}
            &{\mathbfscr H}(t,x,U,{\mathbf Z}_{\bar l})\label{eq:Hscrbold.3}\\
            = &
            \begin{pmatrix}
            \Ltt+\epsilon -|u|_h^{-2} \frac{c_s^2}{1-c_s^2}    {\check L}^{pq}{\check u}_{q} {\check u}_p & 0 & 0 \\
            0 & \Ltt+\epsilon + |u|_h^{-2}   {\check L}^{pq} {\check u}_{q} {\check u}_p&  0 \\
            0 & 0 & \epsilon h_{j}{}^{k}+{\check L}_{j}{}^{k}
            -\frac{   {\check L}^{rq} {\check u}_{q} {\check u}_{r}}{ |u|_h^2} h_{j}{}^{k}
            -\frac{  2   }{ |u|_h^2}u_j {\check L}^{qk} {\check u}_{q}
            \end{pmatrix}+O(t^{\Ltt}).\notag
        \end{align}
    We are now finally ready to define the remaining component $\mathscr{P}_{{\mathbf Z}_{\bar i}}(t^{-1+Q}\Fsc(t,x,W))$ of $t^{-1+Q}\Fsc(t,x,W)$ in \eqref{eq:Wevol} as the sum of the $O(t^\Ltt)$-terms from ${\mathbfscr H}$ in \eqref{eq:Hscrbold.3} and the function $\mathbf H_{\bar i}(t, x, U,{\mathbf Z}_{\bar j})$ given in \eqref{eq:Hbfi}, all of these multiplied with the matrix $\Stt$ in \eqref{eq:Zsymmetriser} from the left. With this definition, it is a straightforward consequence of \eqref{eq:Wevol} with \eqref{eq:Wdef} -- \eqref{WFTwo} and \eqref{eq:Zboldevol} with \eqref{eq:Ctt}, \eqref{eq:mathbfscrH}, \eqref{eq:Hbfi} and \eqref{eq:Zsymmetriser} 
    that $W$ solves the $\mathscr{P}_{{\mathbf Z}_{\bar l}}$-projection of \eqref{eq:Wevol} for all $(t,x)\in (T_1,T_0]\times\Sigma$, and is therefore a solution of the full system \eqref{eq:Wevol}. 

    Conversely, consider initial data $W_{\mathbf{0}}$ in \eqref{eq:W0fromX0}. These initial data are in $H^{k-1}(\Sigma)$ and the values of $W_{\mathbf{0}}$ are contained in $\Omega$ for all points on $\Sigma$. Since \eqref{eq:Wevol} are explicitly symmetric hyperbolic and $k-1>N/2+1$, it follows from \cite[Proposition~1.4, Chapter~16]{taylor2011} that there exists a unique classical  solution $W$ of \eqref{eq:Wevol} in $C^0((T_1,T_0],H^{k-1}(\Sigma))\cap C^1((T_1,T_0],H^{k-2}(\Sigma))$ for some $T_1\in (0,T_0)$ which agrees with $W_{\mathbf{0}}$ at $t=T_0$. Since we assume that $T_1$ is sufficiently close to $T_0$, it follows by continuity, that $W(t,x)$ is contained in $\Omega$ for all $(t,x)\in (T_1,T_0]\times\Sigma$.
    A not difficult, but slightly lengthy standard argument can be used to show that the $U$- and ${\mathbf Z}_{\bar i}$-components of this solution $W$ (cf.\ \eqref{eq:Pscrdef}) satisfy \eqref{eq:DUZbold} with \eqref{eq:DUZbold.1} -- \eqref{eq:DUZbold.2} and that the $z$-component of $W$ satisfies \eqref{eq:zconsist} for all $(t,x)\in (T_1,T_0]\times\Sigma$. This implies a number of things: First, same arguments we used above reveal that the $\Psc_U$-projection of \eqref{eq:Wevol} and \eqref{eq:Uevol} are therefore identical equations and hence the $U$-component of $W$ is a solution of \eqref{eq:Uevol} that, due to the structure of $W_{\mathbf{0}}$, satisfies \eqref{eq:Uconsist} for all $(t,x)\in (T_1,T_0]\times\Sigma$. 
    Since therefore $U$ and $\Dc_{\lb} U$ are both members of $C^0((T_1,T_0],H^{k-1}(\Sigma))\cap C^1((T_1,T_0],H^{k-2}(\Sigma))$, it follows that $U$ is actually a member of $C^0((T_1,T_0],H^{k}(\Sigma))\cap C^1((T_1,T_0],H^{k-1}(\Sigma))$.      
    This then implies that the function $X$ which we define by \eqref{U2Z} from $U$ is a solution of \eqref{eq:PDEgen.transf.nonsymm.N} in $C^0((T_1,T_0],H^{k}(\Sigma))\cap C^1((T_1,T_0],H^{k-1}(\Sigma))$ and that $X$ agrees with \eqref{eq:X0} at $t=T_0$. 

    Now we come to the second part of the proof of this proposition: to verify that the map $\Fsc(t,x,W)$ in \eqref{eq:Wevol} (which we have now fully defined) has the stated uniformly bounded polynomial property in the sense of Section~\ref{sec:analysisbg.N}. Again, this part only requires straightforward arguments; in order to keep the presentation compact here we therefore restrict our discussion to the most important steps. In particular, we do not discuss in any further detail that,
    given a smooth function $Q$ satisfying the bound \eqref{eq:Qinequal.1} and observing \eqref{eq:Wchidef} and \eqref{eq:WOmegadef}, our results about \eqref{eq:Uevol} and \eqref{eq:Zevol} in Section~\ref{sec:variablescont.N} imply the stated property for the components $\Psc_z(\Fsc(t,x,W))$ and $\Psc_U(\Fsc(t,x,W))$ almost directly.

    We now conclude the proof of Proposition~\ref{prop:Fuchsiansystem} by listing the most important arguments that lead to the stated property of the remaining component $\mathscr{P}_{{\mathbf Z}_{\bar l}}(\Fsc(t,x,W))$, again assuming that $Q$ satisfies the bound \eqref{eq:Qinequal.1} observing \eqref{eq:Wchidef} and \eqref{eq:WOmegadef}. Recall that this component was defined above as the sum of the $O(t^\Ltt)$-terms from ${\mathbfscr H}$ in \eqref{eq:Hscrbold.3} and the function $\mathbf H_{\bar i}(t, x, U,{\mathbf Z}_{\bar j})$ given in \eqref{eq:Hbfi}, multiplied with the matrix $\Stt$ in \eqref{eq:Zsymmetriser} from the left. In can be checked easily that the map $\Stt$ satisfies the required property directly as a consequence of \eqref{eq:Wchidef} and \eqref{eq:WOmegadef}. It is also not difficult to show that the $O(t^\Ltt)$-terms from ${\mathbfscr H}$ satisfy this property. Let us focus a bit more on \eqref{eq:Hbfi}. 
    It is straightforward to verify directly from their definitions that $\Ctt^l(t,x,Z)$ in \eqref{eq:Ctt} as well as $(\Dc_{\bar i} \Ctt^{l})(t,x,Z)$ and $\mathbf D\Ctt^l(t,x,Z)$ in \eqref{eq:defdCtt} can be identified with uniformly bounded polynomials in the sense of Section~\ref{sec:analysisbg.N} 
    for 
    \begin{equation}
        \label{eq:deftildechi}
        \tilde\chi(t,x,Z)=\bigl(u, z, z^{-1}, t^{\Ltt} z z_i, \sqrt{u^2 t^{2 \Ltt}+z^2}\bigr)^T,
    \end{equation}
    and 
    \begin{equation}
        \tilde\Omega=\left\{Z\in\Rbb^{2+N} \,\left|\, u>u_*, z>\sqrt{h^{ij}\hat u_{*i}\hat u_{*j}}\right.\right\},
    \end{equation}
    where $u_*>0$ and $\hat u_{*i}$ are the constants given in the proposition. 
    The same conclusions hold, as one can also directly check from their definitions, for $\Htt(t,x,Z)$ (whose polynomial properties have been verified in Section~\ref{sec:variablescont.N}) and therefore for $\Htt_{\bar i}(t,x,Z)$ and $\mathbf D\Htt(t,x,Z)$ (cf.\ \eqref{eq:defdHtt}). The corresponding maps $\Ctt^l(t,x,U)$, $(\Dc_{\bar i} \Ctt^{l})(t,x,U)$, $\mathbf D\Ctt^l(t,x,U)$, $\Htt(t,x,U)$, $\Htt_{\bar i}(t,x,U)$ and $\mathbf D\Htt(t,x,U)$ are then obtained by composing $\tilde\chi$ in \eqref{eq:deftildechi} with the maps \eqref{U2Z} and \eqref{U2X}. This composition does not alter the uniformly bounded polynomial properties of the implied maps, and it follows directly that we have thereby now fully established the statement of Proposition~\ref{prop:Fuchsiansystem}.
\end{proof}

\subsection{The proof of Theorem~\ref{thm:Euler1}}
\label{sec:proof1}
\phantom{.}
\vspace{1ex}

With Proposition~\ref{prop:Fuchsiansystem} at our disposal we are now in the position to prove Theorem~\ref{thm:Euler1}. The proof is mostly an adaptation of previous proofs using the Fuchsian method such as \cite{BeyerOliynyk:2020, BeyerOliynyk:2021, BeyerOliynykZheng:2025, beyerStabilityFLRWSolutions2023}. One of the main differences here is the reference-independence: our result is not a perturbation of a specific reference solution of the Euler equations.

\noindent \underline{The initial value problem and local-in-time existence and uniqueness:} 
We pick so far arbitrary constants $T_0>0$, $\epsilon>0$ and $\eta>0$, and initial data $\nu_{\mathbf{0},i}$ and $\rho_{\mathbf{0}}$ as
in the hypothesis of Theorem~\ref{thm:Euler1}. Then we define 
\begin{equation}
    \label{eq:Wdata}
    W(T_0,x)=W_{\mathbf{0}}(x)=\begin{pmatrix}
        z_{\mathbf{0}}(x)\\
        U_{\mathbf{0}}(x)\\
        {\mathbf Z}_{\mathbf{0},\bar i}(x)
    \end{pmatrix}
    =\begin{pmatrix}
        \sqrt{h^{ij} u_{\mathbf{0},i} u_{\mathbf{0},j}}
        u_\mathbf{0}\\
        \hat h_i{}^j u_{\mathbf{0},j}\\
        \check h_i{}^j u_{\mathbf{0},j}\\
        T_0^{\Ltt+\epsilon}\Dc_{\bar i} u_\mathbf{0}\\
        T_0^{\Ltt+\epsilon}\Dc_{\bar i}|u_\mathbf{0}|_h\\
        T_0^{\Ltt+\epsilon}\Dc_{\bar i}\bigl(
            T_0^{-\Ltt}\frac{u_{\mathbf{0},i}}{|u_\mathbf{0}|_h}
        \bigr)
    \end{pmatrix},
\end{equation}
in full agreement with the formula \eqref{eq:W0fromX0} given in Proposition~\ref{prop:Fuchsiansystem}, provided we choose
\begin{equation}
    \label{eq:Xdata}
        u_\mathbf{0}=T_0^{P-\Ltt} \left(\frac{c_s^2}{\mathcal P_0}\right)^{c_s^2/(1+c_s^2)}\rho_\mathbf{0}^{c_s^2/(1+c_s^2)},\quad
        u_{\mathbf{0},i}=\frac{T_0^{\Ltt} u_\mathbf{0}}{\sqrt{1-|\nu_\mathbf{0}|^2}}\nu_{\mathbf{0},i}.
\end{equation}
Let us then pick the constants $u_*>0$ and $\hat u_{*i}$ in $\Rbb^{N}$ such that $u_\mathbf{0}$ and $\hat h_i{}^j u_{\mathbf{0},i}$ satisfy
\[u_\mathbf{0}>u_*,\quad \hat h^{ij} u_{\mathbf{0},i}u_{\mathbf{0},j}>\hat u_{*i}\hat u_{*}^i\]
for each point in $\Sigma$. 
Given \eqref{eq:parameterchoices}, it is easy to see that  \eqref{eq:theorem1restr} implies \eqref{eq:lem.restr} provided we assume that $\epsilon>0$ is sufficiently small.
Equation \eqref{eq:Wdata} represents the initial data set for the Cauchy problem for the $W$-variables equivalent to the initial value problem of the $X$-variables which are directly related to  the physical variables via 
\eqref{eq:Gammanu.Var2}. As we noticed as part of Proposition~\ref{prop:Fuchsiansystem}, the vector $W_{\mathbf{0}}$ satisfies the consistency conditions \eqref{eq:Uconsist}, \eqref{eq:znormalisationconstr}, \eqref{eq:DUZbold} and \eqref{eq:zconsist}, and our choice of the constants $u_*>0$ and $\hat u_{*i}$ above implies that its values are contained in the $\Omega$ given in \eqref{eq:WOmegadef} for all points on $\Sigma$. 
Proposition~\ref{prop:Fuchsiansystem} asserts the existence of a $T_1\in [0,T_0)$ and a  unique  classical solution 
\begin{equation}
    \label{eq:Wlocalsol}
    W \in C^0\bigl((T_1,T_0],H^{k-1}(\Sigma)\bigr)\cap C^1\bigl((T_1,T_0],H^{k-2}(\Sigma)\bigr)
\end{equation}
with values contained in $\Omega$
of the Cauchy problem of \eqref{eq:Wevol} launched from the initial data $W_{\mathbf{0}}$ in \eqref{eq:Wdata} at $t=T_0$. The solution $W$ satisfies the consistency conditions \eqref{eq:Uconsist}, \eqref{eq:znormalisationconstr}, \eqref{eq:DUZbold} and \eqref{eq:zconsist} for all $(t,x)\in (T_1,T_0]\times\Sigma$.
The $z$, $U$- and ${\mathbf Z}_{\bar l}$-components of $W$ are members of  $C^0((T_1,T_0],H^{k}(\Sigma))\cap C^1((T_1,T_0],H^{k-1}(\Sigma))$, 
$C^0((T_1,T_0],H^{k}(\Sigma))\cap C^1((T_1,T_0],H^{k-1}(\Sigma))$ and $C^0((T_1,T_0],H^{k-1}(\Sigma))\cap C^1((T_1,T_0],H^{k-2}(\Sigma))$, respectively. The function $X$ defined from $U$ via \eqref{U2X} is in $C^0((T_1,T_0],H^{k}(\Sigma))\cap C^1((T_1,T_0],H^{k-1}(\Sigma))$ and is
a solution of  \eqref{eq:PDEgen.transf.nonsymm.N} for all $(t,x)\in (T_1,T_0]\times\Sigma$. 

\phantom{.}

\noindent \underline{The leading-order term and the Fuchsian system for the remainder:} 
On $(0,T_0]\times\Sigma$ we next define  the function 
\begin{equation}
    \label{eq:Wcircformula}
    \mathring W(t,x)=\mathrm{diag}\Bigl(1,
            \id, \id, e^{(\eta \hat h+\check L) \log \frac{t}{T_0}}, \left(\frac{t}{T_0}\right)^{\Ltt+\ep}, \left(\frac{t}{T_0}\right)^{\Ltt+\ep}, e^{(\epsilon h+\check L) \log \frac{t}{T_0}}
        \Bigr)W_{\mathbf{0}}.
\end{equation}
We observe that
\begin{equation}
    \mathring W(T_0,x)=W_{\mathbf{0}}(x),
\end{equation}
and that
\begin{equation}
    \label{eq:backgroundODE}
    A^0(t,x,\mathring W)\partial_t \mathring W=t^{-1}\Asc(t,x,\mathring W)\Pbb \mathring W
\end{equation}
holds for all $(t,x)\in(0,T_0]\times\Sigma$. A crucial property is that $A^0(t,x,W)^{-1}\Asc(t,x,W)$ does not actually dependent on $W$ (cf.\ \eqref{eq:Wevol}, \eqref{WA0} and \eqref{WAsc}). We call \eqref{eq:backgroundODE} the \emph{truncated Euler equations} and notice that \eqref{eq:Wcircformula} is the unique solution of the truncated Euler equations corresponding to our initial data  \eqref{eq:Wdata}. 

The heuristic idea which motivates the definition of $\mathring W$ in \eqref{eq:Wcircformula} is that $\mathring W$
is expected to be the leading-order term of the solution $W$ of the full Euler equations (if $W$ extends from $t=T_1$ to $t=0$ of course), that is, the \emph{remainder quantities}
\begin{equation}
    \label{Wbardef}
    \bar U=U-\mathring U,\quad \bar {\mathbf Z}_{\bar i}={\mathbf Z}_{\bar i}-\mathring {\mathbf Z}_{\bar i},
    \quad \bar W=W-\mathring W,
\end{equation}
are expected to be uniformly small  (but do not necessarily decay) in some appropriate sense near $t=0$. 
Especially we observe that
\begin{equation}
    \label{eq:barvanishingID}
    \bar W(T_0)=0.
\end{equation}

We next derive the evolution equation implied for the remainder quantity $\bar W$ implied by \eqref{eq:Wevol} and \eqref{Wbardef}. 
A straightforward calculation, which makes crucial use of  \eqref{eq:backgroundODE} (and especially of the before mentioned $W$-independence of $A^0(t,x,W)^{-1}\Asc(t,x,W)$), yields
\begin{align}
    &A^0(t,\mathring  W+\bar W)\partial_t \Wb
    + \alpha t^{-1+\bar p} \tau^{\bar l}_l b_{\bar l}{}^\Omega A^{l}(t,\mathring  W+\bar W) \Dc_\Omega \Wb\label{eq:Wbarevol}\\
    = &t^{-1}\Asc(t,\mathring  W+\bar W)\Pbb\Wb+t^{-1}\Tsc_2(t,\mathring W,\Wb)
    +t^{-1+\bar Q}\Tsc(t,\mathring W,\Dc_\Omega \mathring W,\Wb),\notag
\end{align}
where
\begin{align}
    \Tsc_2(t,\mathring W,\Wb)&\label{WbarTwo}\\
    &\hspace{-1.5cm}:=\left(
    \frac {(1-c_s^2)z {\check L}^{kq} \bar{\check u}_{q} \bar{\check u}_k}{u^2 t^{2\Ltt}+(1-c_s^2)z^2}\quad    
    \left(\begin{smallmatrix}
    -\frac{c_s^2 \check L^{ll'}\bar{\check u}_l\bar{\check\Vsp}_{l'}}{t^{2\Ltt}\Vtm^2+(1-c_s^2) \Vspnorm^2}\Vtm\\
    0\\
    0
    \end{smallmatrix}\right)\quad - \left(\begin{smallmatrix}
        {\check L}^{pq}\bar{\check u}_{q} \bar{\check u}_p & 0 & 0 \\
    0 &  -u^2 |u|_h^{-2}   {\check L}^{pq} \bar{\check u}_{q} \bar{\check u}_p&  0 \\
    0 & 0 & 
    |u|_h^2{   {\check L}^{rq}\bar{\check u}_{q} \bar{\check u}_{r}} h^{pk}
    +{  2   }{ |u|_h^2}u^p {\check L}^{qk} \bar{\check u}_{q}
    \end{smallmatrix}\right)\bar{\mathbf Z}_{\bar i}\right),\notag
\end{align}
and
\begin{align}
    t^{-1+\bar Q}\Tsc(t,\mathring W,\Dc_\Omega \mathring W,\Wb)=&t^{-1}(\Fsc_2(t,\mathring W+\Wb)-\Tsc_2(t,\mathring W,\Wb))\label{WbarF}\\
    &+t^{-1+Q}\Fsc(t,\mathring W+\Wb)-\alpha t^{-1+\bar p} \tau^{\bar l}_l b_{\bar l}{}^\Omega A^{l}(t,\mathring  U+\bar U) \Dc_\Omega \mathring W\notag.
\end{align}
The quantities $A^0$, $A_l$, $\mathscr A$, $\Pbb$, $\Fsc_2$, $\Fsc$ and $Q$ were given in \eqref{WPbb} -- \eqref{WFTwo} and \eqref{eq:Qinequal.1}. The reason for the particular definition of the map \eqref{WbarTwo} will be discussed below.
It is not difficult to establish using Proposition~\ref{prop:Fuchsiansystem} and equations \eqref{eq:Wcircformula} and \eqref{Wbardef} that the map $\Tsc$ can be identified with a uniformly bounded polynomial in the sense of Section~\ref{sec:analysisbg.N} for
\begin{itemize} 
    \item the map 
    \begin{align*}
        \bar\chi(t,x,\Wb)=\bigl(&
        \bar z, \bar u, \bar{\hat u}_i, \bar{\check u}_i, \bar{\mathbf Z}_{\bar j}, |\mathring u(t,x)+\bar u|^{-1}_h, 
        \sqrt{(\mathring u(t,x)+\bar u)^2 t^{2 \Ltt}+|\mathring u(t,x)+\bar u|^{2}_h},\\
        & 1/\sqrt{(\mathring u(t,x)+\bar u)^2 t^{2 \Ltt}+|\mathring u(t,x)+\bar u|^{2}_h}, 1/((\mathring u(t,x)+\bar u)^2 t^{2 \Ltt}+(1-c_s^2)|\mathring u(t,x)+\bar u|^{2}_h)\bigr)^T,
    \end{align*}
\item the open ball $\bar\Omega$ with radius $R$ around the origin in $\Rbb^{2+4N+N^2}$ where $R>0$ is small enough so that $\bar W\in\bar\Omega$ implies that $W=\mathring W(t,x)+\bar W\in\Omega$ (cf.\ \eqref{eq:WOmegadef}) for every $(t,x)\in (0,T_0]\times\Sigma$,  
\end{itemize}
provided we choose the smooth function $\bar Q$ on $\Sigma$ such that
\begin{equation}
    \label{eq:Qbarinequal.1}
    0<\bar Q<\min\{Q,\bar p,\Ltt+\epsilon,\Delta P\}=\min\{\Ltt,\Delta P,\bar p-\Ltt-\epsilon,q-\Ltt, \bar q-\Ltt\}
\end{equation}
everywhere on $\Sigma$. This follows from  Definition~\ref{def:asympKasner} and \eqref{eq:Wcircformula} where $\Delta P$ represents the smallest non-zero eigenvalue of $\check L_i{}^j$ (cf.\ \eqref{eq:Kttspecdecomp}). 

Because $W$ is a solution of \eqref{eq:Wevol} for all $(t,x)\in (T_1,T_0]\times\Sigma$ with $W(T_0,x)=W_{\mathbf{0}}(x)$, it follows that $\bar W=W-\mathring W$ is a solution of \eqref{eq:Wbarevol} for all $(t,x)\in (T_1,T_0]\times\Sigma$ with $\bar W(T_0,x)=0$.

In preparation for the Fuchsian analysis we now also split the map $\Tsc$ into the following two parts  
\begin{equation}
    \label{WbarF.2}
    \Tsc(t,\mathring W,\Dc_\Omega \mathring W,\Wb)=:\tilde\Tsc(t,\mathring W,\Dc_\Omega \mathring W)+\Tsc_0(t,\mathring W,\Dc_\Omega \mathring W,\Wb),
\end{equation}
where
\begin{equation}
    \label{WbarF.3}
    \tilde\Tsc(t,\mathring W,\Dc_\Omega \mathring W)=\Tsc(t,\mathring W,\Dc_\Omega \mathring W,0),\quad \Tsc_0(t,\mathring W,\Dc_\Omega \mathring W,0)=0,
\end{equation}
which means that
\begin{equation}
    \Tsc_2(t,\mathring W,0)=0.
\end{equation}

\phantom{.}

\noindent \underline{Fuchsian analysis and global existence:} We now come to the central step in the proof of Theorem~\ref{thm:Euler1}: the Fuchsian analysis. The first step is to show that the PDE \eqref{eq:Wbarevol} satisfies the so-called \emph{coefficient assumptions} which are the foundation of the Fuchsian global existence theory established in \cite{BOOS:2021}\footnote{The existence result we employ is Theorem~A.2 from \cite{BeyerOliynyk:2020} together with Remark~A.3 from that same article, which, together, amount to a slight generalization of the Fuchsian global existence theory from \cite{BOOS:2021}. Notice in all of what follows we adopt the convention that $T_0$ and $t$ are positive here.}. 

We refer the reader to the details about the coefficient assumptions \cite[Section~A.1]{BeyerOliynyk:2020}. 
    For now, we let $T_0\in (0,1]$ be an arbitrary constant. However, in a later step of the proof we will want to use the size of $T_0$ to control certain error terms that arise in our reference-independent approach. We therefore pay particular attention to how the size of $T_0$ affects each of the following steps and, in particular, we want to determine which quantities can be picked \emph{independently of $T_0$} within the interval $(0,1]$.
    
    We identify the constant $R>0$ in \cite[Section~A.1]{BeyerOliynyk:2020} with the constant $R$ we chose above to define the set $\bar\Omega$, but will make use of the opportunity to shrink $R$ whenever this is useful for our arguments. The constant $\Rc>0$ in \cite[Section~A.1]{BeyerOliynyk:2020} is an arbitrary constant so far, but later it will have to be chosen so that $\sup_{(0,T_0]\times\Sigma}|\mathring W|<\Rc$ and $\sup_{(0,T_0]}\norm{\Dc\mathring W}_{H^{k-3}(\Sigma)}<\Rc$. Equation~\eqref{eq:Wcircformula} allows us to easily show that it is possible to pick $\Rc$ so that this holds independently of $T_0$.
    
    Let the vector bundles $Z_1$ and $Z_2$ in \cite[Section~A.1]{BeyerOliynyk:2020} be the trivial $\Rbb^{2+4N+N^2}$- and $\Rbb^{2+6N+5N^2+N^3}$-bundles over $\Sigma$, respectively, and consider $w_1=\mathring W(t)$ as a section in a bounded subset $\Zc_1$ of $Z_1$ and $w_2=(\mathring W(t),\Dc_\Omega \mathring W(t))$ as a section in a bounded subset $\Zc_2$ of $Z_2$. As discussed above, because of the particular form of $\mathring W$, uniform bounds for $\Zc_1$ and $\Zc_2$ over $(0,T_0]\times\Sigma$ can be chosen independently of $T_0$.
    
    Let the vector bundle $V$ in \cite[Section~A.1]{BeyerOliynyk:2020} be the trivial $\Rbb^{2+4N+N^2}$-bundle over $\Sigma$; notice that our solutions $\bar W(t)$ of \eqref{eq:Wbarevol} will be considered as sections in the $R$-ball subset of this bundle $V$.
    
    The map $\Pbb$ defined in \eqref{WPbb} and \eqref{eq:Pi} clearly satisfies property (1) of \cite[Section~A.1]{BeyerOliynyk:2020}. Moreover, our map $A^0$ in \eqref{eq:Wbarevol} given by \eqref{WA0} and \eqref{eq:Zsymmetriser} clearly satisfies the properties of the map $B^0$ of (2) in \cite[Section~A.1]{BeyerOliynyk:2020}. The same is true for the map $\Asc+\Pbb\Fsc_2$ in \eqref{eq:Wbarevol} given by \eqref{WAsc}, \eqref{WFTwo} and \eqref{eq:Gscdef} when it is identified with $\Bc$ in \cite[Section~A.1]{BeyerOliynyk:2020}. The precise choice of the constants $\gamma_1$, $\kappa$ and $\gamma_2$ in \cite[Equation~(A.3)]{BeyerOliynyk:2020} is irrelevant for us here, but it is important again that they can be chosen independently of $T_0$. In the same way the precise form of the maps $\tilde B^0$ and $\tilde\Bc$ in \cite[Section~A.1]{BeyerOliynyk:2020} is not of importance for us.
    
    The map $F$ in (3) in \cite[Section~A.1]{BeyerOliynyk:2020} satisfies the required properties if we identify the constant $p$ in \cite[Section~A.1]{BeyerOliynyk:2020} with $\bar Q$, the map $\tilde F$ in \cite[Section~A.1]{BeyerOliynyk:2020} with $\tilde\Tsc$ (cf.\ \eqref{WbarF.3}), the map $F_1$ in \cite[Section~A.1]{BeyerOliynyk:2020} with $0$, and the map $F_2$ in \cite[Section~A.1]{BeyerOliynyk:2020} with $\Pbb^\perp \Tsc_2$ (cf.\ \eqref{WbarTwo}). Crucially, the last map is quadratic in $\Pbb\bar W$. We can therefore make $T_0$-independent choices for $\lambda_1$ and $\lambda_2$, and particularly, $\lambda_3=O(R)$ in \cite[Section~A.1]{BeyerOliynyk:2020}. Showing that the map $B$ in (4) in \cite[Section~A.1]{BeyerOliynyk:2020} satisfies the required properties is very straightforward and is omitted here.
    
    Finally (5) in \cite[Section~A.1]{BeyerOliynyk:2020} is concerned with the map $\Div \! B$ which is defined locally in \cite[Equation~(2.6)]{BeyerOliynyk:2020}: In the case of \eqref{eq:Wbarevol}, it is not difficult to show that $\Pbb^\perp\Div \! B=0$ and hence we can choose the constants $\beta_4$, $\beta_5$, $\beta_6$ and $\beta_7$ in \cite[Section~A.1]{BeyerOliynyk:2020} to be zero. Regarding $\Pbb\Div \! B$, it is again not difficult to establish that we can choose $\theta>0$ in \cite[Section~A.1]{BeyerOliynyk:2020} independently of $T_0$ and that we can set $\beta_0=\beta_2=0$. What is crucial here\footnote{This is exactly the reason why the $z$-component was added to the $W$-vector in Proposition~\ref{prop:Fuchsiansystem}.} is that only the variables $z$ and $u$ appear in $A^0$ in \eqref{WA0} and hence the only contributions to the time derivatives terms of $\Pbb\Div \! B$ are from the components of $\mathring W$ in \eqref{eq:Wcircformula} that are constant in time. This is important because time derivatives of the other components of $\mathring W$ can \emph{not be bounded uniformly by a constant that is independent of $T_0$}. In any case, with the same arguments we also conclude that all $t^{-1}$-terms in $\Pbb\Div \! B$ are $O(\Pbb(\bar W))$ with constants independent of $T_0$ as required and, especially, that we can pick $\beta_1=\beta_3=O(R)$ independently of $T_0$. This means that we have now verified the coefficient assumptions.
    
    Now we are in the position to apply the Fuchsian theorem \cite[Theorem~A.2, Remark~A.3 (ii)]{BeyerOliynyk:2020}.
    We identify the differentiability parameter $k$ in \cite[Theorem~A.2, Remark~A.3 (ii)]{BeyerOliynyk:2020} with the value $k-2$ here. This means that our assumption $k>N/2+3$ in Theorem~\ref{thm:Euler1} is identical to the condition $k>N/2+1$ in \cite[Theorem~A.2, Remark~A.3 (ii)]{BeyerOliynyk:2020}. 
    
    From our discussion above about the time derivative of $\mathring W$, we see that the next condition in \cite[Theorem~A.2]{BeyerOliynyk:2020}, namely 
    \begin{equation*}
        \sup_{0<t<T_0} \max\Bigl\{ \norm{\Dc w_1(t)}_{H^{k-3}(\Sigma)},
        \norm{t^{1-q}\del{t}w_1(t)}_{H^{k-3}(\Sigma)}\Bigr\} < \frac{\Rsc}{C_{\text{Sob}}},
    \end{equation*}
    can unfortunately not be satisfied independently of $T_0$, irrespective of the value of $q>0$. However, inspecting the proof of \cite[Theorem~A.2]{BeyerOliynyk:2020}, we observe that the condition about the time derivative only enters the analysis of the map $\Div \! B$. And, as discussed above, the special structure of $A^0$ and $\mathring W$ implies that only components of $\mathring W$ contribute that are constant in time. It is this special structure that allows us to apply \cite[Theorem~A.2]{BeyerOliynyk:2020} despite the fact that the above condition can literally not be satisfied independently of $T_0$. 
    
    Now, since $\mathtt{b}=0$ and we can make all the non-zero constants $\beta_i$ and $\lambda_3$ arbitrarily small independently of $T_0$ by shrinking $R$, the inequality for $\gamma$ in \cite[Theorem~A.2, Remark~A.3 (ii)]{BeyerOliynyk:2020} can be satisfied. 

It is now a consequence of \cite[Theorem~A.2, Remark~A.3 (ii)]{BeyerOliynyk:2020} and \eqref{eq:barvanishingID} that there is a $\delta_0$, whose value is independent of $T_0$, such that the initial value problem 
of \eqref{eq:Wbarevol} with initial data \eqref{eq:barvanishingID} has a classical solution \emph{globally} in time, as long as the smallness condition
\begin{equation} 
    \label{eq:smallnesscond}
    \int_0^{T_0} s^{-1+\bar Q} \norm{\tilde\Tsc(s,\mathring W(s))}_{H^{k-2} (\Sigma)} ds<\delta_0
\end{equation}
holds (cf.\ \eqref{WbarF.3}).
This means that the local-in-time solution
$W$ in \eqref{eq:Wlocalsol} 
of \eqref{eq:Wevol} determined by initial data $W_{\mathbf{0}}$ in \eqref{eq:Wdata}, together with all the consistency conditions
\eqref{eq:Uconsist}, \eqref{eq:znormalisationconstr}, \eqref{eq:DUZbold} and \eqref{eq:zconsist},
extends as a classical solution to $t=0$. Especially, \eqref{eq:Wlocalsol} holds with $T_1=0$, and $W$ takes values in $\Omega$ (cf.\ \eqref{eq:WOmegadef}) for all $(t,x)\in (0,T_0]\times\Sigma$. All that is required to do to satisfy the smallness condition \eqref{eq:smallnesscond} is to shrink the size of $T_0>0$.

Because the vector $U$ determined by the $U$-component of $W$ satisfies \eqref{eq:DUZbold}, it follows
\begin{equation}
    \label{eq:Ureg.1}
    U\in C^0((0,T_0],H^{k}(\Sigma))\cap C^1((0,T_0],H^{k-1}(\Sigma)).
\end{equation}
The same holds for $X$ determined from \eqref{U2X}, and 
we obtain from \eqref{eq:Gammanu.Var2} that the physical fluid fields have the property
\begin{equation}
    \rho,\Gamma, \nu_i\in C^0((0,T_0],H^{k}(\Sigma))\cap C^1((0,T_0],H^{k-1}(\Sigma)).
\end{equation}

The Fuchsian theorem \cite[Theorem~A.2, Remark~A.3 (ii)]{BeyerOliynyk:2020} also implies that the remainder $\Wb$ satisfies the energy estimate
\begin{equation}
    \label{eq:FuchsianEnergyEstimate}
    \norm{\bar W(t)}_{H^{k-2}(\Sigma)}  \lesssim \int^{T_0}_t s^{-1+\bar Q} \norm{\tilde \Tsc(s,\mathring W(s))}_{H^{k-2}(\Sigma)}\, ds
\end{equation}
for all $t\in (0,T_0]$. In addition to \eqref{eq:Ureg.1}, this, together with \eqref{eq:Wcircformula} and \eqref{eq:theorem1restr}  implies that $U\in L^\infty((0,T_0],H^{{k-2}}(\Sigma))$\footnote{We remark that we can in general not conclude anything better than the statement that $t^{\Ltt+\ep}U\in L^\infty((0,T_0],H^{k-1}(\Sigma))$ because of \eqref{eq:DUZbold}, \eqref{eq:DUZbold.1} and \eqref{eq:DUZbold.2}.}.

Another consequence of the Fuchsian theorem \cite[Theorem~A.2, Remark~A.3 (ii)]{BeyerOliynyk:2020} is the existence of certain limits and decay estimates at $t=0$. It is a remarkable consequence of the special structure of \eqref{eq:Wevol}, in particular that $\Pbb^\perp A^l=A^l\Pbb^\perp=0$ (cf.\ \eqref{WPbb} and \eqref{WAl}), i.e.\ the entire $U$-component of \eqref{eq:Wevol} does not have any spatial derivative terms, that one can sharpen the statements regarding convergence and decay given in \cite[Theorem~A.2; Remark~A.3]{BeyerOliynyk:2020}\footnote{For the special case of the $U$-component of \eqref{eq:Wevol},  the $H^{k-1}$-norm used in \cite[Proposition~3.2]{BOOS:2021} can be replaced by the $H^k$-norm. As a consequence, the $H^{k-1}$-norm on the left side of the estimate in \cite[equation (3.82)]{BOOS:2021} can be replaced by the $H^k$-norm.} and escalate certain $H^{k-3}$-norms to $H^{k-2}$-norms.
We conclude that there exists a non-negative function $u_*\in H^{k-2}(\Sigma)$ (instead of $H^{k-3}(\Sigma)$) and functions $\hat u_{*,i}\in H^{k-2}(\Sigma)$ (instead of $H^{k-3}(\Sigma)$) with $\check h_i{}^j \hat u_{*,j}=0$, together with a constant $\kappa>0$, such that 
\begin{align}
    \label{eq:fluidestimate.1}
    &\Bigl\|\ub(t)-{T_0^{P-L} \left(\frac{c_s^2}{\mathcal P_0}\right)^{c_s^2/(1+c_s^2)}\rho_\mathbf{0}^{c_s^2/(1+c_s^2)}}u_*\Bigr\|_{H^{k-2}(\Sigma)}\\
    &+\left\|{\bar{\hat u}_i(t)-\frac{T_0^{P} \left(\frac{c_s^2}{\mathcal P_0}\right)^{c_s^2/(1+c_s^2)}\rho_\mathbf{0}^{c_s^2/(1+c_s^2)}}{\sqrt{1-|\nu_\mathbf{0}|^2}}\hat u_{*,i}}\right\|_{H^{k-2}(\Sigma)}\lesssim t^\kappa\notag
\end{align}
for all $t\in (0,T_0]$. 
A similar argument can be used to escalate the decay estimates from \cite[Theorem~A.2; Remark~A.3]{BeyerOliynyk:2020} from the $H^{k-3}$-norm to the $H^{k-2}$-norm which yields
\begin{equation}
    \label{eq:fluidestimate.2}
    \norm{\bar{\check u}_i(t)}_{H^{k-2}(\Sigma)}\lesssim t^\kappa,
\end{equation}
for all $t\in (0,T_0]$.

\phantom{.}

\noindent \underline{Asymptotics of the physical fluid variables:} 
Recall from \eqref{eq:Wcircformula} and \eqref{Wbardef} that $u=\mathring u+\bar u$ and $\hat u_i=\mathring {\hat u}_i+\bar {\hat u}_i$ and that $\mathring u$ and $\mathring{\hat u}_i$ are the constants in time given by \eqref{eq:Wdata} and \eqref{eq:Xdata}. 
By decreasing the size of $T_0$ further if necessary, the energy estimate \eqref{eq:FuchsianEnergyEstimate} can be used to guarantee that the sizes $\norm{u_*}_{H^{k-2}(\Sigma)}$ and $\norm{\hat u_{*,i}}_{H^{k-2}(\Sigma)}$ are as small as we like. In this way we can ensure that 
the limits of $u$ and $\hat u_i$ are non-vanishing on $\Sigma$. 
In fact, we can conclude from \eqref{eq:Wdata}, \eqref{eq:Wcircformula}, \eqref{eq:fluidestimate.1} and \eqref{eq:fluidestimate.2},  that
\begin{align}
    \left\|{u(t)-{T_0^{c_s^2\frac{1-P}{1-c_s^2}} \left(\frac{c_s^2}{\mathcal P_0}\right)^{c_s^2/(1+c_s^2)}\rho_\mathbf{0}^{c_s^2/(1+c_s^2)}}(1+u_*)}\right\|_{H^{k-2}(\Sigma)}&\lesssim t^\kappa,\\
    \left\|{\hat h_i{}^j u_j(t)-\frac{T_0^{P} \left(\frac{c_s^2}{\mathcal P_0}\right)^{c_s^2/(1+c_s^2)}\rho_\mathbf{0}^{c_s^2/(1+c_s^2)}}{\sqrt{1-|\nu_\mathbf{0}|^2}}(\hat h_i{}^j\nu_{\mathbf{0},j}+\hat u_{*,i})}\right\|_{H^{k-2}(\Sigma)}
    &\lesssim t^\kappa,\\
    \norm{\check h_i{}^j u_j(t)}_{H^{k-2}(\Sigma)}&\lesssim t^\kappa,
\end{align}
for all $t\in (0,T_0]$, where we have used the freedom to adapt the value of the unspecified positive constant $\kappa$. It is then not difficult to show by using \eqref{eq:Gammanu.Var2} that
\begin{equation}
    \Bnorm{\nu_i(t)-\hat h_i{}^j\frac{ \nu_{\mathbf{0},j}+\hat u_{*,j}}{\sqrt{\hat h^{kl}(\nu_{\mathbf{0},k}+\hat u_{*,k})(\nu_{\mathbf{0},l}+\hat u_{*,l})}}}_{H^{k-2}(\Sigma)}
    \lesssim t^\kappa+t^{\Ltt},
\end{equation}
for all $t\in (0,T_0]$, which implies the required estimate \eqref{eq:thm1.tiltestimate} (if we adjust the constant $\kappa>0$ again if necessary).
The physical fluid tilt vector $\nu_i$ therefore indeed approaches a unit vector in the eigenspace of the largest eigenvalue $P$. In consistency with this it is also a straightforward consequence of \eqref{eq:Gammanu.Var2} that the Lorentz factor $\Gamma$ diverges, that is,
\begin{equation}
    \Bnorm{t^{(P-c_s^2)/(1-c_s^2)}\Gamma-T_0^{(P-c_s^2)/(1-c_s^2)}\frac{\sqrt{\hat h^{kl}(\nu_{\mathbf{0},k}+\hat u_{*,k})(\nu_{\mathbf{0},l}+\hat u_{*,l})}}{\sqrt{1-|\nu_\mathbf{0}|^2}(1+u_*)}}_{H^{k-2}(\Sigma)}
    \lesssim t^\kappa+t^{\Ltt},
\end{equation}
for all $t\in (0,T_0]$.
We have therefore established the required estimate \eqref{eq:thm1.Lorentzestimate}.
Next we also conclude from \eqref{eq:Gammanu.Var2} that the fluid density satisfies
\begin{equation}
    \Bnorm{t^{c_s^2\frac{1-P}{1-c_s^2}}\rho^{c_s^2/(1+c_s^2)}-T_0^{c_s^2\frac{1-P}{1-c_s^2}} \rho_\mathbf{0}^{c_s^2/(1+c_s^2)}(1+u_*)}_{H^{k-2}(\Sigma)}\lesssim t^\kappa,
\end{equation}
for all $t\in (0,T_0]$.
The following yields the required estimate \eqref{eq:thm1.rhoestimate} and is a straightforward application of the Moser inequality:
\begin{align*}
&\Bnorm{t^{(1+c_s^2)\frac{1-P}{1-c_s^2}}\rho-T_0^{(1+c_s^2)\frac{1-P}{1-c_s^2}} \rho_\mathbf{0}(1+u_*)^{(1+c_s^2)/c_s^2}}_{H^{k-2}(\Sigma)}
=\Bnorm{\frac{\mathcal P_0}{c_s^2} u^{(1+c_s^2)/c_s^2}
-T_0^{(1+c_s^2)\frac{1-P}{1-c_s^2}} \rho_\mathbf{0}(1+u_*)^{(1+c_s^2)/c_s^2}}_{H^{k-2}(\Sigma)}\\
\lesssim&\left\|{u(t)-{T_0^{c_s^2\frac{1-P}{1-c_s^2}} \left(\frac{c_s^2}{\mathcal P_0}\right)^{c_s^2/(1+c_s^2)}\rho_\mathbf{0}^{c_s^2/(1+c_s^2)}}(1+u_*)}\right\|_{H^{k-2}(\Sigma)}\lesssim t^\kappa,
\end{align*}
which holds for all $t\in (0,T_0]$. This completes the proof of Theorem~\ref{thm:Euler1}.

\section{The strongly anisotropic asymptotically extremely-tilted fluid regime near the big bang}
\label{sec:thm2}

\subsection{The second main theorem}
\label{subsec:thm2}
The main technical problem in analysing the asymptotically extremely-tilted regime of perfect fluids on spacetimes with Kasner big bang asymptotics is that, while the $U$-variables are expected to be well-behaved heuristically (see Section~\ref{sec:convvariables.N}), their evolution equations \eqref{eq:Uevol} are not uniformly symmetrisable. By contrast, the variable $Z$ introduced in Section~\ref{sec:variablescont.N} obeys a uniformly symmetrisable evolution system \eqref{eq:Zevol} (provided $u$ and $z$ are positive and certain consistency conditions hold), but this system contains singular terms with the ``wrong'' sign for the Fuchsian analysis. In the proof of Theorem~\ref{thm:Euler1} we have addressed this by constructing a combined system \eqref{eq:Wevol} for the variable $W$ in \eqref{eq:Wdef}, as stated in Proposition~\ref{prop:Fuchsiansystem}, in which these two sets of variables  play different roles as lower-order and top-order variables, respectively. However, as a consequence of this construction, the restriction \eqref{eq:lem.restr} arises, which in turn implies the restriction \eqref{eq:theorem1restr}. This restriction can be problematic, especially in the low-speed-of-sound case and for backgrounds that converge relatively slowly to their Kasner big-bang limits at $t=0$.

Intuitively, the origin of the problematic ``wrong sign'' for the $Z$-variable is the factor $t^{-\Ltt}$ in the definition \eqref{U2Z} of $z_i$. This factor is needed in order to derive a uniformly symmetrisable system for $Z$ (and, as a consequence, for $W$, since this symmetriser is inherited by the $W$-system \eqref{eq:Wevol}). The heuristics in Section~\ref{sec:convvariables.N} suggests that $\hat h_i{}^j u_j/|u|_h=t^\Ltt \hat h_i{}^j z_j$ should be convergent at $t=0$, and hence $\hat h_i{}^j z_j$ is $O(t^{-\Ltt})$. In contrast to this, the component $\check h_i{}^j u_j/|u|_h$ is expected to decay to zero at $t=0$, and so $\check h_i{}^j z_j$ could potentially be significantly less singular at $t=0$ than $O(t^{-\Ltt})$. 

The main idea of the  approach discussed in this section now is to eliminate the most singular component $\hat h_i{}^j z_j$ from the list of variables and then identify conditions which imply that the less singular remaining component $\check h_i{}^j z_j$ \emph{actually decays} at $t=0$. 
To this end, we restrict ourselves now to spacetimes with Kasner big bang asymptotics according to Definition~\ref{def:asympKasner} for which there is a smooth covector field $\hat\kappa_i$ with time-independent orthonormal frame components such that
    \begin{equation}
        \label{eq:1deigenspace}
        \hat h_i{}^j=\hat\kappa_i\hat\kappa^j,\quad \hat\kappa_i\hat\kappa^i=1.
    \end{equation} 
This means that the eigenspace corresponding to the largest eigenvalue $P$ is one-dimensional everywhere on $\Sigma$. Under this assumption we can write
    \begin{equation}
        \label{eq:zi.oned}
        z_i=\zeta \hat\kappa_i+\check z_i,\quad\check z_i\hat\kappa^i=0.
    \end{equation}
The normalisation condition \eqref{eq:znormalisationconstr} implies 
    \begin{equation}
        \label{eq:zeta.oned}
        \zeta=t^{-\Ltt}\sqrt{1-t^{2\Ltt}\check z_i\check z^i}.
    \end{equation}
As we generally assume that $\hat u_i$ does not vanish anywhere we do not need to worry about possible sign changes of $\zeta$ and can hence restrict to the positive sign in \eqref{eq:zeta.oned} without loss of generality.

The key idea is now to use \eqref{eq:zi.oned} and \eqref{eq:zeta.oned} to eliminate the problematic $\hat h_i{}^j z_j$-component from our list of variables. By also imposing a sufficiently strong anisotropy condition (hence the label ``strong anisotropy'') we can then ensure that the remaining component $\check h_i{}^j z_j$ decays (despite the $t^{-\Ltt}$-factor in the definition of $z_i$).
\begin{thm}
    \label{thm:Euler2}
    Let $(M,g_{\mu\nu})$ be a spacetime with Kasner big bang asymptotics as in Definition~\ref{def:asympKasner}. 
    Suppose that there is a smooth unit covector field $\hat\kappa_i$ with time-independent orthonormal frame components such that \eqref{eq:1deigenspace} holds for the eigenspace of the largest eigenvalue $P$ of the asymptotic rescaled Weingarten map $\Ktt_i{}^j$ in \eqref{eq:Kttdef}.   
    Pick $\mathcal P_0>0$ and $c_s^2\in (0,1)$ such that
    \begin{equation}
        \label{eq:theorem2restr}
        0<\frac{P(x)-c_s^2}{1-c_s^2}<\min\{1,q(x)+\bar p(x),P(x)-\check P(x)\},
    \end{equation}
    for each point $x$ on $\Sigma$, where $\check P$ is the largest eigenvalue of $\check h_k{}^i\Ktt_i{}^j\check h_j{}^l$ at each point in $\Sigma$.
    Given arbitrary fluid initial data $\nu_{\mathbf{0},i}\in H^{k}(\Sigma)$ and $\rho_{\mathbf{0}}\in H^{k}(\Sigma)$ with $k>N/2+2$ such that $|\nu_\mathbf{0}|^2_h<1$ and neither $\rho_{\mathbf{0}}$ nor $\hat \kappa^j\nu_{\mathbf{0},j}$ vanish anywhere on $\Sigma$, the Cauchy problem of the Euler equations for a perfect fluid with linear equation of state \eqref{eq:EOS} and speed of sound parameter $c_s^2$ launched from $t=T_0>0$ with $T_0$ sufficiently small has a unique global classical solution  $\rho,\Gamma,\nu_i\in C^0((0,T_0],H^{k}(\Sigma))\cap C^1((0,T_0],H^{k-1}(\Sigma))\cap L^\infty((0,T_0],H^{k-1}(\Sigma))$, and $\rho>0$ and $0<\nu_i\nu^i<1$ everywhere on $(0,T_0]\times\Sigma$.

    Moreover, there exist non-negative functions $u_*,z_*\in H^{k-2}(\Sigma)$  and a constant $\kappa>0$ such that
    \begin{equation}
        \label{eq:thm2.rhoestimate}
        \Bnorm{t^{(1+c_s^2)\frac{1-P}{1-c_s^2}}\rho-T_0^{(1+c_s^2)\frac{1-P}{1-c_s^2}} \rho_\mathbf{0}(1+u_*)^{(1+c_s^2)/c_s^2}}_{H^{k-2}(\Sigma)}\lesssim t^\kappa,
    \end{equation}
    \begin{equation}
        \label{eq:thm2.Lorentzestimate}
        \Bnorm{t^{(P-c_s^2)/(1-c_s^2)}\Gamma-T_0^{(P-c_s^2)/(1-c_s^2)}\frac{ |\nu_\mathbf{0}|}{\sqrt{1-|\nu_\mathbf{0}|^2}}\frac{1+z_{*}}{1+u_*}}_{H^{k-2}(\Sigma)}
        \lesssim t^\kappa,
    \end{equation}
    and
    \begin{equation}
        \label{eq:thm2.tiltestimate}
        \Bnorm{\nu_i(t)-\hat \kappa_i}_{H^{k-2}(\Sigma)}
    \lesssim t^\kappa,
    \end{equation}
    for all $t\in (0,T_0]$.
    This means that the fluid represented by this solution is \emph{asymptotically extremely-tilted} near the big bang, in the sense that the tilt vector field $\nu_i$ converges to a unit vector in the eigenspace of the largest eigenvalue $P$ of $\Ktt_i{}^j$ at each point in $\Sigma$ and that the Lorentz factor diverges in the limit $t\searrow 0$. The implicit constants here depend on the spacetime, the value of $c_s^2$ and the fluid initial data, but in particular not on $T_0$.
\end{thm}

Before we discuss the proof of this theorem in Section~\ref{sec:proof2}, we proceed with a few remarks about Theorem~\ref{thm:Euler2}. Theorem~\ref{thm:Euler2} stipulates that the background spacetime $(M,g_{\mu\nu})$ is \emph{strongly anisotropic} in the sense that the $P$-eigenspace is one dimensional on the one hand, and, on the other hand, that the inequality
\begin{equation}
    \label{eq:eigenvalueseparation}
    P-\check P>\frac{P-c_s^2}{1-c_s^2}
\end{equation}
implied by \eqref{eq:theorem2restr} holds at each point in $\Sigma$: Since $P$ is the largest eigenvalue of the rescaled Weingarten map and $\check P$ the second-largest eigenvalue, this inequality can be understood as the requirement that the gap between $P$ and all other Kasner exponents must be larger than the right side of \eqref{eq:eigenvalueseparation} at each point in $\Sigma$. This rules out, for example, the asymptotically isotropic case where all eigenvalues are equal. Interestingly, the closer $P$ is to the critical value $P=c_s^2$, the smaller the eigenvalue gap $P-\check P$ is required to be according to \eqref{eq:eigenvalueseparation}. For instance, backgrounds that are almost isotropic (in the sense that $P$ and $\check P$ are close to $1/N$) are admitted by \eqref{eq:eigenvalueseparation} provided with have a near-radiation fluid, that is, $c_s^2$ is close to $1/N$. 

One of the particular advantages of Theorem~\ref{thm:Euler2} over Theorem~\ref{thm:Euler1} is that, if \eqref{eq:eigenvalueseparation} holds, \eqref{eq:theorem2restr} is less restrictive than \eqref{eq:theorem1restr}. One can write the inequality \eqref{eq:theorem1restr} in the same form as \eqref{eq:theorem1restr.2} and \eqref{eq:theorem1restr.2.cstar} but where $\Ltt_*$ is now the (in general larger) quantity $\Ltt_*=\min\{1,q+\bar p\}$. 
It is particularly noteworthy that the restriction for $q+\bar p$ in \eqref{eq:theorem2restr} is caused by a single term $t^{-\Ltt} \check h_{i}{}^j\dot n_j$ in the Euler equations only. If this particular term was absent for some reason, for example because $h_{i}{}^j\dot n_j=0$, there would be \emph{no restriction on $q$ and $\bar p$ in Theorem~\ref{thm:Euler2} whatsoever} (beyond positivity). 

Another improvement of Theorem~\ref{thm:Euler2} over Theorem~\ref{thm:Euler1} is that it yields slightly better regularity of the solutions. Given initial data in $H^k(\Sigma)$, Theorem~\ref{thm:Euler1} requires $k>N/2+3$, and while the solution is continuous with respect to $H^k(\Sigma)$ and continuously differentiable with respect to $H^{k-1}(\Sigma)$, it is in general only bounded with respect to $H^{k-2}(\Sigma)$. For Theorem~\ref{thm:Euler2}, given data in $H^k(\Sigma)$, we only need to assume $k>N/2+2$, and while the solution is then also continuous with respect to $H^k(\Sigma)$ and continuously differentiable with respect to $H^{k-1}(\Sigma)$, it is bounded with respect to the stronger norm $H^{k-1}(\Sigma)$.

\subsection{Proof of Theorem~\ref{thm:Euler2}}
\label{sec:proof2}

\phantom{.}

\noindent \underline{Variable transformation:} 
In this proof we want to exploit \eqref{eq:zi.oned} and \eqref{eq:zeta.oned} and obtain a symmetrisable hyperbolic PDE system by formulating a corresponding variable transformation in the spirit of Section~\ref{sec:nonlineartransf}: 
the new variable vector $\check Z$ is related to X by 
    \begin{equation}
        \label{eq:checkZ2Z}
        X=\check\Phi(t,x,\check Z)=\begin{pmatrix}
            u\\
             z \sqrt{1-t^{2\Ltt}\check z_i\check z^i} \hat\kappa_i+z t^\Ltt\check z_i
        \end{pmatrix},\quad
        \check Z=\begin{pmatrix}
            u\\
            z\\
            \check z_i
        \end{pmatrix}.
    \end{equation}
We recall the definition of $X$ in \eqref{eq:Xdef} and note that \eqref{eq:checkZ2Z} is a combination of \eqref{Z2U}, \eqref{eq:zi.oned} and \eqref{eq:zeta.oned}. The map
\begin{equation}
    \label{eq:X2checkZ}
    \check Z=\check\Phi^{-1}(t,x,X)=\begin{pmatrix}
        u\\
        |u|_h\\
        t^{-\Ltt}\frac{\check h_i{}^j u_j}{|u|_h}
    \end{pmatrix}
\end{equation}
is the inverse
provided the consistency condition
\begin{equation}
    \label{eq:zcheckconsist}
    \check z_i\hat\kappa^i=0
\end{equation}
holds.
Similar to our previous discussions it is convenient to consider solutions of our PDEs even when \eqref{eq:zcheckconsist} is violated. In the sense of Section~\ref{sec:nonlineartransf}, we regard the map $\check\Phi$ in \eqref{eq:checkZ2Z} as a variable transformation in the submersion case where $s=2+N$ and $\bar s=\sigma=1+N$.
Doing this, a  calculation that is very similar to the one that led us from \eqref{eq:PDEgen.transf.nonsymm.N} to \eqref{eq:Zevol}, yields
\begin{equation}
    \label{eq:Zcheckevol}
    \partial_t \check Z+\alpha \check{\mathtt C}^l\Dc_l\check Z
    =\frac 1t\check\Hsc \check{\mathbb P}\check Z+t^{-1+\check Q}\check\Htt(t,x,\check Z),
\end{equation}
with
\begin{align}
    \check{\mathtt C}^l
    =&\frac{1}{\sqrt{t^{2\Ltt} u^2+z^2}\mathcal A}
    \begin{pmatrix}
        \begin{smallmatrix}
            \frac{z \left((1-2 {c_s^2}) t^{2\Ltt} u^2+(1-{c_s^2}) z^2\right) \left({\hat\kappa}^{l} \mathcal A^2+t^{\Ltt}  \check z^{l} {\mathcal A}\right)}{(1-{c_s^2}) z^2+t^{2\Ltt} u^2} 
            & \frac{{c_s^2} u^3 \left(t^{2\Ltt} {\hat\kappa}^{l} \mathcal A^2+\check z^{l} t^{3 \Ltt} {\mathcal A}\right)}{(1-{c_s^2}) z^2+t^{2\Ltt} u^2} 
            & -\frac{{c_s^2} t^{\Ltt}  u z \left(t^{2\Ltt} u^2+z^2\right) \left(t^{\Ltt}  \check z^k {\hat\kappa}^{l}-{\mathcal A} {\check h}^{l k}\right)}{(1-{c_s^2}) z^2+t^{2\Ltt} u^2} \\
            u \left(\check z^{l} t^{3 \Ltt} {\mathcal A}+{\hat\kappa}^{l} t^{2\Ltt}\mathcal A\right) & z \left({\hat\kappa}^{l}\mathcal A^2+t^{\Ltt}  \check z^{l} {\mathcal A}\right) & 0 \\
            \frac{u \left(t^{\Ltt}  {\mathcal A} {\check h}_{m}{}^{l}-\check z_m \left(\check z^l t^{3 \Ltt} {\mathcal A}+t^{2\Ltt} {\hat\kappa}^{l} \mathcal A^2\right)\right)}{z} & 0 & z {\check h}_{m}{}^{k} \left({\hat\kappa}^{l}\mathcal A^2+t^{\Ltt}  \check z^{l} {\mathcal A}\right) 
        \end{smallmatrix}
    \end{pmatrix}\notag\\
    =&\begin{pmatrix}
        {\hat\kappa}^{l}+O\left(t^{\Ltt}\right) & O\left(t^{2\Ltt}\right) & O\left(t^{\Ltt}\right) \\
        O\left(t^{2\Ltt}\right) & {\hat\kappa}^{l}+O\left(t^{\Ltt}\right) & 0 \\
        O\left(t^{\Ltt}\right) & 0 & {\hat\kappa}^{l} {\check h}_{m}{}^{k}+O\left(t^{\Ltt}\right)
    \end{pmatrix},\label{eq:defcheckCl}\\
    \label{eq:defcheckHsc}
    \check\Hsc
    =&\mathrm{diag}\left(
        1, 1, \check\eta  {\hat\kappa}_{m} {\hat\kappa}^{p}+{\check L}_{m}{}^{p}-\Ltt {\check h}_{m}{}^{p}
    \right),\\
    \check{\mathbb P}
    =&\mathrm{diag}\left(
        0, 0, h_p{}^k
    \right),\label{eq:Pbbcheck}
\end{align}
where we have exploited the freedoms discussed in Section~\ref{sec:nonlineartransf} to introduce an arbitrary constant $\check\eta>0$ (whose presence is inconsequential for any solution that satisfies \eqref{eq:zcheckconsist}), we employ the $O$-notation from Section~\ref{sec:analysisbg.N}, and we use the abbreviation
\begin{equation}
    \label{eq:mathcalAdef}
    \mathcal A:=\sqrt{1-t^{2\Ltt} \check z_r\check z^r}.
\end{equation}
Following the same line of arguments as before, we conclude that provided the function $\check Q:\Sigma\rightarrow\Rbb$ is smooth and satisfies the bound
\begin{equation}
    \label{eq:Qchinequal.1}
    0<\check Q<\min\{q,\bar q, \bar p, 2\Ltt, q+\bar p-\Ltt\}
\end{equation}
at each point of $\Sigma$,
the map $\check\Htt$ in \eqref{eq:Zcheckevol} 
can be identified with a uniformly bounded polynomial in the sense of Section~\ref{sec:analysisbg.N} for 
\begin{equation}
    \check\chi(t,x,\check Z)=(u, z, 1/z, \check z_i, \mathcal A, 1/\mathcal A, \sqrt{u^2 t^{2\Ltt}+z^2}, 1/\sqrt{u^2 t^{2\Ltt}+z^2})^T,
\end{equation}
and 
\begin{equation}
    \label{eq:checkOmegadef}
    \check\Omega=\left\{\check Z\in\Rbb^{2+N} \,\left|\, u>u_*,\, z>z_*,\,\check z_{i}\check z^i<\check z_*^2\right.\right\}
\end{equation}
given any constants $u_*,z_*,\check z_*>0$ and provided we make the restriction
\begin{equation}
    \label{eq:T0restriction.check}
    T_0<\check z_*^{-\Ltt},
\end{equation}
which ensures that $\mathcal A$ in \eqref{eq:mathcalAdef} is well-defined.
Notice that the right-hand side of the inequality \eqref{eq:Qchinequal.1} is strictly positive at each point of $\Sigma$ as a consequence of \eqref{eq:theorem2restr}. Most notably, the quantity $q+\bar p-\Ltt$ arises in \eqref{eq:Qchinequal.1} from the single term $t^{-\Ltt} \check h_{i}{}^j\dot n_j$ in the Euler equations.

It also follows from the construction in Section~\ref{sec:nonlineartransf} is that \eqref{eq:Zcheckevol} is symmetrisable, and in fact, that the symmetriser 
\begin{align}
    \check\Stt=
    \mathcal A^2\left(
    \begin{smallmatrix}
        \frac{ \left((1-{c_s^2}) z^2+t^{2\Ltt} u^2\right)}{{c_s^2}} & 0 & 0 \\
        0 & u^2  & 0 \\
        0 & 0 &  z^2\frac{t^{2\Ltt}  \check z^{p} \check z^{q} \left(t^{2\Ltt} u^2+z^2\right) {\check h}_{p}{}^{j} {\check h}_{q}{}^{k}}{\mathcal A^2}
        +z^2\left(t^{2\Ltt} u^2 +z^2\right) {\check h}^{j k}+{\hat\kappa}^{j} {\hat\kappa}^{k}
    \end{smallmatrix}
    \right)\\
    =
    \begin{pmatrix}
        \frac{(1-{c_s^2}) z^2}{{c_s^2}}+O\left(t^{2\Ltt}\right) & 0 & 0 \\
        0 & u^2+O\left(t^{2\Ltt}\right) & 0 \\
        0 & 0 & z^4 {\check h}^{j k}+{\hat\kappa}^{j} {\hat\kappa}^{k}+O\left(t^{2\Ltt}\right),
    \end{pmatrix}\notag
\end{align}
is uniformly positive definite so long as $t\in (0,T_0]$ and $\check Z\in\check\Omega$.

\phantom{.}

\noindent \underline{The initial value problem and local-in-time existence and uniqueness:} 
Consider now arbitrary constants $T_0>0$ and $\check\eta>0$. We pick initial data $\nu_{\mathbf{0},i}$ and $\rho_{\mathbf{0}}$ which satisfy the hypothesis of 
Theorem~\ref{thm:Euler2} and define 
\begin{equation}
    \label{eq:Zcheckdata} 
    \check Z(T_0,x)=\check Z_{\mathbf{0}}(x)
    =\begin{pmatrix}
        u_\mathbf{0}\\
        |u_\mathbf{0}|_h\\
        T_0^{-\Ltt}\frac{\check h_i{}^j u_{\mathbf{0},j}}{|u_\mathbf{0}|_h}
    \end{pmatrix},
\end{equation}
in full agreement with \eqref{eq:X2checkZ},
where $u_\mathbf{0}$ and $u_{\mathbf{0},i}$ are given by \eqref{eq:Xdata}. This initial data set is in $H^{k}(\Sigma)$ and satisfies \eqref{eq:zcheckconsist} for $t=T_0$. The vector \eqref{eq:Zcheckdata} can be guaranteed to be in $\check\Omega$ (cf.\ \eqref{eq:checkOmegadef}) provided we make appropriate choices for the positive constants $u_*$, $z_*$ and $\check z_*$ and assume that \eqref{eq:T0restriction.check} holds.

Because \eqref{eq:Zcheckevol} is symmetrisable, it follows from \cite[Proposition~1.4, Chapter~16]{taylor2011} that the initial value problem of \eqref{eq:Zcheckevol} with initial data $\check Z_{\mathbf{0}}$ has a unique classical solution
$\check Z$ on $(T_1,T_0]\times\Sigma$ for some $T_1\in [0,T_0)$ and
\begin{equation}
    \label{eq:checkZlocalsol}
    \check Z \in C^0\bigl((T_1,T_0],H^{k}(\Sigma)\bigr)\cap C^1\bigl((T_1,T_0],H^{k-1}(\Sigma)\bigr).
\end{equation}
The consistency condition \eqref{eq:zcheckconsist} holds everywhere on $(T_1,T_0]\times\Sigma$, and the values of $\check Z$ are contained in $\check\Omega$ for all $(t,x)\in (T_1,T_0]\times\Sigma$.

\phantom{.}

\noindent \underline{The leading-order term and the Fuchsian PDE for the remainder:} 
On $(0,T_0]\times\Sigma$ we next define  the function 
\begin{equation}
    \label{eq:Zcheckcircformula}
    \mathring{\check Z}(t,x)=\mathrm{diag}\Bigl(
            1, 1, e^{(\check L-\Ltt\check h) \log \frac{t}{T_0}}
        \Bigr)\check Z_{\mathbf{0}}.
\end{equation}
We observe that
\begin{equation}
    \mathring {\check Z}(T_0,x)={\check Z}_{\mathbf{0}}(x),
\end{equation}
and the truncated Euler equations
\begin{equation}
    \label{eq:backgroundODE.check}
    \partial_t \mathring {\check Z}=\frac 1t\check\Hsc \check{\mathbb P} \mathring {\check Z}
\end{equation}
hold for all $(t,x)\in(0,T_0]\times\Sigma$. The intuition for the definition \eqref{eq:Zcheckcircformula} is the same as for the definition \eqref{eq:Wcircformula} in the proof of Theorem~\ref{thm:Euler1}. 
Given this we define the remainder 
\begin{equation}
    \label{eq:Zcheckbardef}
    \bar{\check Z}=\check Z-\mathring {\check Z}.
\end{equation}

Because $\check Z$ is a solution of \eqref{eq:Zcheckevol} which agrees with ${\check Z}_{\mathbf{0}}$ at $t=T_0$, a straightforward computation exploiting \eqref{eq:backgroundODE.check} reveals that $\bar{\check Z}$ is a solution of
\begin{equation}
    \label{eq:Zcheckevol.bar}
    \partial_t \bar{\check Z}+\alpha \check{\mathtt C}^l\Dc_l\bar{\check Z}
    =\frac 1t\check\Hsc \check{\mathbb P}\bar{\check Z}
    +t^{-1+\bar{\check Q}}\tilde\Htt+t^{-1+\bar{\check Q}}\Htt_0,
\end{equation}
for all $(t,x)\in (T_1,T_0]\times\Sigma$ with 
\begin{equation}
    \label{eq:barvanishingID.check}
    \bar{\check Z}(T_0,x)=0,\quad x\in\Sigma.
\end{equation}
Here, we write
\begin{equation}
    t^{-1+\bar{\check Q}}\tilde\Htt+t^{-1+\bar{\check Q}}\Htt_0
    :=+t^{-1+\check Q}\check\Htt
    -\alpha \check{\mathtt C}^l\Dc_l\mathring{\check Z},
\end{equation}
given the previously defined maps in \eqref{eq:Zcheckevol} with \eqref{eq:defcheckCl} -- \eqref{eq:Pbbcheck}; the two maps on the left side are (up to the choice of $\bar{\check Q}$ discussed below) uniquely determined by the condition that
\begin{equation}
    \Htt_0(t,x,\mathring {\check Z},0)=0.
\end{equation}
As several times before, it is then a straightforward but lengthy calculation to establish that,
provided we choose $\bar{\check Q}$ to be a smooth function bounded by
\begin{equation}
    \label{eq:Qbarinequal.1.check}
    0<\bar{\check Q}<\min\{\check Q,\bar p\}=\min\{q,\bar q, \bar p, 2\Ltt, q+\bar p-\Ltt\},
\end{equation}
the map $\Htt_0$ inherits the uniformly bounded polynomial properties of $\check\Htt$ and $\check{\mathtt C}^l$ given for \eqref{eq:Zcheckevol} for all vectors $\bar{\check Z}$ in the $R$-ball $\bar{\check\Omega}$ in $\Rbb^{2+N}$ where $R>0$ is chosen so that $\mathring Z(t,x)+\bar{\check Z}$ is contained in $\check\Omega$ (cf.\ \eqref{eq:checkOmegadef}) for all $(t,x)\in (0,T_0]\times\Sigma$.

\phantom{.}

\noindent \underline{Fuchsian analysis and global existence:} In full analogy to the proof of Theorem~\ref{thm:Euler1} the central step is now to verify the \emph{coefficient assumptions} \cite[Section~A.1]{BeyerOliynyk:2020} and then to apply the Fuchsian theorem \cite[Theorem~A.2, Remark~A.3 (ii)]{BeyerOliynyk:2020}.

The discussion of the coefficient assumptions here is now very similar to the one in the proof of Theorem~\ref{thm:Euler1}, and we will therefore omit most of the details here. For now, we assume that $T_0$ is an arbitrary constant satisfying \eqref{eq:T0restriction.check}. As before we will want to use the size of $T_0$ to control certain error terms that arise in our reference-independent approach. We therefore pay particular attention to determining which quantities can be picked \emph{independently of $T_0$}. We identify the constant $R>0$ in \cite[Section~A.1]{BeyerOliynyk:2020} with the constant $R$ we chose above to define the set $\bar{\check\Omega}$, but will make use of the opportunity to shrink $R$ whenever this is useful for our arguments. The constant $\Rc>0$ in \cite[Section~A.1]{BeyerOliynyk:2020} is an arbitrary constant so far, but later it will have to be chosen so that $\sup_{(0,T_0]\times\Sigma}|\mathring {\check Z}|<\Rc$ and $\sup_{(0,T_0]}\norm{\Dc\mathring {\check Z}}_{H^{k-2}(\Sigma)}<\Rc$. Equation~\eqref{eq:Zcheckcircformula} allows us to show that it is possible to pick $\Rc$ so that this holds independently of $T_0$. We define the vector bundles $Z_1$ and $Z_2$ in \cite[Section~A.1]{BeyerOliynyk:2020} as the trivial $\Rbb^{2+N}$- and $\Rbb^{2+3N+N^2}$-bundles over $\Sigma$, respectively, and consider $w_1=\mathring {\check Z}(t)$ as a section in a bounded subset $\Zc_1$ of $Z_1$ and $w_2=(\mathring {\check Z}(t),\Dc_\Omega \mathring {\check Z}(t))$ as a section in a bounded subset $\Zc_2$ of $Z_2$. As discussed above, because of the particular form of $\mathring {\check Z}$, uniform bounds for $\Zc_1$ and $\Zc_2$ over $(0,T_0]\times\Sigma$ can be chosen independently of $T_0$. Given all these choices now and especially observing \eqref{eq:theorem2restr}, our verification of the coefficient assumptions follows precisely the same steps as in the proof of Theorem~\ref{thm:Euler2}. All choices can be made independently of $T_0$ and $\lambda_3=\beta_1=\beta_2=O(R)$ while all other $\beta_i$ can be chosen as zero. The same argument which allowed us to deal with time derivatives of $w_1$ in the analysis of the map $\Div \! B$ in a $T_0$-independent way in the proof of Theorem~\ref{thm:Euler1} also applies here.

Now we are in the position to apply the Fuchsian theorem \cite[Theorem~A.2, Remark~A.3 (ii)]{BeyerOliynyk:2020}.
We identify the differentiability parameter $k$ in \cite[Theorem~A.2, Remark~A.3 (ii)]{BeyerOliynyk:2020} with the value $k-1$ here. This means that our assumption $k>N/2+2$ in Theorem~\ref{thm:Euler2} is identical to the condition $k>N/2+1$ in \cite[Theorem~A.2, Remark~A.3 (ii)]{BeyerOliynyk:2020}. 

The argument we give in the proof of Theorem~\ref{thm:Euler1} to resolve the issue regarding the condition 
\begin{equation*}
    \sup_{0<t<T_0} \max\Bigl\{ \norm{\Dc w_1(t)}_{H^{k-2}(\Sigma)},
    \norm{t^{1-q}\del{t}w_1(t)}_{H^{k-2}(\Sigma)}\Bigr\} < \frac{\Rsc}{C_{\text{Sob}}},
\end{equation*}
in \cite[Theorem~A.2]{BeyerOliynyk:2020} in a $T_0$-independent way also applies here. The inequality for $\gamma$ in \cite[Theorem~A.2, Remark~A.3 (ii)]{BeyerOliynyk:2020} can also be satisfied in the same $T_0$-independent way as before.
    
It is now a consequence of \cite[Theorem~A.2, Remark~A.3 (ii)]{BeyerOliynyk:2020} and \eqref{eq:barvanishingID} that there is a $\delta_0$, whose value is independent of $T_0$, such that the initial value problem 
of \eqref{eq:Zcheckevol.bar} with initial condition \eqref{eq:barvanishingID.check} has a classical solution \emph{globally} in time, as long as the smallness condition
\begin{equation} 
    \int_0^{T_0} s^{-1+\bar{\check Q}} \norm{\tilde\Htt(s,\mathring {\check Z}(s))}_{H^{k-1} (\Sigma)} ds<\delta_0.
\end{equation}
To satisfy the smallness condition, it suffices to shrink $T_0>0$.
The local-in-time solution $\check Z$ found above extends as a classical solution down to $t=0$, \eqref{eq:checkZlocalsol} holds with $T_1=0$, the consistency condition \eqref{eq:zcheckconsist} holds and the values of $\check Z$ are contained in $\check\Omega$ (cf.\ \eqref{eq:checkOmegadef}) for all $(t,x)\in (0,T_0]\times\Sigma$. The Fuchsian theorem \cite[Theorem~A.2, Remark~A.3 (ii)]{BeyerOliynyk:2020} also implies the energy estimate
\begin{equation}
    \label{eq:FuchsianEnergyEstimate.check}
    \norm{\bar {\check Z}(t)}_{H^{k-1}(\Sigma)}  \lesssim s^{-1+\bar{\check Q}} \norm{\tilde\Htt(s,\mathring {\check Z}(s))}_{H^{k-1} (\Sigma)} ds
\end{equation}
for all $t\in (0,T_0]$. Together with \eqref{eq:Zcheckcircformula} and \eqref{eq:theorem2restr} this yields
\begin{equation}
    \label{eq:checkZglobalsol}
    \check Z \in C^0\bigl((0,T_0],H^{k}(\Sigma)\bigr)\cap C^1\bigl((0,T_0],H^{k-1}(\Sigma)\bigr)\cap L^\infty\bigl((0,T_0],H^{k-1}(\Sigma)\bigr).
\end{equation}
The map $X$ found from \eqref{eq:checkZ2Z} yields the physical fluid variables via \eqref{eq:Gammanu.Var2}. This yields
\begin{equation}
    \rho,\Gamma, \nu_i\in C^0((0,T_0],H^{k}(\Sigma))\cap C^1((0,T_0],H^{k-1}(\Sigma))\cap L^\infty\bigl((0,T_0],H^{k-1}(\Sigma)\bigr).
\end{equation}

Using the same argument used for the proof of Theorem~\ref{thm:Euler1}, which allowed us to improve the decay and convergence statement of \cite[Theorem~A.2, Remark~A.3 (ii)]{BeyerOliynyk:2020}, we obtain the existence of a constant $\kappa>0$ and of non-negative functions $u_*$ and $z_*$ in $H^{k-2}(\Sigma)$ 
such that 
\begin{equation}
    \norm{\bar{\check z}_i(t)}_{H^{k-2}(\Sigma)}\lesssim t^{\kappa},
\end{equation}
and
\begin{align}
    \label{eq:fluidestimate.1.check}
    \Bigl\|&\ub(t)-{T_0^{P-L} \left(\frac{c_s^2}{\mathcal P_0}\right)^{c_s^2/(1+c_s^2)}\rho_\mathbf{0}^{c_s^2/(1+c_s^2)}}u_*\Bigr\|_{H^{k-2}(\Sigma)}\\
    +&\left\|{\bar{z}(t)-\frac{T_0^{P} \left(\frac{c_s^2}{\mathcal P_0}\right)^{c_s^2/(1+c_s^2)}\rho_\mathbf{0}^{c_s^2/(1+c_s^2)}|\nu_\mathbf{0}|}{\sqrt{1-|\nu_\mathbf{0}|^2}}z_{*}}\right\|_{H^{k-2}(\Sigma)}\lesssim t^\kappa\notag
\end{align}
for all $t\in (0,T_0]$. 

\phantom{.}

\noindent \underline{Asymptotics of the physical fluid variables:} 
Recall from \eqref{eq:Zcheckcircformula} and \eqref{eq:Zcheckbardef} that $u=\mathring u+\bar u$ and $z=\mathring z+\bar z$ and that $\mathring u$ and $\mathring z$ are the constants in time given by \eqref{eq:Zcheckdata} and \eqref{eq:Xdata}. 
By shrinking $T_0$ further if necessary, we can exploit \eqref{eq:FuchsianEnergyEstimate.check} to control the sizes of $\norm{u_*}_{H^{k-2}(\Sigma)}$ and $\norm{z_*}_{H^{k-2}(\Sigma)}$. In this way we can ensure that 
the limits of $u$ and $z$ are non-vanishing on $\Sigma$. 
We then obtain
\begin{align}
    \left\|{u(t)-{T_0^{c_s^2\frac{1-P}{1-c_s^2}} \left(\frac{c_s^2}{\mathcal P_0}\right)^{c_s^2/(1+c_s^2)}\rho_\mathbf{0}^{c_s^2/(1+c_s^2)}}(1+u_*)}\right\|_{H^{k-2}(\Sigma)}&\lesssim t^\kappa,\\
    \left\|{u_i(t)-\frac{T_0^{P} \left(\frac{c_s^2}{\mathcal P_0}\right)^{c_s^2/(1+c_s^2)}\rho_\mathbf{0}^{c_s^2/(1+c_s^2)}|\nu_\mathbf{0}|}{\sqrt{1-|\nu_\mathbf{0}|^2}}(1+z_*)\hat\kappa_i}\right\|_{H^{k-2}(\Sigma)}
    &\lesssim t^\kappa,\\
    \norm{\check h_i{}^j u_j(t)}_{H^{k-2}(\Sigma)}&\lesssim t^\kappa,
\end{align}
for all $t\in (0,T_0]$,
where we use the freedom to adapt the value of the constant $\kappa$ if necessary.
The required estimates \eqref{eq:thm2.rhoestimate}, \eqref{eq:thm2.Lorentzestimate} and \eqref{eq:thm1.tiltestimate} then follow directly by combining the estimates above with \eqref{eq:Gammanu.Var2} and the Moser inequality in the same way as in the proof of Theorem~\ref{thm:Euler2}.

\appendix
\section{Nonlinear variable transformations}
\label{sec:nonlineartransf}
\renewcommand{\Vsy}{V}
\renewcommand{\Vtm}{{V_0}}
\renewcommand{\Vsp}{{\Vsy}}
\renewcommand{\VspU}[1]{\Vsp^{#1}} 
\renewcommand{\VspD}[1]{\Vsp_{#1}}
\renewcommand{\Vspnorm}{|\Vsp|_h}
\renewcommand{\CoeffSymmPre}[1]{{\tt {#1}}{}}
\renewcommand{\CoeffSymm}{\CoeffSymmPre{a}}
\renewcommand{\SourceSymmPre}[1]{{\tt {#1}}}
\renewcommand{\SourceSymm}{\SourceSymmPre{f}}
\renewcommand{\Coeff}{\CoeffSymmPre{A}}
\renewcommand{\Source}{\SourceSymmPre{F}}
\renewcommand{\NewVar}{U}
\newcommand{\Kb}{\bar K} 

Let $T_0$ and $\Sigma$ be as discussed in Section~\ref{sec:background.N}. Consider a quasilinear PDE of the form\footnote{Recall the abstract index as well as the index-free notations introduced in Section~\ref{sec:analysisbg.N}. }
\begin{equation}
    \label{eq:PDEgen}
    \CoeffSymmPre{\bar{\tt a}}^{0\Jb \Kb} \partial_t \Ub_{\Kb}
    +\CoeffSymmPre{\bar{\tt a}}^{l\Jb \Kb} \Dc_l{} \Ub_{\Kb}
    =\SourceSymmPre{\bar{\tt f}}^{\Jb},
\end{equation}
where all the fields are smooth maps defined for $t\in (0,T_0]$, $x\in\Sigma$, and, given any $T_1\in [0,T_0)$, the quantity $\Ub_{\bar K}$ is a classical solution on $(T_1,T_0]\times\Sigma$ taking values in some non-empty open subset $\bar\Omega$ of $\Rbb^{\bar s}$. We let $\Dc_l$ be a connection defined on $\Sigma$. 

A \emph{(nonlinear) variable transformation} $\Phi_{\bar K}(t,x,u)$ is a smooth map from the differentiable manifold $(0,T_0]\times\Sigma\times\Omega$ to $\bar\Omega$, where $\Omega$ is an open non-empty subset of $\Rbb^s$. We define
\begin{equation}
    \label{eq:VariableTrafo.abstr}
    \Ub_{\Kb}(t,x)=\Phi_{\Kb}(t,x,U(t,x))
\end{equation}
on $(T_1,T_0]\times\Sigma$, and think of $\Ub_{\Kb}$ as the \emph{original variables} -- which locally satisfy \eqref{eq:PDEgen} -- and $U_L$ as \emph{new} variables. The first goal is to express the PDE \eqref{eq:PDEgen} in terms of these new variables. 

To this end, the following notation is useful\footnote{The map $\mathbf D$ here is the same variable derivative that we introduced in \eqref{eq:defdCtt} -- \eqref{eq:defdHtt}.}
\begin{equation}
    \label{eq:derivativetransmatr}
    \mathbf D{\Phi_{\Kb}}^L(t,x,u):=(\mathbf D_{\NewVar_L}\Phi_{\Kb})(t,x,u),\,\,
    \Phi_{0,\Kb}(t,x,u)=(\partial_t\Phi_{\Kb})(t,x,u),\,\,
    \Phi_{i,{\Kb}}(t,x,u)=(\Dc_i\Phi_{\Kb})(t,x,u).
\end{equation}
In terms of this notation, we can write the PDE for the new variables implied by \eqref{eq:PDEgen} as follows 
\begin{equation}
    \label{eq:PDEgen.transf}
    \CoeffSymmPre{a}^{0JK} \partial_t \NewVar_K
    +\CoeffSymmPre{a}^{lJK} \Dc_l{} \NewVar_K
    =\SourceSymmPre{f}^J,
\end{equation}
where
\begin{align}
    \label{eq:transfcoeff.1}
    \CoeffSymmPre{a}^{0ML}&=\mathbf D\Phi_{\Jb}{}^M\CoeffSymmPre{\bar{\tt a}}^{0\Jb \Kb} \mathbf D\Phi_{\Kb}{}^L,\\
    \CoeffSymmPre{a}^{lML}&=\mathbf D\Phi_{\Jb}{}^M\CoeffSymmPre{\bar{\tt a}}^{l\Jb \Kb} \mathbf D\Phi_{\Kb}{}^L,\\
    \label{eq:transfcoeff.3}
    \SourceSymmPre{f}^M&=\mathbf D\Phi_{\Jb}{}^M\SourceSymmPre{\bar{\tt f}}^{\Jb}
    -\mathbf D\Phi_{\Jb}{}^M\CoeffSymmPre{\bar{\tt a}}^{0\Jb \Kb} \Phi_{0,\Kb}
    -\mathbf D\Phi_{\Jb}{}^M\CoeffSymmPre{\bar{\tt a}}^{l\Jb \Kb} \Phi_{l,\Kb}.
\end{align}
In our notation, any field that \emph{is not} decorated with a bar is understood as a function of $(t,x,u)$ evaluated at $(t,x,U(t,x))$, while any map that \emph{is} decorated with a bar is understood as a function of $(t,x,\ub)$ evaluated at $(t,x,\Ub(t,x))$. 

Let $\sigma: (0,T_0]\times\Sigma\times\Omega\rightarrow\Rbb$ be the rank of the $\bar s\times s$-matrix-valued Jacobi matrix function $\mathbf D{\Phi_{\Kb}}^L$. In what follows, we assume that $\sigma$ is a \emph{constant} function (whose constant value is simply written as $\sigma$). The map
\begin{equation}
    \label{eq:extendedPhi}
    (0,T_0]\times\Sigma\times\Omega\rightarrow (0,T_0]\times\Sigma\times\bar \Omega:\quad (t,x,u)\mapsto (t,x, \Phi_{\Kb}(t,x,u))
\end{equation}
is therefore \emph{a map of constant rank} between smooth manifolds.

Let us start with the most important case $s=\bar s=\sigma$. According to the inverse function theorem, the variable transformation \eqref{eq:VariableTrafo.abstr} is locally invertible, the map \eqref{eq:extendedPhi} is a local diffeomorphism, and, in particular, $\mathbf D\Phi(t,x,u)$ is invertible for all $(t,x,u)\in (0,T_0]\times\Sigma\times\Omega$. We conclude that \eqref{eq:PDEgen} and \eqref{eq:PDEgen.transf} are \emph{equivalent PDE systems} assuming appropriate choices of the subsets $\Omega$ and $\bar\Omega$, i.e., given a classical local solution $U_K$ of \eqref{eq:PDEgen.transf} yields the classical local solution $\Ub_{\Kb}$ of \eqref{eq:PDEgen} given by \eqref{eq:VariableTrafo.abstr} and vice versa.
Also, $\CoeffSymmPre{\bar{\tt a}}^{0\Jb \Kb}$ and $\CoeffSymmPre{\bar{\tt a}}^{l\Jb \Kb}$ are symmetric bilinear forms acting on $\Rbb^{\bar s}$ everywhere $\CoeffSymmPre{a}^{0ML}$ and $\CoeffSymmPre{a}^{lML}$ are symmetric bilinear forms acting on $\Rbb^{s}$. Moreover, $\CoeffSymmPre{\bar{\tt a}}^{0\Jb \Kb}$ is positive definite everywhere $\CoeffSymmPre{a}^{0ML}$ is positive definite. In this case, 
\eqref{eq:PDEgen} and \eqref{eq:PDEgen.transf} are \emph{equivalent symmetrisable PDE systems}. Finally,
if $\CoeffSymmPre{\bar{\tt a}}^{0\Jb \Kb}$ is invertible (which is the case if and only if $\CoeffSymmPre{{\tt a}}^{0J K}$ is invertible) with inverse $(\CoeffSymmPre{\bar{\tt a}}^0)^{-1}_{\Jb \Kb}$ it is sometimes useful to write \eqref{eq:PDEgen} in non-symmetrised form
\begin{equation}
    \label{eq:PDEgen.nonsymm}
    \partial_t \Ub_{\Kb}
    +\CoeffSymmPre{\bar{\tt A}}^{l}{}_{\Kb}{}^{\Lb} \Dc_l{} \Ub_{\Lb}
    =\SourceSymmPre{\Fb}_{\Kb},
\end{equation}
where
\begin{align}
    \CoeffSymmPre{\bar{\tt A}}^{l}{}_{\Kb}{}^{\Lb}
    =(\CoeffSymmPre{\bar{\tt a}}^0)^{-1}_{\Kb \Jb}\CoeffSymmPre{\bar{\tt a}}^{l\Jb \Lb},\\
    \SourceSymmPre{\Fb}_{\Kb}=(\CoeffSymmPre{\bar{\tt a}}^0)^{-1}_{\Kb \Jb}\SourceSymmPre{\fb}^{\Jb}.
\end{align}
Denoting the inverse of $\mathbf D\Phi_{\Jb}{}^M$ as 
$(\mathbf D\Phi^{-1})_P{}^{\Qb}$, we get 
\begin{equation}
    (\CoeffSymmPre{a}^{0})^{-1}_{PM}=(\mathbf D\Phi^{-1})_P{}^{\Qb}(\CoeffSymmPre{\bar a}^{0})^{-1}_{\Qb\Lb} (\mathbf D\Phi^{-1})_M{}^{\Lb},
\end{equation}
and the PDE \eqref{eq:PDEgen.transf} takes the equivalent form
\begin{equation}
    \label{eq:PDEgen.transf.nonsymm}
    \partial_t \NewVar_K
    +\CoeffSymmPre{A}^{l}{}_K{}^L \Dc_l{} \NewVar_L
    =\SourceSymmPre{F}_K,
\end{equation}
with
\begin{align}
    \CoeffSymmPre{A}^{l}{}_P{}^R&=(\CoeffSymmPre{a}^{0})^{-1}_{PM}\CoeffSymmPre{a}^{lMR}=(\mathbf D\Phi^{-1})_P{}^{\Qb}\CoeffSymmPre{\bar A}^{l}{}_{\Qb}{}^{\Kb} \mathbf D\Phi_{\Kb}{}^R,\\
    \label{eq:PDEgen.transf.nonsymm.G}
    \SourceSymmPre{F}_P&=(\CoeffSymmPre{a}^{0})^{-1}_{PM}\SourceSymmPre{f}^M
    =(\mathbf D\Phi^{-1})_{\Pb}{}^{\Qb} \left(\SourceSymmPre{\Fb}_{\Qb}
    - \Phi_{0,\Qb}
    -\CoeffSymmPre{\bar A}^{l}{}_{\Qb}{}^{\Kb} \Phi_{l,\Kb}\right).
\end{align}

The next important case is when $s\ge \bar s=\sigma$. In this case, the map \eqref{eq:extendedPhi} is a submersion instead of a local diffeomorphism. This means that \eqref{eq:VariableTrafo.abstr} is in general not invertible (not even locally), and, while $\mathbf D{\Phi}(t,x,\NewVar)$ is surjective for all $(t,x,\NewVar)\in (0,T_0]\times\Sigma\times\Omega$, it is in general not injective and its kernel $\ker\mathbf D\Phi$ is an $s-\bar s$-dimensional subspace of $\Rbb^s$. For each $(t,x,\NewVar)\in (0,T_0]\times\Sigma\times\Omega$, let $\pi^\perp$ be the projection from $\Rbb^s$ onto the subset $\ker\mathbf D\Phi$ of $\Rbb^s$, and define $\pi$ by
\begin{equation}
    \pi=\mathrm{id}-\pi^{\perp}.
\end{equation}
The map $\pi$ is the projection from $\Rbb^s$ onto the $\bar s$-dimensional direct complement subspace of $\ker\mathbf D\Phi$ in $\Rbb^s$. Suppose now there is additionally a smooth map $\mathbf F:(0,T_0]\times\Sigma\times\Omega\rightarrow\Rbb^{s-\bar s}$ such that $\ker\mathbf D \mathbf F$ is an $\bar s$-dimensional subspace of $\Rbb^s$ and
\begin{equation}
    \label{eq:kernelproperty}
    \ker\mathbf D\Phi\oplus \ker\mathbf D\mathbf F=\Rbb^s
\end{equation}
at each $(t,x,u)\in (0,T_0]\times\Sigma\times\Omega$. It follows from the implicit function theorem that given a smooth local map $\Ub(t,x)$ with values in $\bar\Omega$, the combined system
\begin{align}
    \label{eq:VariableTrafo.abstr.ext.1}
    \Ub(t,x)&=\Phi(t,x,U(t,x)),\\ 
    \label{eq:VariableTrafo.abstr.ext.2}
    0&=\mathbf F(t,x,U(t,x)),
\end{align}
has a unique smooth local solution $U(t,x)$ with values in $\Omega$. The map $\mathbf F$ therefore resolves the degeneracy of \eqref{eq:VariableTrafo.abstr} in the submersion case; it, in a sense, adds the missing information to the variable transformation ``to make it invertible''.  We call \eqref{eq:VariableTrafo.abstr.ext.2} the \emph{consistency condition}. 

Now, given the consistency condition \eqref{eq:VariableTrafo.abstr.ext.2}, we suppose that the original variable $\Ub$ locally satisfies \eqref{eq:PDEgen}. The same calculations as before yield that the new variable $U$ locally satisfies a PDE system of the same form as before \eqref{eq:PDEgen.transf} with \eqref{eq:transfcoeff.1} -- \eqref{eq:transfcoeff.3}. The main problem implied by the non-triviality of the kernel $\ker\mathbf D\Phi$ is, however, that this PDE is in general degenerate. In particular, the matrix in front of the time derivative is
\[(\mathbf D\Phi)^T \CoeffSymmPre{\bar{\tt a}}^{0}\mathbf D\Phi
=\pi^T(\mathbf D\Phi)^T \CoeffSymmPre{\bar{\tt a}}^{0}\mathbf D\Phi \pi;\]
cf.\ \eqref{eq:transfcoeff.1}. Assuming that $\CoeffSymmPre{\bar{\tt a}}^{0}$ is positive definite, the matrix here is not when viewed as a bilinear form on $\Rbb^s$. We conclude that while the original PDE system for $\Ub$ has a well-posed initial value problem, it is not obvious whether the PDE system implied for $U$ does. The goal is now to demonstrate that one can break the degeneracy of the new PDE system thanks to the consistency condition (essentially by adding $\pi^\perp$-terms to the new PDE system as discussed below) so that:
\begin{enumerate}
    \item The resulting new PDE is a closed system for $U$ which is symmetrisable whenever $\CoeffSymmPre{\bar{\tt a}}^{0}(t,x,\Phi(t,x,u))$ is symmetric and positive definite -- crucially, irrespectively of whether $U$ satisfies the consistency condition \eqref{eq:VariableTrafo.abstr.ext.2} or not.
    \item The resulting PDE system enforces the consistency condition \eqref{eq:VariableTrafo.abstr.ext.2} in the sense that the solution $U$ satisfies the consistency condition everywhere $U$ solves the PDE as long as the initial data for $U$ satisfy the consistency condition. We say that the \emph{consistency condition is preserved}.
    \item The resulting system is \emph{equivalent} to \eqref{eq:PDEgen} assuming appropriate choices of the subsets $\Omega$ and $\bar\Omega$ in the following sense. Any classical local solution $\Ub$ of \eqref{eq:PDEgen} yields $U$ by locally inverting \eqref{eq:VariableTrafo.abstr.ext.1} -- \eqref{eq:VariableTrafo.abstr.ext.2}, and $U$ is a local classical solution of the resulting new PDE system. Conversely, any classical local solution $U$ of the resulting new PDE system launched from initial data satisfying the consistency condition \eqref{eq:VariableTrafo.abstr.ext.2} yields, by the variable transformation, a local classical solution $\Ub$ of the original PDE \eqref{eq:PDEgen}.
\end{enumerate}
There are many ways one can approach these goals and consequentially, there are many resulting new PDE systems which achieve these goals.  The point is, however, that two different choices for the resulting new PDE system have the same ``physical solutions'' (that is, solutions which satisfy the consistency condition), but may have different ``unphysical solutions'' (that is, solutions which violate the consistency condition). Let us consider the following symmetric hyperbolic homogeneous system imposed for the function $\mathbf F(t,x,U(t,x))$ of the schematic form
\begin{equation}
    \label{eq:consistencyconservation}
    M \partial_t (\mathbf F)+\tilde M^l\Dc_l(\mathbf F)=\kappa\mathbf F,
\end{equation}
where $M$, $\tilde M^l$ and $\kappa$ are some given smooth maps defined on $(0,T_0]\times\Sigma\times\tilde\Omega$ which are $(s-\bar s)\times (s-\bar s)$-matrix-valued, where $\tilde\Omega$ is a non-empty open neighbourhood of $0$ in $\Rbb^{s-\bar s}$. Suppose also that $M$ and $\tilde M$ are symmetric matrices and that $M$ is positive definite on $(0,T_0]\times\Sigma\times\{0\}$. Clearly, this implies that $\mathbf F$, considered as a classical solution of \eqref{eq:consistencyconservation}, vanishes identically if its initial data imposed at $t=t_0$ vanish. Given any classical solution $\mathbf F$ of \eqref{eq:consistencyconservation} and assuming that $\mathbf F$ is a map of the form $\mathbf F=\mathbf F(t,x,U(t,x))$, the chain rule yields
\begin{equation}
    \label{eq:consistencyconservation.2}
    M \mathbf D\mathbf F \pi^\perp \partial_t U + \tilde M^l \mathbf D\mathbf F\pi^\perp \Dc_l U
    =\kappa\mathbf F
    -M \mathbf F_0-\tilde M^l \mathbf F_l
\end{equation}
using notation analogous to \eqref{eq:derivativetransmatr} and exploiting \eqref{eq:kernelproperty}. This system can be symmetrised by multiplying it from the left with $(\pi^\perp)^T (\mathbf D\mathbf F)^T$. The sum of this and the PDE obtained for $U$ above now turns out to be a \emph{non-degenerate} symmetrisable system for $U$ of the same form as \eqref{eq:PDEgen.transf}, but with coefficients
\begin{align}
    \label{eq:transfcoeff.1.ext}
    \CoeffSymmPre{a}^{0}=&\pi^T(\mathbf D\Phi)^T\CoeffSymmPre{\bar{\tt a}}^{0} \mathbf D\Phi \pi
        +(\pi^\perp)^T (\mathbf D\mathbf F)^TM \mathbf D\mathbf F\pi^\perp,\\
    \label{eq:transfcoeff.2.ext}
    \CoeffSymmPre{a}^{l}=&\pi^T(\mathbf D\Phi)^T\CoeffSymmPre{\bar{\tt a}}^{l} \mathbf D\Phi \pi
        +(\pi^\perp)^T (\mathbf D\mathbf F)^T \tilde M^l \mathbf D\mathbf F\pi^\perp,\\
    \label{eq:transfcoeff.3.ext}
    \SourceSymmPre{f}^M=&\pi^T(\mathbf D\Phi)^T\SourceSymmPre{\bar{\tt f}}
    -\pi^T(\mathbf D\Phi)^T\CoeffSymmPre{\bar{\tt a}}^{0} \Phi_{0}
    -\pi^T(\mathbf D\Phi)^T\CoeffSymmPre{\bar{\tt a}}^{l} \Phi_{l}\\
    &+(\pi^\perp)^T (\mathbf D\mathbf F)^T\kappa\mathbf F
    -(\pi^\perp)^T (\mathbf D\mathbf F)^T M \mathbf F_0
    -(\pi^\perp)^T (\mathbf D\mathbf F)^T\tilde M^l \mathbf F_l.\notag
\end{align}
The system \eqref{eq:PDEgen.transf} with these coefficients satisfies the properties listed above.

The coefficient matrices of this PDE system are block diagonal with respect to the kernel decomposition \eqref{eq:kernelproperty} of $\Rbb^s$. They may however not be block diagonal when the system is represented with respect to a particular local basis, which, in some application, may be a disadvantage. We now explore some further freedom which we possess due to the existence of a non-trivial kernel of $\mathbf D\Phi$, especially, to introduce blocks that, while they are off-diagonal with respect to the kernel decomposition, may lead to a more advantageous form in a given local basis.  We do not attempt to do this in any generality, but only consider one particular case where we attempt to introduce an off-diagonal component to \eqref{eq:transfcoeff.1.ext}. This can be achieved, for example, by adding an expression of the form
\begin{equation}
    \label{eq:aksdjaksdj}
    \pi^T a^{\pi\pi^\perp} \mathbf D\mathbf F \pi^\perp \partial_t U 
    =-\pi^T a^{\pi\pi^\perp} \mathbf F_0
\end{equation}
to \eqref{eq:PDEgen.transf} with \eqref{eq:transfcoeff.1.ext} -- \eqref{eq:transfcoeff.3.ext},
where $a^{\pi\pi^\perp}$ is any suitable matrix-valued function. Notice that \eqref{eq:aksdjaksdj} is just the $\pi^T$-projection of the special case of \eqref{eq:consistencyconservation.2} given by $M=a^{\pi\pi^\perp}$, $\tilde M^l=0$ and $\kappa=0$. Adding \eqref{eq:aksdjaksdj} to \eqref{eq:PDEgen.transf} with \eqref{eq:transfcoeff.1.ext} -- \eqref{eq:transfcoeff.3.ext} only affects the matrix $\CoeffSymmPre{a}^{0}$ in the principal part which then takes the form:
\[ \pi^T(\mathbf D\Phi)^T\CoeffSymmPre{\bar{\tt a}}^{0} \mathbf D\Phi \pi
        +(\pi^\perp)^T (\mathbf D\mathbf F)^TM \mathbf D\mathbf F\pi^\perp
        +\pi^T a^{\pi\pi^\perp} \mathbf D\mathbf F \pi^\perp;
\]
cf.\ \eqref{eq:transfcoeff.1.ext}.
Next we attempt to symmetrise this matrix 
by multiplying it from the left with a matrix of the form
\begin{equation} 
    \label{eq:abstractsymmetriser}
    B= \pi^T \pi^T
    +(\pi^\perp)^T (\mathbf D\mathbf F)^T B^{\pi^\perp\pi} \pi^T
    +(\pi^\perp)^T B^{\pi^\perp\pi^\perp} (\pi^\perp)^T.
\end{equation}
A straightforward calculation yields
\begin{gather*} 
    \pi^T(\mathbf D\Phi)^T\CoeffSymmPre{\bar{\tt a}}^{0} \mathbf D\Phi \pi
    +(\pi^\perp)^T \left((\mathbf D\mathbf F)^T B^{\pi^\perp\pi}\pi^T a^{\pi\pi^\perp} 
    +B^{\pi^\perp\pi^\perp}(\pi^\perp)^T (\mathbf D\mathbf F)^TM \right)\mathbf D\mathbf F\pi^\perp\\
    +\pi^T a^{\pi\pi^\perp} \mathbf D\mathbf F \pi^\perp
    +(\pi^\perp)^T (\mathbf D\mathbf F)^T B^{\pi^\perp\pi}\pi^T(\mathbf D\Phi)^T\CoeffSymmPre{\bar{\tt a}}^{0} \mathbf D\Phi \pi.
\end{gather*}
A necessary condition for symmetry is
\begin{equation}
    \label{eq:transformedabstracta0}
    (a^{\pi\pi^\perp})^T=B^{\pi^\perp\pi}\pi^T(\mathbf D\Phi)^T\CoeffSymmPre{\bar{\tt a}}^{0} \mathbf D\Phi.
\end{equation}
The positive definiteness of $\CoeffSymmPre{\bar{\tt a}}^{0}$ implies that this equation uniquely determines $B^{\pi^\perp\pi}$ (up to contributions that do not contribute to the matrix $B$ in \eqref{eq:abstractsymmetriser}). However, it is in general not clear how to use the remaining condition that $\left((\mathbf D\mathbf F)^T B^{\pi^\perp\pi}\pi^T a^{\pi\pi^\perp}+B^{\pi^\perp\pi^\perp}(\pi^\perp)^T (\mathbf D\mathbf F)^TM \right)\mathbf D\mathbf F$ is symmetric to determine $B^{\pi^\perp\pi}$ and $B^{\pi^\perp\pi^\perp}$, except in the special case that $\ker\mathbf D\mathbf \Phi$ is one-dimensional, that is $s-\bar s=1$. In this case, that matrix is a scalar (and therefore automatically symmetric). Given any $B^{\pi^\perp\pi}$ and $a^{\pi\pi^\perp}$, all we then need to is to require that  $B^{\pi^\perp\pi^\perp}$ is a sufficiently positive scalar in order to guarantee that the total matrix in \eqref{eq:transformedabstracta0} is symmetric and positive definite.

Given this property of \eqref{eq:transformedabstracta0} in the special case that $\ker\mathbf D\mathbf \Phi$ is one-dimensional now, let us next consider the symmetry of \eqref{eq:transfcoeff.2.ext} when this is multiplied from the left with the same matrix $B$:
\begin{equation} 
    \label{eq:sdf93fa}
    \pi^T(\mathbf D\Phi)^T\CoeffSymmPre{\bar{\tt a}}^{l} \mathbf D\Phi \pi
    +(\pi^\perp)^T (\mathbf D\mathbf F)^TB^{\pi^\perp\pi^\perp}(\pi^\perp)^T (\mathbf D\mathbf F)^T \tilde M^l \mathbf D\mathbf F\pi^\perp
    +(\pi^\perp)^T (\mathbf D\mathbf F)^T B^{\pi^\perp\pi}\pi^T(\mathbf D\Phi)^T\CoeffSymmPre{\bar{\tt a}}^{l} \mathbf D\Phi \pi.
\end{equation}
The matrix $\CoeffSymmPre{\bar{\tt a}}^{l}$ is symmetric. Hence, we can symmetrise the matrix in \eqref{eq:sdf93fa} by adding a suitable $\pi^T$-$\pi^\perp$-part. The ``correct'' matrix to add can again be constructed from a special case of \eqref{eq:consistencyconservation.2}: 
\begin{equation}
     \pi^T(\mathbf D\Phi)^T\CoeffSymmPre{\bar{\tt a}}^{l} \mathbf D\Phi \pi (B^{\pi^\perp\pi})^T \mathbf D\mathbf F \pi^\perp \partial_l U 
    =- \pi^T(\mathbf D\Phi)^T\CoeffSymmPre{\bar{\tt a}}^{l}\mathbf F_l.
\end{equation}
In the special situation where $\ker\mathbf D\mathbf \Phi$ is one-dimensional, we have therefore constructed a fully symmetric hyperbolic PDE system. It is easy to check that this PDE system satisfies all the properties listed above.

\bibliographystyle{amsplain}
\bibliography{bibliography}
\end{document}